%% file: main.tex
\documentclass[journal,twocolumn]{IEEEtran}
\usepackage[dvipdfmx]{graphicx}
\usepackage{multirow}
\usepackage{amsmath}
\usepackage{amssymb}
\usepackage{theorem}
\usepackage{color}
\usepackage{footnote}
\DeclareMathAlphabet{\bm}{OML}{cmm}{b}{it}
\theorembodyfont{\rmfamily}
\newtheorem{theorem}{Theorem}
\newtheorem{lemma}{Lemma}

\newtheorem{corollary}{Corollary}
\newtheorem{remark}{Remark}
\newtheorem{proposition}{Proposition}
\newtheorem{assumption}{Assumption}

\newcommand{\qed}{\hfill \IEEEQED}

\newcommand{\rom}[1]{\mathrm{#1}}
\newcommand{\san}[1]{\mathsf{#1}}

\newcommand{\argmax}{\mathop{\rm argmax}\limits}

\def\Label#1{\label{#1}\ [\ \text{#1}\ ]\ }
\def\Label{\label}

\newcommand{\MH}[1]{#1}

\allowdisplaybreaks[2]

\begin{document}

\title{Uniform Random Number Generation from Markov Chains: Non-Asymptotic and Asymptotic Analyses\thanks{This paper was
presented in part at 2014 Information Theory and Applications Workshop \cite{2014ITA}.}}

\author{Masahito~Hayashi~\IEEEmembership{Senior Member,~IEEE}
\thanks{The first author is with the Graduate School of Mathematics, Nagoya University, Japan. He is also with the Centre for
Quantum Technologies, National University of Singapore, Singapore. e-mail:masahito@math.nagoya-u.ac.jp}
and Shun~Watanabe~\IEEEmembership{Member,~IEEE}       
\thanks{The second author is with 
the Department of Computer and Information Sciences, Tokyo University of Agriculture and Technology, Koganei, Tokyo, Japan.
He was with the Department of Information Science and Intelligent Systems, 
University of Tokushima,
Tokushima, Japan.
e-mail:e-mail: shunwata@cc.tuat.ac.jp}

\thanks{Manuscript received ; revised }}

\markboth{Journal of \LaTeX\ Class Files,~Vol.~6, No.~1, January~2007}%
{Shell \MakeLowercase{\textit{et al.}}: Bare Demo of IEEEtran.cls for Journals}

\maketitle
\begin{abstract}
In this paper, we derive non-asymptotic achievability and converse bounds
on the random number 
generation with/without side-information.
Our bounds are efficiently computable in the sense that the computational complexity 
does not depend on the block length. 
We also characterize the asymptotic behaviors of
the large deviation regime and the moderate deviation regime by using our bounds, which
implies that our bounds are asymptotically tight in those regimes. 
We also show the second order rates of those problems, and derive single letter forms 
of the variances characterizing the second order rates.
Further, we address \MH{the relative entropy rate and the modified mutual information rate} for these problems.
\end{abstract}


\IEEEpeerreviewmaketitle

\input{./intro}

\input{./Multi-Entropy}

\begin{savenotes}
 \begin{table*}
\begin{center}
\caption{Summary of the bounds for the uniform random number generation.}
\Label{table:summary:single-random-number}
{\renewcommand{\arraystretch}{1.4}
\begin{tabular}{|c|c|c|c|c|c|c|c|c|} \hline
  \multirow{2}{*}{Ach./Conv.}  
  & \multirow{2}{*}{Markov} & \multirow{2}{*}{Single Shot} &  \multirow{2}{*}{$\Delta$,$\overline{\Delta}$,$D$,$\overline{D}$} &  \multirow{2}{*}{Complexity} 
  & Large & Moderate & Second & \MH{RER}\\
  & & & & & Deviation & Deviation & Order & Rate \\ \hline
\multirow{3}{*}{Achievability} 
 & Theorem \ref{theorem:single-random-finite-markov-direct-1} & Lemma \ref{lemma:single-random-number-exponential-bound} & $\overline{\Delta}$ & $O(1)$ & $\checkmark^*$     & \checkmark &  &
\\ 
\cline{2-9}
 & \multicolumn{2}{|c|}{Lemma \ref{lemma:single-random-number-leftover-loose}}  &  $\overline{\Delta}$ & Tail &  & \checkmark & \checkmark & 
\\ 
\cline{2-9}
& Theorem \ref{Th11-25-3} & Theorem \ref{Th11-25-1} & $\overline{D}$ & $O(1)$ & & & & \checkmark
\\
\hline 
\multirow{4}{*}{Converse} 
 & Theorem \ref{theorem:single-random-sphere-packing-converse-finite-markov} & Theorem \ref{theorem:random-number-exponential-converse}  & $\Delta$ & $O(1)$ &  & \checkmark &  &
\\ \cline{2-9} 
 & Theorem \ref{theorem:single-random-number-strong-universal-finite-markov-converse}  & Theorem \ref{theorem:random-number-strong-universal-exponential-converse}  & $\overline{\Delta}$ & $O(1)$ &$\checkmark^*$ & \checkmark &  &
\\ 
\cline{2-9} 
 & \multicolumn{2}{|c|}{Lemma \ref{lemma:single-random-sphere-packing-converse}} &  $\Delta$ & Tail &  & \checkmark & \checkmark &
\\ 
\cline{2-9}
& Theorem \ref{Th11-25-4} & Proposition \ref{Th11-25-2} & $D$ & $O(1)$ & & & & \checkmark
\\
\hline 
\end{tabular}
}
\end{center}
\end{table*}
\end{savenotes}

\input{./Single-Random-Number}

\begin{savenotes}[b]
 \begin{table*}
\begin{center}
\caption{Summary of the bounds for uniform random number generation with side-information.}
\Label{table:summary:multi-random-number}
{\renewcommand{\arraystretch}{1.4}
\begin{tabular}{|c|c|c|c|c|c|c|c|c|} \hline
  \multirow{2}{*}{Ach./Conv.}  
  & \multirow{2}{*}{Markov} & \multirow{2}{*}{Single Shot} &  \multirow{2}{*}{$\Delta$,$\overline{\Delta}$,$D$,$\overline{D}$} &  \multirow{2}{*}{Complexity} 
  & Large & Moderate & Second & \MH{MMIR} \\
  & & & & & Deviation & Deviation & Order & Rate
\\ \hline
\multirow{4}{*}{Achievability} 
 & Theorem \ref{theorem:multi-random-finite-markov-assumption-1-direct} (Ass.~1) & (Lemma \ref{lemma:exponential-bound}) & $\overline{\Delta}$ & $O(1)$ &      & \checkmark &  &
\\ \cline{2-9}
 & Theorem \ref{theorem:multi-random-finite-markov-assumption-2-direct} (Ass.~2) & Lemma \ref{lemma:exponential-bound} & $\overline{\Delta}$ & $O(1)$ & $\checkmark^*$     & \checkmark &  &
\\ \cline{2-9}
 & \multicolumn{2}{|c|}{Lemma \ref{lemma:multi-information-spectrum}}  &  $\overline{\Delta}$ & Tail &  & \checkmark & \checkmark &
\\ \cline{2-9}
& Theorem \ref{Th11-25b-3} (Ass.~1) & Theorem \ref{Th11-25b-1} & $\overline{D}$ & $O(1)$ & & & & \checkmark
\\
\hline 
\multirow{4}{*}{Converse} 
 & Theorem \ref{theorem:multi-random-finite-markov-assumption-1-converse} (Ass.~1) & Theorem \ref{theorem:multi-random-exponential-converse}  & $\Delta$ & $O(1)$ &  & \checkmark &  &
\\ \cline{2-9} 
 & Theorem \ref{theorem:multi-random-finite-markov-assumption-2-converse} (Ass.~2) & Theorem \ref{theorem:multi-random-exponential-converse-stong-universal} & $\overline{\Delta}$ & $O(1)$ &$\checkmark^*$ & \checkmark &  &
\\ \cline{2-9} 
 & \multicolumn{2}{|c|}{Lemma \ref{lemma:multi-random-sphere-packing-converse}} &  $\Delta$ & Tail &  & \checkmark & \checkmark &
\\ \cline{2-9} 
& Theorem \ref{Th11-25b-4} (Ass.~1) & Proposition \ref{Th11-25b-2} & $D$ & $O(1)$ & & & & \checkmark
\\
\hline 
\end{tabular}
}
\end{center}
\end{table*}
\end{savenotes}

\input{./example}


\input{./Multi-Random-Number}

\input{./conclusion}

\appendices
\input{./Appendix-Preparation}
\input{./Appendix-Single-Random-Number}
\input{./Appendix-Multi-Random-Number}

\section*{Acknowledgment}

The authors would like to thank Prof.~Junji Shikata for informing the authors of
the alternative security criteria in Remark \ref{remark:alternative-security}.
The authors also would like to thank Dr.~Marco Tomamichel and Dr.~Mario Berta for valuable comments.
HM is partially supported by a MEXT Grant-in-Aid for
Scientific Research (A) No. 23246071. He is partially supported by the National Institute of Information and Communication Technology (NICT), Japan. 
SW is partially supported by JSPS Postdoctoral Fellowships for Research Abroad.
The Centre for Quantum Technologies is funded by the Singapore Ministry of Education and the
National Research Foundation as part of the Research Centres of Excellence programme.

\bibliographystyle{./IEEEtran}
\bibliography{./reference.bib}


\begin{IEEEbiographynophoto}{Masahito Hayashi}(M'06--SM'13) was born in Japan in 1971.
He received the B.S. degree from the Faculty of Sciences in Kyoto 
University, Japan, in 1994 and the M.S. and Ph.D. degrees in Mathematics from 
Kyoto University, Japan, in 1996 and 1999, respectively.

He worked in Kyoto University as a Research Fellow of the Japan Society of the 
Promotion of Science (JSPS) from 1998 to 2000,
and worked in the Laboratory for Mathematical Neuroscience, 
Brain Science Institute, RIKEN from 2000 to 2003,
and worked in ERATO Quantum Computation and Information Project, 
Japan Science and Technology Agency (JST) as the Research Head from 2000 to 2006.
He also worked in the Superrobust Computation Project Information Science and Technology Strategic Core (21st Century COE by MEXT) Graduate School of Information Science and Technology, The University of Tokyo as Adjunct Associate Professor from 2004 to 2007. 
In 2006, he published the book ``Quantum Information: An Introduction'' from Springer. 
He worked in the Graduate School of Information Sciences, Tohoku University as Associate Professor from 2007 to 2012. 
In 2012, he joined the Graduate School of Mathematics, Nagoya University as Professor. 
He also worked in Centre for Quantum Technologies, National University of Singapore as Visiting Research Associate Professor from 2009 to 2012 
and as Visiting Research Professor from 2012 to now. 
In 2011, he received Information Theory Society Paper Award (2011) for Information-Spectrum Approach to Second-Order Coding Rate in Channel Coding.
In 2016, he received the Japan Academy Medal from the Japan Academy
and the JSPS Prize from Japan Society for the Promotion of Science.

He is on the Editorial Board of {\it International Journal of Quantum Information}
and {\it International Journal On Advances in Security}. 
His research interests include classical and quantum information theory and classical and quantum statistical inference.  
\end{IEEEbiographynophoto}

\begin{IEEEbiographynophoto}{Shun Watanabe}
(M'09) received the B.E.,
M.E., and Ph.D.\ degrees from the Tokyo Institute of Technology
in 2005, 2007, and 2009, respectively. During April 2009 to February 2015, he was an
Assistant Professor in the Department of Information
Science and Intelligent Systems at  the University of Tokushima.
During April 2013 to March 2015, he was a visiting Assistant Professor
in the Institute for Systems Research at the University of Maryland, College Park.
Since February 2015, he has been an Associate Professor in the Department of
Computer and Information Sciences at Tokyo University of Agriculture and Technology.
His current research interests are in the areas of
information theory, quantum information theory,
cryptography, and computer science.
\end{IEEEbiographynophoto}

\end{document}

%% file: intro.tex
\section{Introduction}

\subsection{Uniform random number generation (URNG)}\Label{s1-1}
Uniform random number generation is one of important tasks 
for information theory as well as secure communication.
When a non-uniform random number is generated subject to independent and identical distribution
and the source distribution is known to $P_X$,
we can convert it to the uniform random number,
whose optimal conversion rate is known to be the entropy $H(P_X)$ \cite{elias:72}.
Vembu and Verd\'{u} \cite{vembu:95} extended this problem to the general information source.
Applying their result to the Markovian source, 
we find that the optimal conversion rate is the entropy rate. 

On the other hand, many researchers in information theory are attracted by non-asymptotic analysis recently \cite{polyanskiy:10,hayashi:09,hayashi:08}. 
Since all of realistic situations are non-asymptotic,
it is strongly desired to evaluate the performance of a protocol in the non-asymptotic setting.
In the case of uniform random number generation,
we need to consider two issues:
\begin{description}
\item[A1)] How to {\em quantitatively} guarantee the security for finite block length $n$.
As the criterion, we employ the variational distance criterion 
because it is universal composable\cite{renner:05c}.

\item[A2)]  How to implement the extracting method efficiently.
\end{description}

Fortunately, the latter problem has been solved by employing universal$_2$ hash functions, which can be constructed by combination of Toeplitz matrix and the identity matrix \cite{hayashi:10}.
This construction has small amount of complexity and was implemented in a real demonstration \cite{asai:11,H-T}.
Recently, the paper \cite{tsurumaru:11} proposed a new class of hash functions, $\varepsilon$-almost dual universal hash functions,
and the paper \cite{H-T} proposed more efficient hash functions belonging to this new class.
Hence, it is needed to solve the first problem.

So far, with a huge size $n$,
quantitative evaluation of the security has been done only for the i.i.d. source \cite{hayashi:10,hayashi:10b}.
However, the source is not necessarily i.i.d. in the real world, and it is necessary to develop 
a technique to evaluate the security for non i.i.d. source. As a first step of this direction of research,
we consider the Markov source in this paper. In the following, we explain difficulties to extend the 
existing results for the i.i.d. source to the Markov source. 

Although it is not stated explicitly in any literatures, we believe that there are two important criteria for
non-asymptotic bounds:
\begin{description}
\item[B1)] Computational complexity, and
 
\item[B2)] Asymptotic optimality.
\end{description}

Let us first consider the first criterion, i.e., the computational complexity.
For example, Han \cite{han:book} introduced lower and upper bounds for 
the variational distance criterion by using the inf-spectral entropy,
which are called the inf-spectral entropy bounds. 
For i.i.d. sources, these bounds can be computed by numerical calculation packages.
However, there is no known method to efficiently compute these bounds for Markov sources.
Consequently, there is no bound that is efficiently computable for the Markov chain so far.
The first purpose of this paper is to derive non-asymptotic bounds that are efficiently computable.  

Next, let us consider the second criterion, i.e., asymptotic optimality.
So far, three kinds of asymptotic regimes have been studied in the information theory: 
\begin{description}
\item[B2-1)]
The large deviation regime in which the error probability $\varepsilon$ asymptotically behaves like $e^{- n r}$ for some $r > 0$ \cite{gallager:63},

\item[B2-2)]
The moderate deviation regime in which $\varepsilon$ asymptotically behaves like $e^{- n^{1 - 2t} r}$ for some $r > 0$ and $t \in (0,1/2)$ \cite{altug:10,dakehe:09,kuzuoka:12}, and

\item[B2-3)]
The second order regime in which $\varepsilon$ is a constant
\cite{Strassen:62,polyanskiy:10,hayashi:09,hayashi:08,altug:10,dakehe:09,tan:12b}.
\end{description}
We shall claim that a good non-asymptotic bound should be asymptotically optimal 
\MH{in} at least one of the above mentioned three regimes.

Further, when the generation rate is too large, 
the variational distance is close to $1$.
In this case, we cannot measure how far from the uniform random number
the generated random number is. 
Hence, we employ 
\MH{the relative entropy rate (RER)}.

\subsection{Secure uniform random number generation (SURNG)}
When the initial random number $X$ is partially leaked to the third party $Y$,
to guarantee the security, we need to convert the random number to the uniform random number that has almost no correlation with the third party.
When a non-uniform random number is generated subject to independent and identical distribution of the joint distribution is known to $P_{X,Y}$,
we can convert it to the uniform random number,
whose optimal conversion rate is known to be the conditional entropy $H(X|Y)$ \cite{maurer:93,ahlswede:93}.

Bennett et al. \cite{bennett:88, bennett:95} and H\r{a}stad et al. \cite{hastad:99} proposed to use universal$_2$ hash functions for this purpose,
and derived two universal hashing lemma, which provides an upper bound for leaked information based on R\'{e}nyi entropy of order $2$.
The paper \cite{tsurumaru:11} proposed to use
$\varepsilon$-almost dual universal hash functions \cite{tsurumaru:11} that includes the hash functions by \cite{H-T}.
Hence, the problem A2) has been solved by employing universal$_2$ hash functions.

Therefore, the remaining problem is the problem A1), i.e., 
to quantitatively guarantee the security for finite block length $n$
under these hash functions.
For the security criterion, we employ the variational distance between 
the true distribution and the ideal distribution
because it satisfies the universal composable property \cite{renner:05c}.
To achieve the rate $H(X|Y)$ via two universal hashing lemma, 
Renner \cite{renner:05b} attached the smoothing to min entropy\footnote{Bennett et al. \cite{bennett:95} also employed a similar idea without use of the terminology of smoothing, and derived the conversion rate $H(X|Y)$.},
which is a lower bound on the above conditional R\'{e}nyi entropy of order $2$\footnote{In \cite{renner:05b}, Renner also showed a
quantum extension of the two universal hashing lemma.}.
That is, he proposed to maximize the min-entropy among the sub-distributions whose 
variational distance to the true distribution is less than a given threshold.
Using Renner's method, the paper \cite{hayashi:10b} derived a lower bound of the exponential decreasing rate. 
Tomamichel and Hayashi \cite{tomamichel:12} derived an upper bound of the universal composable quantity of extracted key
with a finite block-length $n$ by combining the Renner's method and the method of information spectrum by Han.
Further, Watanabe and Hayashi \cite{watanabe:13c} compared two approaches: the combination of 
the Renner's method and the method of information spectrum\footnote{The approach to derive
a bound in \cite{watanabe:13c} is almost the same as that in \cite{tomamichel:12}, but it should be noted that
the security criterion in \cite{watanabe:13c} is based on the variational distance while that in \cite{tomamichel:12} is based 
on the purified distance.}, and the exponential bounding approach of \cite{hayashi:10b}.
Further, the paper \cite{hayashi:13} showed that similar evaluations are possible even for 
$\varepsilon$-almost dual universal hash functions \cite{tsurumaru:11}.

For convenience, let us call the bound derived by the former approach the {\em inf-spectral entropy bound}, and
the bound derived by the latter approach the {\em exponential bound}.
It turned out that the exponential bound is tighter than the inf-spectral entropy bound when the required security level $\varepsilon$ is rather small.
A bound that interpolate both \MH{approaches} was also derived in \cite{watanabe:13c}, which we called the {\em hybrid bound}.

Similar to uniform random number generation,
for i.i.d. sources, the inf-spectral entropy bound 
and the hybrid bound can be computed by numerical calculation packages.
However, there is no known method to efficiently compute these bounds for Markov sources.
The computational complexity of the exponential bound is $O(1)$ since the exponential bound is 
described by using the Gallager function, which is an additive quantity.
However, this is not the case for Markov sources.
Consequently, there is no bound that is efficiently computable for the Markov chain so far.
Further, the first order results for Markov sources
have not been revealed as long as
the authors know, and they are clarified in this paper.

Further, when the generation key rate is too large, 
the variational distance is close to $1$.
In this case, we cannot measure how far from the secure uniform random number
the generated random number is. 
Hence, we employ the relative entropy between the generated random number and the ideal random number, which was introduced by Csisz\'{a}r-Narayan \cite{csiszar:04} \MH{and is called the modified mutual information rate.} 
Indeed, when we surpass axiomatic conditions, 
the leaked information measure must be this quantity \cite{hayashi:13}.

\subsection{Main Contribution for Non-Asymptotic Analysis}
\MH{Although there are several studies for finite-length analysis for URNG and SURNG,
they did not discuss the Markovian chain.
Indeed, while they derived several single-shot bounds,
these bounds cannot be directly applied to the Markovian chain,
because the bounds obtained by such applications are not computable at least in the the Markovian chain.
Hence, we need to derive new finite-length bounds for the Markovian chain
by modifying existing single-shot bounds.
For this purpose, we adopt the structure} similar to the paper \cite{HW14}, which addresses the source coding with Markov chain
because this paper employs the common structure between the uniform random number generation and the source coding.
Hence, the obtained results are also quite similar to those of the paper \cite{HW14}.
To derive non-asymptotic achievability bounds on the problems,
we basically use the exponential type bounds
for the single shot setting. 
When there is no information leakage, 
those exponential type bounds are described by the R\'enyi entropy.
Thus, we need to evaluate R\'enyi entropy for the Markov chain. For this purpose, we introduce 
R\'enyi entropy for transition matrices, which is defined irrespective of initial distributions 
(cf.~\eqref{eq:definition-lower-conditional-renyi-markov}). Then, we evaluate
the R\'enyi entropy for the Markov chain in terms of the R\'enyi entropy for the transition matrix.
From this evaluation, we can also find that the R\'enyi entropy rate for the Markov chain coincides with 
the R\'enyi entropy for the transition matrix. Note that the former is defined as the limit and the latter is single 
letter characterized. 

When a part of information is leaked to the third party,
to generate secure uniform random number,
we consider two assumptions on transition matrices 
(see Assumption \ref{assumption-Y-marginal-markov} and Assumption \ref{assumption-memory-through-Y} of Section \ref{section:preparation-multi}).
Although a computable form of the conditional entropy rate is not known in general, 
Assumption \ref{assumption-Y-marginal-markov}, which is less restrictive than Assumption \ref{assumption-memory-through-Y},
enables us to derive a computable form of the conditional entropy rate.

In the problems with side-information, exponential type bounds are described by 
conditional R\'enyi entropies. There are several definitions of conditional R\'enyi entropies
(see \cite{teixeira:12, iwamoto:13} for extensive review), and we use
the one defined in \cite{hayashi:10} and the one defined by Arimoto \cite{arimoto:75}.
We shall call the former one the {\em lower conditional R\'enyi entropy} (cf.~\eqref{eq:lower-conditional-renyi})
and the latter one the {\em upper conditional R\'enyi entropy} (cf.~\eqref{eq:upper-conditional-renyi}). 
To derive non-asymptotic bounds, we need to evaluate these information measures for
the Markov chain. For this purpose, under Assumption \ref{assumption-Y-marginal-markov}, we introduce  
the lower conditional R\'enyi entropy for transition matrices (cf.~\eqref{eq:definition-lower-conditional-renyi-markov}). 
Then, we evaluate
the lower conditional \MH{R\'enyi} entropy for the Markov chain in terms of its transition matrix counterpart. 
This evaluation gives non-asymptotic bounds for 
secure uniform random number generation 
under Assumption \ref{assumption-Y-marginal-markov}.
Under more restrictive assumption, i.e., Assumption \ref{assumption-memory-through-Y}, we also 
introduce the upper conditional R\'enyi entropy for a transition matrix (cf.~\eqref{eq:definition-upper-conditional-renyi-markov}). 
Then, we evaluate 
the upper R\'enyi entropy for the Markov chain in terms of its transition matrix counterpart. 
This evaluation gives non-asymptotic bounds that are tighter than those obtained under 
Assumption \ref{assumption-Y-marginal-markov}.

We also derive converse bounds for every problem by using the change of measure argument 
developed by the authors in the accompanying paper on information geometry \cite{hayashi-watanabe:13, hayashi-watanabe:13b}.
When there is no information leakage, 
the converse bounds are described
by the R\'enyi entropy for transition matrices. 
When a part of information is leaked to the third party,
we further introduce 
two-parameter conditional R\'enyi entropy and its transition matrix counterpart
(cf.~\eqref{eq:two-parameter-conditional-renyi} and \eqref{eq:definition-two-parameter-renyi-markov}).
This novel information measure includes the lower conditional R\'enyi entropy and the upper conditional
R\'enyi entropy as special cases.  

\MH{In the problem of SURNG, instead of the RER,
we employ the modified mutual information rate (MMIR),
which was introduced by Csisz\'ar and Narayan \cite{csiszar:04}
and whose axiomatic characterization was obtained in the paper \cite{hayashi:13}.
When the uniformity is guaranteed, this quantity
is given by the equivocation rate introduced by Wyner \cite{wyner:75}.}
When there is no information leakage, our lower and upper bounds are given 
by using 
the R\'enyi entropy for the Markov chain in terms of its transition matrix counterpart.
When there exists information leakage, 
our lower and upper bounds are given 
by using 
the lower conditional R\'enyi entropy for the Markov chain in terms of its transition matrix counterpart
under Assumption \ref{assumption-Y-marginal-markov}.

Here, we would like to remark on terminologies. There are a few ways to express 
exponential type bounds. In statistics or the large deviation theory, we usually use 
the cumulant generating function (CGF) to describe exponents. In information theory,
we use the Gallager function or the R\'enyi entropies. Although these three terminologies are 
essentially the same and are related by change of variables, 
the CGF and the Gallager function are convenient for some calculations since they have good properties such as convexity.
However, they are merely mathematical functions. On the other hand, the R\'enyi entropies
are information measures including Shannon's information measures as special cases. 
Thus, the R\'enyi entropies are intuitively familiar in the field of information theory.
The R\'enyi entropies also have an advantage that two types of bounds
(eg.~\eqref{11-15-1} and \eqref{11-15-2}) 
can be expressed in a unified manner. 
For these reasons, we state our main results in terms of
the R\'enyi entropies while we use the CGF and the Gallager function in the proofs. 
For readers' convenience, the relation between the R\'enyi entropies and corresponding CGFs are 
summarized in Appendix \ref{Appendix:preparation}.

\MH{Overall, 
we summarize the contributions for non-asymptotic analysis
in comparison to existing results as follows.}
\begin{description}
\item[(1)]\MH{{\it Finite-length bound:} 
For URNG and SURNG,
we derive finite-length bounds satisfying the conditions B1) and B2) 
for Markovian chain. 
Theorems in Subsections \ref{subsection:single-random-finite-markov} and \ref{subsection:multi-random-finite-markov}
are classified to this type of results.
All existing finite-length bounds with computable form
are obtained with i.i.d. setting.
Indeed, 
several single-shot bounds were obtained in a more general form.
However, their computabilities have not been discussed in the Markovian case.
At least, many of them, (e.g, Lemmas 16, 17, 18, 22, 23, 25, and 28) 
are not given in a computable form in the Markovian case.}

\item[(2)]\MH{{\it Single-shot bound:}
In this paper, we employ several existing single-shot bounds.
However, many of them cannot be given in a useful form.
These bounds cannot be easily calculated at least in the Markovian case.
To apply them to the Markovian case, 
we loosen these bounds.
Lemmas \ref{lemma:single-random-sphere-packing-converse}, \ref{lemma:strong-universal-bound-tail-probability}, \ref{lemma:multi-random-sphere-packing-converse} and \ref{lemma:strong-universal-bound-multi-tail} fall in this case.
Since these bounds have a much simpler form than existing bounds,
they might be applied to other cases.
This discussion for the simplification 
is quite different from the case of source coding \cite{HW14}.
That is, this part
has the most serious technical hardness compared to the paper \cite{HW14}
because the discussion in this paper is specialized to random number generation.}
\end{description}

\begin{savenotes}
 \begin{table*}[!t]
\begin{center}
\caption{Summary of Asymptotic Results and Non-Asymptotic Bounds to Derive Asymptotic Results}
\label{table:summary:Asymptotic-results}
\begin{tabular}{|c|c|c|c|c|c|} \hline
 Problem & First Order & Large Deviation & Moderate Deviation & Second Order & 
\MH{RER/MMIR} \\ \hline
URNG & Solved & $\mbox{Solved}^*$ (U2), $O(1)$ & Solved, $O(1)$ 
& Solved, Tail & Solved, $O(1)$ \\ \hline
\multirow{2}{*}{SURNG} & \multirow{2}{*}{Solved (Ass.~1)} & 
$\mbox{Solved}^*$  (Ass.~2, U2), 
 & Solved (Ass.~1), & Solved (Ass.~1),  & Solved (Ass.~1), \\
&& $O(1)$ & $O(1)$ & Tail & $O(1)$ \\ 
 \hline
\end{tabular}
\end{center}
URNG is the uniform random number generation without information leakage.
SURNG is the secure uniform random number generation when a part of information is leaked to the third party.
\end{table*}
\end{savenotes}

\subsection{Main Contribution for Asymptotic Analysis}
\MH{Among authors' knowledge,
there is no existing study for 
the asymptotic analysis with the Markovian chain 
with respect to URNG and SURNG except for the following.
When the general sequence of single information sources,
the asymptotic rate of URNG is characterized by 
Vembu and Verd\'{u} \cite{vembu:95} and Han \cite{han:book}.
Since the asymptotic entropy rate of Markovian chain is known,
we can calculate the asymptotic rate of URNG for the Markovian chain.
However, further study with respect to URNG and SURNG
has not been discussed for
the Markovian chain nor 
the general sequence of information sources.}

We can easily see that 
these non-asymptotic bounds yields the asymptotic optimal random number generation rate while
the case with information leakage requires Assumption \ref{assumption-Y-marginal-markov}.
For asymptotic analyses of the large deviation and the moderate deviation regimes,
we derive the characterizations\footnote{For the large deviation regime, we only derive the characterizations up to the critical rates.} 
by using our non-asymptotic
achievability and converse bounds, which implies that our non-asymptotic bounds are tight in 
the large deviation regime and the moderate deviation regime.

We also derive the second order rate. 
It is also clarified that 
the reciprocal coefficient of the moderate deviation regime and the variance of the second order regime coincide.
Furthermore, a single letter form of the variance is clarified\footnote{An alternative way to derive a single letter characterization of the variance 
for the Markov chain was shown in \cite[Lemma 20]{tomamichel:13}. It should be also noted that a single letter characterization
can be derived by using the fundamental matrix \cite{kemeny-snell-book}. 
}.

The asymptotic results and the non-asymptotic results are summarized in Table \ref{table:summary:Asymptotic-results}.
As a part of the non-asymptotic results, the table focuses on 
the computational complexities of the non-asymptotic bounds.
"$\mbox{Solved}^*$" indicates that those problems are solved up to the critical rates.
"Ass.~1" and "Ass.~2" indicate that those problems are solved under 
Assumption \ref{assumption-Y-marginal-markov} or Assumption \ref{assumption-memory-through-Y}.
"U2" indicates that the converse results are obtained only for the worst case of the universal two hash family
(see \eqref{eq:single-random-number-worst-u2} and \eqref{eq:multi-random-number-worst-u2}).
"$O(1)$" indicates that both the achievability part and the converse part of those asymptotic results are derived from our 
non-asymptotic achievability bounds and converse bounds whose
computational complexities are $O(1)$. "Tail" indicates that both the achievability part and the converse part of those asymptotic results are derived from the
information-spectrum type achievability bounds and converse bounds whose computational complexities depend on the computational complexities
of tail probabilities.

Exact computations of tail probabilities are difficult in general though it may be feasible for a simple 
case such as an i.i.d.~case. One way to approximately compute tail probabilities is to use the Berry-Ess\'een
theorem \cite[Theorem 16.5.1]{feller:book} or its variant \cite{tikhomirov:80}. 
This direction of research is still continuing \cite{kontoyiannis:03,herve:12}, and an evaluation of
the constant was done in \cite{herve:12} though it is not clear how much tight it is. 
If we can derive a tight Berry-Ess\'een type bound for the Markov chain, we  can derive a non-asymptotic bound
that is asymptotically tight in the second order regime.
However, the approximation errors
of Berry-Ess\'een type bounds converge only in the order of $1/\sqrt{n}$, and cannot be applied when $\varepsilon$ is rather small. 
Even in the cases such that exact computations of tail probabilities are possible, the information-spectrum type bounds
are looser than the exponential type bounds when $\varepsilon$ is rather small, and we need to use appropriate 
bounds depending on the size of $\varepsilon$. In fact, this observation was explicitly clarified in \cite{watanabe:13c}
for the random number generation with side-information. Consequently, we believe that our exponential type non-asymptotic 
bounds are very useful.  

\MH{Further, we derive the asymptotic leaked information rate.
When there is no information leakage, 
we discuss the RER, which is asymptotically given by the entropy rate.
When there exists information leakage, 
we discuss the MMIR, which is asymptotically given by the conditional entropy rate
under Assumption \ref{assumption-Y-marginal-markov}.}

\MH{Overall, 
we summarize the contributions for asymptotic analysis
in comparison to existing results as follows.}
\begin{description}
\item[(1)]\MH{{\it New bounds for Markovian case:}
For URNG and SURNG,
we derive the optimal asymptotic performances in Subsections \ref{subsection:single-random-large-deviation}, 
\ref{subsection:single-random-mdp},
\ref{subsection:single-random-second-order},
\ref{Th11-25-5},
\ref{Equivocation Rate},
\ref{subsection:multi-random-large-deviation},
\ref{subsection:multi-random-mdp},
\ref{subsection:multi-random-second-order},
and \ref{Equivocation Rate-2} 
under the four regimes, 
the large deviation regimes, 
the moderate deviation regimes, 
the second order regimes,
and the asymptotic relative entropy rate regime
(the asymptotic modified mutual information rate regime)
for Markovian chain (with suitable conditions for SURNG).
Except for 
the information spectrum approach,
all existing asymptotic analyses 
with these three regimes
assume the i.i.d. source.
Further, analyses with the information spectrum approach
derived only the general formulas, which did not derive
any computable asymptotic bounds for these three regimes 
for the Markovian chain.}

\item[(2)]\MH{{\it New bound even for i.i.d. case:}
Among the above asymptotic results,
Theorem \ref{th30} is novel even for the i.i.d. case.
This theorem gives the converse bound for large deviation for SURNG.}

\end{description}

\subsection{Two criteria}\Label{two-cri}
\MH{In this paper, to consider a practical issue,
we employ two criteria.
In the channel coding, such a practical issue is discussed as a coding theory 
in a form separate from the fundamental issue.
However, in the random number generation case,
we can discuss the performance of hash functions with a small construction complexity
in the same way as the fundamental issue.
Such a practical issue is also the target of this paper.
Usually,
when we discuss a fundamental aspect of the topic of information theory,
we focus only on the minimum leaked information among all of hash function,
which is denoted by $\Delta(M)$ in this paper, 
whose precise definition will be given in Subsections \ref{subsection:single-random-problem-formulation} and \ref{subsection:multi-random-problem-formulation}.
However, when we take account into the complexity of construction of protocol, 
we need to restrict hash functions into 
hash functions with a small construction complexity.
Hence, it is desired to minimize
the leaked information among 
a class of hash functions with small calculation complexity for its construction.
In this paper we focus on the family of two-universal hash functions, named by 
the two-universal hash family ${\cal F}$
because this family contains a hash function with a small construction complexity.
However,  this paper focuses on 
the worst leaked information $\overline{\Delta}(M)$ among 
the two-universal hash family ${\cal F}$,
which is more important from a practical view point 
than the best case due to the following two reasons.}

\begin{description}
\item[(1)]
\MH{Usually, the optimal hash function depends on the source distribution.
However, it is not easy to perfectly identify the source distribution.
In such a case, 
instead of the optimal hash function, we need to choose a hash function that universally works well.
If we apply a two-universal hash function,
its leaked information is always better than the worst leaked information  $\overline{\Delta}(M)$.
Hence, if the quantity $\overline{\Delta}(M)$ is sufficiently close to the optimal case ${\Delta}(M)$,
we can say that 
any two-universal hash function universally works well.}

\item[(2)]
\MH{Although the two-universal hash family ${\cal F}$
contains a hash function with a small calculation complexity for its construction,
any two-universal hash function does not necessarily 
have a small calculation complexity.
If the quantity $\overline{\Delta}(M)$ is sufficiently close to the optimal case ${\Delta}(M)$,
we can take the priority to minimize the construction complexity among 
the two-universal hash family ${\cal F}$
over the optimization of the leaked information.} 
\end{description}

\MH{In this paper, we show that 
the worst leaked information $\overline{\Delta}(M)$ 
is close to the minimum leaked information ${\Delta}(M)$
in the moderate deviation and the second order.
These results guarantee that 
any two-universal hash function has a sufficiently good performance.
That is, 
they allow us to employ any two-universal hash function
to achieve these asymptotic optimal performances.
These results amplify our choice of hash function 
to achieve the asymptotically optimality.}


\subsection{Organization of Paper and Notations}
As preparation, we explain information measures for single-shot setting
in Subsection \ref{section:multi-terminal-single-shot}.
Then, 
we address conditional R\'{e}nyi entropies for transition matrix
in Subsection \ref{subsection:multi-terminal-measures-markov}, 
and discuss the relation between these information measures and 
Markov chain in Subsection \ref{subsection-multi-terminal-information-measures-markov}.
These information measures and their properties will
be used in the latter sections.
These contents were obtained in the paper \cite{HW14},
and their proofs are available in the paper \cite{HW14}.
However, the paper \cite{HW14} did not address the conditional min entropy, which corresponds to the order parameter $\infty$.
So, in Subsections \ref{s2d} and \ref{s2e},
we discuss the relation between 
the limit of the conditional R\'{e}nyi entropy
and the conditional min entropy,
which are new results and are shown in Appendix.

Section \ref{section:single-random-number}
addresses the uniform random number generation
without information leakage.
The obtained upper and lower bounds are numerically calculated in a typical example in this section.
Then, Section \ref{section:multi-random-number}
proceeds to 
addresses the secure uniform random number generation
with partial information leakage.
As we mentioned above, we state our main result in terms of the
R\'enyi entropies, and we use the CGFs and the Gallager function in the proofs.
In Appendix \ref{Appendix:preparation}, the relation between the R\'enyi entropies and 
corresponding CGFs are summarized. The relation between the R\'enyi entropies and the Gallager function are explained 
as necessary. 
Proofs of some technical results are also shown in the rest of appendices.

A random variable is denoted by upper case letter, and its realization is denoted by
lower case letter. The notation ${\cal P}({\cal X})$ is the set of all distribution on alphabet ${\cal X}$.
The notation $\bar{{\cal P}}({\cal X})$ is the set of all non-negative sub-normalized functions on ${\cal X}$.
$|{\cal X}|$ represent the cardinality of the set ${\cal X}$. 
The cumulative distribution function of the standard Gaussian random variable is denoted by
\begin{eqnarray}
\Phi(t ) = \int_{-\infty}^t \frac{1}{\sqrt{2\pi}} \exp\left[ - \frac{x^2}{2} \right] dx.
\end{eqnarray}
Throughout the paper, the base of the logarithm is $e$.

%% file: Multi-Entropy.tex
\section{Information Measures} \Label{section:preparation-multi}

In this section, we introduce information measures that will be used 
in Section \ref{section:single-random-number}
and Section \ref{section:multi-random-number}.
All of lemmas and theorems in this section except for Lemmas \ref{lemma:extreme-cases-up-conditional-renyi-transition} and \ref{lemma:finite-evaluation-min-entropy} and 
Theorem \ref{theorem:asymptotic-down-conditional-renyi}
were shown in \cite{HW14}.

\subsection{Information Measures for Single-Shot Setting} \Label{section:multi-terminal-single-shot}
\subsubsection{Conditional R\'{e}nyi entropy relative to a general distribution}

In this section, we introduce conditional R\'enyi entropies for the single-shot setting. 
For more detailed review of conditional R\'enyi entropies, see \cite{iwamoto:13}.
For a correlated random variable $(X,Y)$ on ${\cal X} \times {\cal Y}$ with probability distribution $P_{XY}$ and a marginal distribution
$Q_Y$ on ${\cal Y}$, we introduce the conditional R\'enyi entropy of order $1+\theta$ relative to $Q_Y$ as
\begin{eqnarray}
H_{1+\theta}(P_{XY}|Q_Y) := - \frac{1}{\theta} \log \sum_{x,y} P_{XY}(x,y)^{1+\theta} Q_Y(y)^{-\theta},
\end{eqnarray}
where $\theta \in (-1,0) \cup (0,\infty)$. The conditional R\'enyi entropy of order $0$ relative to $Q_Y$ is defined by the limit with
respect to $\theta$.
When ${\cal  Y}$ is singleton, it is nothing but the ordinary R\'enyi entropy, and it is denoted by 
$H_{1+\theta}(X) = H_{1+\theta}(P_X)$ throughout the paper. 

\subsubsection{Lower conditional R\'{e}nyi entropy}
One of important special cases of $H_{1+\theta}(P_{XY}|Q_Y)$ is the case with $Q_Y = P_Y$.
We shall call this special case the {\em lower conditional R\'enyi entropy} of order $1+\theta$ and denote\footnote{
This notation was first introduce in \cite{tomamichel:13b}.}
\begin{eqnarray} \Label{eq:lower-conditional-renyi}
H_{1+\theta}^\downarrow(X|Y) &:=& H_{1+\theta}(P_{XY}|P_Y) \\
&=& - \frac{1}{\theta} \log \sum_{x,y} P_{XY}(x,y)^{1+\theta} P_Y(y)^{-\theta}.
\end{eqnarray}
The following property holds.
\begin{lemma} \Label{lemma:property-down-conditional-renyi}
We have
\begin{eqnarray} \Label{eq:down-conditional-renyi-theta-0}
\lim_{\theta \to 0} H_{1+\theta}^\downarrow(X|Y) = H(X|Y)
\end{eqnarray}
and
\begin{eqnarray}
\san{V}(X|Y) &:=& \mathrm{Var}\left[ \log \frac{1}{P_{X|Y}(X|Y)} \right] \\
&=& \lim_{\theta \to 0} \frac{2\left[ H(X|Y) - H_{1+\theta}^\downarrow(X|Y) \right]}{\theta} \Label{eq:multi-single-shot-variance-1}.
\end{eqnarray}
\end{lemma}

\subsubsection{Upper conditional R\'{e}nyi entropy}
The other important special cases of $H_{1+\theta}(P_{XY}|Q_Y)$ is the measure maximized over $Q_Y$.
We shall call this special case the {\em upper conditional R\'enyi entropy} of order $1+\theta$ and 
denote\footnote{For $-1 < \theta < 0$, \eqref{eq:optimal-choice-upper-conditional} can be proved by using the H\"older inequality,
and, for $0 < \theta$,  \eqref{eq:optimal-choice-upper-conditional} can be proved by using the reverse H\"older 
inequality \cite[Lemma 8]{hayashi:12d}.}
\begin{eqnarray} \Label{eq:upper-conditional-renyi}
&& H_{1+\theta}^\uparrow(X|Y) \nonumber \\
&:=& \max_{Q_Y \in {\cal P}({\cal Y})} H_{1+\theta}(P_{XY}|Q_Y) \\
&=& H_{1+\theta}(P_{XY}|P_Y^{(1+\theta)}) \Label{eq:optimal-choice-upper-conditional} \\
&=& - \frac{1+\theta}{\theta} \log \sum_y P_Y(y) \left[ \sum_x P_{X|Y}(x|y)^{1+\theta} \right]^{\frac{1}{1+\theta}},
\Label{11-14-6}
\end{eqnarray}
where the expression \eqref{11-14-6} is the same as Arimoto's proposal for the conditional R\'{e}nyi entropy \cite{arimoto:75} and 
\begin{eqnarray} \Label{eq:single-shot-optimal-conditioning-distribution}
P_Y^{(1+\theta)}(y) := 
\frac{\left[ \sum_x P_{XY}(x,y)^{1+\theta} \right]^{\frac{1}{1+\theta}}}{\sum_{y^\prime} \left[ \sum_x P_{XY}(x,y^\prime)^{1+\theta} \right]^{\frac{1}{1+\theta}}}. 
\end{eqnarray}

For this measure, we also have properties similar to Lemma \ref{lemma:property-down-conditional-renyi}.
\begin{lemma}[{\cite{HW14,MDSFT,hayashi:12d}}]
 \Label{lemma:property-upper-conditional-renyi-single-shot}
We have
\begin{eqnarray} \Label{eq:up-conditional-renyi-theta-0}
\lim_{\theta \to 0} H_{1+\theta}^\uparrow(X|Y) = H(X|Y)
\end{eqnarray}
and 
\begin{eqnarray}
\lim_{\theta \to 0} \frac{2\left[ H(X|Y) - H_{1+\theta}^\uparrow(X|Y) \right]}{\theta} 
= \san{V}(X|Y).
\Label{eq:multi-single-shot-variance-2}
\end{eqnarray}
\end{lemma}

\subsubsection{Properties of conditional R\'{e}nyi entropies}
When we derive converse bounds, we need to consider the case such that 
the order of the R\'enyi entropy and the order of conditioning distribution defined in \eqref{eq:single-shot-optimal-conditioning-distribution} are different.
For this purpose, we introduce two-parameter conditional R\'enyi entropy:
\begin{eqnarray} \Label{eq:two-parameter-conditional-renyi}
\lefteqn{ H_{1+\theta,1+\theta^\prime}(X|Y) } \\
&:=& H_{1+\theta}(P_{XY}|P_Y^{(1+\theta^\prime)}) \\
&=& - \frac{1}{\theta} \log \sum_y P_Y(y) 
\left[ \sum_x P_{X|Y}(x|y)^{1+\theta} \right] \
\nonumber \\
&& \hspace{2ex} \cdot \left[\sum_x P_{X|Y}(x|y)^{1+\theta^\prime} \right]^{\frac{\theta}{1+\theta^\prime}}
 + \frac{\theta^\prime}{1+\theta^\prime} H_{1+\theta^\prime}^\uparrow(X|Y).
\nonumber 
\end{eqnarray}

The measures defined above has the following properties:
\begin{lemma}[{\cite{HW14,MDSFT,hayashi:12d}}] \Label{lemma:multi-terminal-single-shot-property}
$\phantom{aaa}$ 
\begin{enumerate}
\item \Label{item:multi-terminal-single-shot-property-1}
For fixed $Q_Y$, $\theta H_{1+\theta}(P_{XY}|Q_Y)$ is a concave function of $\theta$,
and it is strict concave iff. $\rom{Var}\left[ \log \frac{Q_Y(Y)}{P_{XY}(X,Y)} \right] > 0$.

\item \Label{item:multi-terminal-single-shot-property-1-b}
For fixed $Q_Y$, $H_{1+\theta}(P_{XY}|Q_Y)$ is a monotonically decreasing\footnote{Technically, $H_{1+\theta}(P_{XY}|Q_Y)$ is always non-increasing and it is monotonically decreasing iff. strict concavity holds in Statement \ref{item:multi-terminal-single-shot-property-1}. Similar remarks are also applied for other information measures throughout the paper.} function of $\theta$.

\item \Label{item:multi-terminal-single-shot-property-2}
The function $\theta H_{1+\theta}^\downarrow(X|Y)$ is a concave function of $\theta$, and 
it is strict concave iff. $\san{V}(X|Y) > 0$.

\item \Label{item:multi-terminal-single-shot-property-2-b}
$H_{1+\theta}^\downarrow(X|Y)$ is a monotonically decreasing function of $\theta$,
\MH{and it is strictly monotonically decreasing iff. $\san{V}(X|Y) > 0$.}

\item \Label{item:multi-terminal-single-shot-property-3}
The function $\theta H_{1+\theta}^\uparrow(X|Y)$ is a concave function of $\theta$, and it is strict concave 
iff. $\san{V}(X|Y) > 0$.

\item \Label{item:multi-terminal-single-shot-property-3-b}
$H_{1+\theta}^\uparrow(X|Y)$ is a monotonically decreasing function of $\theta$,
\MH{and it is strictly monotonically decreasing iff. $\san{V}(X|Y) > 0$.}

\item \Label{item:multi-terminal-single-shot-property-4}
For every $\theta \in (-1,0) \cup (0,\infty)$, we have $H_{1+\theta}^\downarrow(X|Y) \le H_{1+\theta}^\uparrow(X|Y)$.

\item \Label{item:multi-terminal-single-shot-property-5}
For fixed $\theta^\prime$, the function $\theta H_{1+\theta,1+\theta^\prime}(X|Y)$ is a concave function of $\theta$,
and it is strict concave iff. $\san{V}(X|Y) > 0$.

\item \Label{item:multi-terminal-single-shot-property-5-b}
For fixed $\theta^\prime$, $H_{1+\theta,1+\theta^\prime}(X|Y)$ is a monotonically decreasing function of $\theta$.

\item \Label{item:multi-terminal-single-shot-property-6}
We have
\begin{eqnarray}
H_{1+\theta,1}(X|Y) = H_{1+\theta}^\downarrow(X|Y).
\end{eqnarray}

\item \Label{item:multi-terminal-single-shot-property-7}
We have
\begin{eqnarray}
H_{1+\theta,1+\theta}(X|Y) = H_{1+\theta}^\uparrow(X|Y).
\end{eqnarray}

\item \Label{item:multi-terminal-single-shot-property-8}
For every $\theta \in (-1,0) \cup (0,\infty)$, $H_{1+\theta,1+\theta^\prime}(X|Y)$ is maximized at $\theta^\prime = \theta$.


\end{enumerate}
\end{lemma}



\subsubsection{Functions related to lower conditional R\'{e}nyi entropy}
Since Item \MH{5)} of Lemma \ref{lemma:multi-terminal-single-shot-property} guarantees that the function $\theta \mapsto \frac{d[\theta H_{1+\theta}^\downarrow(X|Y)]}{d\theta}$ 
is strictly monotone decreasing,
we can define the inverse functions\footnote{Throughout the paper, the notations $\theta(a)$ and
$a(R )$ are reused for several inverse functions. Although the meanings of those notations are
obvious from the context, we occasionally put superscript $\downarrow$ or $\uparrow$ to
emphasize that those inverse functions are induced from corresponding conditional R\'enyi entropies.
\MH{This definition is related to Legendre transform of the concave function 
$\theta \mapsto \theta H_{1+\theta}^\downarrow(X|Y)$.
For its detail, see \cite{HW14}.}
} 
$\theta(a) = \theta^\downarrow(a)$ and $a(R ) = a^\downarrow(R )$ by
\begin{eqnarray} \Label{eq:definition-inverse-theta-multi-one-shot-2}
\frac{d[\theta H_{1+\theta}^\downarrow(X|Y)]}{d\theta} \bigg|_{\theta = \theta(a)} = a
\end{eqnarray}
and 
\begin{eqnarray}
(1+\theta(a(R ))) a(R ) - \theta(a(R )) H_{1+\theta(a(R ))}^\downarrow(X|Y) = R,
\end{eqnarray}
for $R(\underline{a}) < R \le H_0^\downarrow(X|Y)$,
where $\underline{a} =\underline{a}^\downarrow := \lim_{\theta\to \infty} \frac{d[\theta H_{1+\theta}^\downarrow(X|Y)]}{d\theta}$.

\subsubsection{Functions related to upper conditional R\'{e}nyi entropy}
For $\theta H_{1+\theta}^\uparrow(X|Y)$,
we also introduce the inverse functions $\theta(a) = \theta^\uparrow(a)$ and $a(R ) = a^\uparrow(R )$ by
\begin{eqnarray} \Label{eq:definition-rho-inverse-Gallager-one-shot}
\frac{d\theta H_{1+\theta}^\uparrow(X|Y)}{d\theta} \bigg|_{\theta = \theta(a)} = a
\end{eqnarray}
and 
\begin{eqnarray} \Label{eq:definition-a-inverse-Gallager-one-shot}
(1+\theta(a(R ))) a(R ) - \theta(a(R )) H_{1+\theta(a(R ))}^\uparrow(X|Y) = R,
\end{eqnarray}
for $R(\underline{a}) < R \le H_0^\uparrow(X|Y)$,
where $\underline{a} =\underline{a}^\uparrow := \lim_{\theta\to \infty} \frac{d[\theta H_{1+\theta}^\uparrow(X|Y)]}{d\theta}$. 

\subsection{Information Measures for Transition Matrix}  \Label{subsection:multi-terminal-measures-markov}
\subsubsection{Conditions for transition matrices}
Let $\{ W(x,y|x^\prime,y^\prime) \}_{((x,y),(x^\prime,y^\prime)) \in ({\cal X} \times {\cal Y})^2}$ be 
an ergodic and irreducible transition matrix.
The purpose of this section is to introduce transition matrix
 counterparts of those measures in Section \ref{section:multi-terminal-single-shot}.
For this purpose, we first need to introduce some assumptions on transition matrices:
\begin{assumption}[Non-Hidden \cite{HW14,hayashi-watanabe:13,hayashi-watanabe:13b}] \Label{assumption-Y-marginal-markov}
We say that a transition matrix $W$ is {\em non-hidden} (with respect to ${\cal Y}$) if 
\begin{eqnarray}
\sum_x W(x,y|x^\prime,y^\prime) = W_Y(y|y^\prime)\Label{12-4-a}
\end{eqnarray}
for every $x^\prime \in {\cal X}$ and $y,y^\prime \in {\cal Y}$\footnote{
\MH{The reason of the name ``non-hidden'' is the following.
In general, the random variable $Y$ is subject to a hidden Markov process.
However, when the condition \eqref{12-4-a} holds,
the random variable $Y$ is subject to a Markov process.
Hence, we call the condition \eqref{12-4-a} non-hidden.}}.
\end{assumption}
\begin{assumption}[Strongly Non-Hidden] \Label{assumption-memory-through-Y}
We say that a transition matrix $W$ is {\em strongly non-hidden} (with respect to ${\cal Y}$)
if, for every $\theta \in (-1,\infty)$ and $y,y^\prime \in {\cal Y}$,
\begin{eqnarray} \Label{eq:condition-assumption-2}
W_{Y,\theta}(y|y^\prime) := \sum_x W(x,y|x^\prime,y^\prime)^{1+\theta} 
\end{eqnarray}
is well defined, i.e., the right hand side of \eqref{eq:condition-assumption-2} is
independent of $x^\prime$.
\end{assumption}
Assumption \ref{assumption-Y-marginal-markov} requires \eqref{eq:condition-assumption-2} to hold only
for $\theta = 0$, and thus Assumption \ref{assumption-memory-through-Y} implies Assumption \ref{assumption-Y-marginal-markov}.
However, Assumption \ref{assumption-memory-through-Y} is strictly stronger condition than 
Assumption \ref{assumption-Y-marginal-markov}. For example, let consider the case such that the transition matrix is a product form,
i.e., $W(x,y|x^\prime,y^\prime) = W_X(x|x^\prime) W_Y(y|y^\prime)$. In this case, Assumption \ref{assumption-Y-marginal-markov} is
 obviously satisfied. However, Assumption \ref{assumption-memory-through-Y} is not satisfied in general.

Assumption \ref{assumption-Y-marginal-markov} means that we can decompose $W(x,y|x^\prime,y^\prime)$ as 
\begin{eqnarray}
W(x,y|x^\prime,y^\prime) = W_Y(y|y^\prime) W_{X|X',Y',Y}(x|x^\prime,y^\prime,y).
\end{eqnarray}
Thus, Assumption \ref{assumption-memory-through-Y} can be rephrased as 
\begin{eqnarray}
\sum_x W_{X|X',Y',Y}(x|x^\prime,y^\prime,y)^{1+\theta} \Label{12-12-5}
\end{eqnarray}
does not depend on $x^\prime$. By taking $\theta$ sufficiently large, we find that the largest value of 
$W_{X|X',Y',Y}(x|x^\prime,y^\prime,y)$ does not depend on $x^\prime$. By repeating this argument for the second largest value
of $W_{X|X',Y',Y}(x|x^\prime,y^\prime,y)$ and so on, we eventually find that Assumption \ref{assumption-memory-through-Y} is
satisfied iff., for every $x^\prime \neq \tilde{x}^\prime$, there exists a
 \MH{permutation} $\pi$ on ${\cal X}$ such that
$W_{X|X',Y',Y}(x|x^\prime,y^\prime,y) = W_{X|X',Y',Y}(\pi(x) | \tilde{x}^\prime,y^\prime,y)$.
 
Non-trivial examples satisfying Assumption \ref{assumption-Y-marginal-markov}
and Assumption \ref{assumption-memory-through-Y}
are given in \cite{HW14}.


\subsubsection{Lower conditional R\'{e}nyi entropy $H_{1+\theta}^{\downarrow,W}(X|Y)$}
First, we introduce information measures under Assumption \ref{assumption-Y-marginal-markov}. 
In order to define a transition matrix counterpart of \eqref{eq:lower-conditional-renyi}, let us introduce the following tilted matrix:
\begin{eqnarray}
\tilde{W}_\theta(x,y|x^\prime,y^\prime) := W(x,y|x^\prime,y^\prime)^{1+\theta} W_Y(y|y^\prime)^{-\theta}.
\end{eqnarray}
\MH{Here, we should notice that the tilted matrix
$\tilde{W}_\theta$ is not normalized, i.e., is not a transition matrix.}
Let $\lambda_\theta$ be the Perron-Frobenius eigenvalue and $\tilde{P}_{\theta,XY}$ be its normalized eigenvector. 
Then, we define the lower conditional R\'enyi entropy for $W$ by
\begin{eqnarray} \Label{eq:definition-lower-conditional-renyi-markov}
H_{1+\theta}^{\downarrow,W}(X|Y) := - \frac{1}{\theta} \log \lambda_\theta,
\end{eqnarray}
where $\theta \in (-1,0) \cup (0,\infty)$. For $\theta = 0$, we define the lower conditional R\'enyi entropy for $W$ by 
\begin{eqnarray}
H_1^{\downarrow,W}(X|Y) 
:= \lim_{\theta \to 0} H_{1+\theta}^{\downarrow,W}(X|Y). \Label{eq:lower-conditional-renyi-markov-theta-0}
\end{eqnarray}
When we define the conditional entropy $H^W(X|Y)$ for $W$
by using the stationary distribution $P_{0,XY}$ as
\begin{align*}
&H^W(X|Y) \\
:=& -
\sum_{x',y'} P_{0,XY}(x',y')
\sum_{x,y} W(x,y|x',y') 
\log \frac{W(x,y|x',y')}{W_Y(y|y')},
\end{align*}
as shown below, we have
\begin{eqnarray}
H^W(X|Y) = 
H_1^{\downarrow,W}(X|Y) .
\Label{eq:lower-conditional-renyi-markov-theta-0B}
\end{eqnarray}
Taking the derivative with respect to $\theta$,
we can show \eqref{eq:lower-conditional-renyi-markov-theta-0B} as follows
\begin{align*}
& H_1^{\downarrow,W}(X|Y)
= 
\frac{d \theta H_\theta^{\downarrow,W}(X|Y)}{d \theta}\Bigr|_{\theta=0} 
= 
- \frac{d \lambda_{\theta}}{d \theta}\Bigr|_{\theta=0}\\
=&
 - \frac{d}{d \theta}
\sum_{x,y,x',y' }
\tilde{W}_\theta(x,y|x^\prime,y^\prime) 
\tilde{P}_{\theta,XY}(x^\prime,y^\prime) 
\Bigr|_{\theta=0} \\
=&
\sum_{x,y,x',y' }
 - \frac{d}{d \theta}
\tilde{W}_\theta(x,y|x^\prime,y^\prime) 
\Bigr|_{\theta=0} 
\tilde{P}_{0,XY}(x^\prime,y^\prime) \\
& - 
\sum_{x,y,x',y' }
\tilde{W}_0(x,y|x^\prime,y^\prime) 
\frac{d}{d \theta}
\tilde{P}_{\theta,XY}(x^\prime,y^\prime) 
\Bigr|_{\theta=0} \\
=&
\sum_{x,y,x',y' }
\tilde{P}_{0,XY}(x^\prime,y^\prime) 
W(x,y|x^\prime,y^\prime) 
\log \frac{W(x,y|x',y')}{W_Y(y|y')} \\
& - 
\frac{d}{d \theta}
\sum_{x,y,x',y' }
W(x,y|x^\prime,y^\prime) 
\tilde{P}_{\theta,XY}(x^\prime,y^\prime) 
\Bigr|_{\theta=0} \\
=&
H^W(X|Y),
\end{align*}
where the final equation follows from the relation $\sum_{x,y,x',y' }
W(x,y|x^\prime,y^\prime) 
\tilde{P}_{\theta,XY}(x^\prime,y^\prime) =1$.


As a counterpart of \eqref{eq:multi-single-shot-variance-1}, we also define 
\begin{eqnarray}
\san{V}^W(X|Y) := 
\lim_{\theta \to 0} \frac{2\left[ H^W(X|Y) - H_{1+\theta}^{\downarrow,W}(X|Y) \right]}{\theta}.
\Label{eq:lower-conditional-renyi-markov-theta-0-derivative}
\end{eqnarray}

\begin{remark}
When a transition matrix $W$ satisfies Assumption \ref{assumption-memory-through-Y}, $H_{1+\theta}^{\downarrow,W}(X|Y)$
can be written as 
\begin{eqnarray}
H_{1+\theta}^{\downarrow,W}(X|Y) = - \frac{1}{\theta} \log \lambda_\theta^\prime,
\end{eqnarray}
where $\lambda_\theta^\prime$ is the Perron-Frobenius eigenvalue of
$W_{Y,\theta}(y|y^\prime) W_Y(y|y^\prime)^{-\theta}$.
In fact, for the left Perron-Frobenius eigenvector 
$\hat{Q}_\theta$ of 
$W_{Y,\theta}(y|y^\prime) W_Y(y|y^\prime)^{-\theta}$, we have
\begin{eqnarray}
\sum_{x,y} \hat{Q}_\theta(y) W(x,y|x^\prime,y^\prime)^{1+\theta} W_Y(y|y^\prime)^{-\theta}
= \lambda_\theta^\prime Q_\theta(y^\prime),
\end{eqnarray}
which implies that $\lambda_\theta^\prime$ is the Perron-Frobenius eigenvalue of 
$\tilde{W}_\theta$.
Consequently, we can evaluate $H_{1+\theta}^{\downarrow,W}(X|Y)$ by calculating the Perron-Frobenius 
eigenvalue of $|{\cal Y}| \times |{\cal Y}|$ matrix instead of $|{\cal X}| |{\cal Y}| \times |{\cal X}| |{\cal Y}|$ matrix
when $W$ satisfies Assumption \ref{assumption-memory-through-Y}.
\end{remark}

\subsubsection{Upper conditional R\'{e}nyi entropy $H_{1+\theta}^{\uparrow,W}(X|Y)$}
Next, we introduce information measures under Assumption \ref{assumption-memory-through-Y}.
In order to define a transition matrix counterpart of \eqref{eq:upper-conditional-renyi}, let us introduce the 
following $|{\cal Y}| \times |{\cal Y}|$ matrix:
\begin{eqnarray}
K_\theta(y|y^\prime) := W_{Y,\theta}(y|y^\prime)^{\frac{1}{1+\theta}},
\end{eqnarray}
where $W_{Y,\theta}$ is defined by \eqref{eq:condition-assumption-2}.
Let $\kappa_\theta$ be the Perron-Frobenius eigenvalue of $K_\theta$.
Then, we define the upper conditional R\'enyi entropy for $W$ by
\begin{eqnarray} \Label{eq:definition-upper-conditional-renyi-markov}
H_{1+\theta}^{\uparrow,W}(X|Y) := - \frac{1+\theta}{\theta} \log \kappa_\theta,
\end{eqnarray}
where $\theta \in (-1,0)\cup (0,\infty)$. 

\begin{lemma}[\protect{\cite[Lemma 5]{HW14}}] \Label{lemma:properties-upper-conditional-renyi-transition-matrix}
We have
\begin{eqnarray}
\lim_{\theta \to 0} H_{1+\theta}^{\uparrow,W}(X|Y) = H^W(X|Y)
\end{eqnarray}
and
\begin{eqnarray} \Label{eq:upper-conditional-renyi-transition-variance}
\lim_{\theta \to 0} \frac{2\left[ H^W(X|Y) - H_{1+\theta}^{\uparrow,W}(X|Y) \right]}{\theta}
= \san{V}^W(X|Y).
\end{eqnarray}
\end{lemma}

Now, let us introduce a transition matrix counterpart of \eqref{eq:two-parameter-conditional-renyi}. For this purpose, we introduce 
the following $|{\cal Y}| \times |{\cal Y}|$ matrix:
\begin{eqnarray}
N_{\theta,\theta^\prime}(y|y^\prime) := W_{Y,\theta}(y|y^\prime)
W_{Y,\theta^\prime}(y|y^\prime)^{\frac{-\theta}{1+\theta^\prime}}.
\end{eqnarray}
Let $\nu_{\theta,\theta^\prime}$ be the Perron-Frobenius eigenvalue of $N_{\theta,\theta^\prime}$.
Then, we define the two-parameter conditional R\'enyi entropy by
\begin{eqnarray} \Label{eq:definition-two-parameter-renyi-markov}
H_{1+\theta,1+\theta^\prime}^W(X|Y) 
 := -\frac{1}{\theta} \log \nu_{\theta,\theta^\prime} 
 + \frac{\theta^\prime}{1+\theta^\prime} H_{1+\theta^\prime}^{\uparrow,W}(X|Y).
\end{eqnarray}

\begin{remark}
Although we defined $H_{1+\theta}^{\downarrow,W}(X|Y)$ and $H_{1+\theta}^{\uparrow,W}(X|Y)$
by \eqref{eq:definition-lower-conditional-renyi-markov} and \eqref{eq:definition-upper-conditional-renyi-markov} respectively,
we can alternatively define these measures in the same spirit as the single-shot setting 
by introducing a transition matrix counterpart of $H_{1+\theta}(P_{XY}|Q_Y)$ as follows.
For the marginal $W_Y(y|y^\prime)$ of $W(x,y|x^\prime,y^\prime)$,
let ${\cal Y}^2_{W_Y} := \{(y,y^\prime) : W(y|y^\prime) > 0\}$. 
For another transition matrix $\overline{W}_Y$ on ${\cal Y}$,
we define ${\cal Y}_{\overline{W}_Y}^2$ in a similar manner. 
For $\overline{W}_Y$ satisfying ${\cal Y}_{W_Y}^2 \subset {\cal Y}_{\overline{W}_Y}^2$, 
we define\footnote{Although we can also define $H_{1+\theta}^{W|\overline{W}_Y}(X|Y)$ even if 
${\cal Y}_{W_Y}^2 \subset {\cal Y}_{\overline{W}_Y}^2$
is not satisfied (see \cite{hayashi-watanabe:13} for the detail), for our purpose of defining $H_{1+\theta}^{\downarrow,W}(X|Y)$ and 
$H_{1+\theta}^{\uparrow,W}(X|Y)$, other cases are irrelevant.}
\begin{eqnarray}
H_{1+\theta}^{W|\overline{W}_Y}(X|Y) := - \frac{1}{\theta} \log \lambda_{\theta}^{W|\overline{W}_Y}
\end{eqnarray}
for $\theta \in (-1,0) \cup (0,\infty)$, where $\lambda_\theta^{W|\overline{W}_Y}$
is the Perron-Frobenius eigenvalue of 
\begin{eqnarray}
W(x,y|x^\prime,y^\prime)^{1+\theta} \overline{W}_Y(y|y^\prime)^{-\theta}.
\end{eqnarray}
By using this measure, we obviously have
\begin{eqnarray}
H_{1+\theta}^{\downarrow,W}(X|Y) = H_{1+\theta}^{W|W_Y}(X|Y).
\end{eqnarray}
Furthermore, under Assumption \ref{assumption-memory-through-Y}, 
the relation
\begin{eqnarray} \Label{eq:alternative-definition-of-upper-conditional-W}
H_{1+\theta}^{\uparrow,W}(X|Y) = \max_{\overline{W}_Y} H_{1+\theta}^{W|\overline{W}_Y}(X|Y)
\end{eqnarray}
holds \cite[(62)]{HW14}, where the maximum is taken over all transition matrices 
satisfying ${\cal Y}_{W_Y}^2 \subset {\cal Y}_{\overline{W}_Y}^2$.
\end{remark}

\subsubsection{Properties of conditional R\'{e}nyi entropies}
The information measures introduced in this section 
have the following properties:

\begin{lemma}[\protect{\cite[Lemma 6]{HW14}}] \Label{lemma:multi-terminal-markov-property}
$\phantom{a}$
\begin{enumerate}
\item \Label{item:multi-terminal-markov-property-1} 
The function $\theta H_{1+\theta}^{\downarrow,W}(X|Y)$ is a concave function of $\theta$, and it is strict concave iff. 
$\san{V}^W(X|Y) > 0$.

\item \Label{item:multi-terminal-markov-property-1-b}
$H_{1+\theta}^{\downarrow,W}(X|Y)$ is a monotonically decreasing function of $\theta$, 
\MH{and it is strictly monotonically decreasing iff. $\san{V}(X|Y) > 0$.}

\item \Label{item:multi-terminal-markov-property-2} 
The function $\theta H_{1+\theta}^{\uparrow,W}(X|Y)$ is a concave function of $\theta$, and it is strict concave iff.
$\san{V}^W(X|Y) > 0$.

\item \Label{item:multi-terminal-markov-property-2-b}
$H_{1+\theta}^{\uparrow,W}(X|Y)$ is a monotonically decreasing function of $\theta$,
\MH{and it is strictly monotonically decreasing iff. $\san{V}(X|Y) > 0$.}

\item \Label{item:multi-terminal-markov-property-3} 
For every $\theta \in (-1,0) \cup (0,\infty)$, we have 
$ H_{1+\theta}^{\downarrow,W}(X|Y) \le H_{1+\theta}^{\uparrow,W}(X|Y)$.

\item \Label{item:multi-terminal-markov-property-4} 
For fixed $\theta^\prime$, the function $\theta H_{1+\theta,1+\theta^\prime}^W(X|Y)$ is a concave function of $\theta$,
and it is strict concave iff. $\san{V}^W(X|Y) > 0$.

\item \Label{item:multi-terminal-markov-property-4-b} 
For fixed $\theta^\prime$, $H_{1+\theta,1+\theta^\prime}^W(X|Y)$ is a monotonically decreasing function of $\theta$.

\item \Label{item:multi-terminal-markov-property-5} 
We have
\begin{eqnarray}
H_{1+\theta,1}^W(X|Y) = H_{1+\theta}^{\downarrow,W}(X|Y). 
\end{eqnarray}

\item \Label{item:multi-terminal-markov-property-6} 
We have
\begin{eqnarray}
H_{1+\theta,1+\theta}^W(X|Y) = H_{1+\theta}^{\uparrow,W}(X|Y). 
\end{eqnarray}

\item \Label{item:multi-terminal-markov-property-7} 
For every $\theta \in (-1,0) \cup (0,\infty)$, 
$H_{1+\theta,1+\theta^\prime}^W(X|Y)$ is maximized at $\theta^\prime = \theta$, i.e.,
\begin{align}
\left.
\frac{d H_{1+\theta,1+\theta^\prime}^W(X|Y)}{d \theta^\prime}
\right|_{\theta^\prime=\theta}=0. \Label{11-27-1}
\end{align}

\end{enumerate}
\end{lemma}

\subsubsection{Functions related to $H_{1+\theta}^{\downarrow,W}(X|Y)$}
From Statement \ref{item:multi-terminal-markov-property-1} of Lemma \ref{lemma:multi-terminal-markov-property},
$\frac{d [\theta H_{1+\theta}^{\downarrow,W}(X|Y)]}{d\theta}$ is monotonically decreasing.
Thus, we can define the inverse function $\theta(a) = \theta^\downarrow(a)$ 
of $\frac{d [\theta H_{1+\theta}^{\downarrow,W}(X|Y)]}{d\theta}$ by
\begin{eqnarray} 
\Label{eq:definition-theta-inverse-multi-markov}
\frac{d [\theta H_{1+\theta}^{\downarrow,W}(X|Y)]}{d\theta} \bigg|_{\theta = \theta(a)} = a
\end{eqnarray}
for $\underline{a} < a \le \overline{a}$, where 
$\underline{a}=\underline{a}^{\downarrow} := \lim_{\theta \to \infty} \frac{d [\theta H_{1+\theta}^{\downarrow,W}(X|Y)]}{d\theta}$
and $\overline{a}=
\overline{a}^\downarrow := \lim_{\theta \to -1} \frac{d [\theta H_{1+\theta}^{\downarrow,W}(X|Y)]}{d\theta}$.
Then, 
due to the definition \eqref{eq:definition-theta-inverse-multi-markov},
we have the following lemma
because the function $\theta \mapsto \theta H_{1+\theta}^{\downarrow,W}(X|Y)$ is concave.
\begin{lemma}\Label{L10}
\MH{The function $\theta(R ) $ defined in \eqref{eq:definition-theta-inverse-multi-markov}
satisfies that}
\begin{align}
\theta(R ) H_{1+\theta(R )}^{\downarrow,W}(X|Y)
 - \theta(R ) R 
=
\sup_{0 \le \theta } 
(\theta H_{1+\theta}^{\downarrow,W}(X|Y)- \theta R).
\end{align}
\end{lemma}

Next, let 
\begin{eqnarray}
R^{\downarrow}(a) := (1+\theta(a)) a - \theta(a) H_{1+\theta(a)}^{\downarrow,W}(X|Y).
\Label{11-14-2}
\end{eqnarray}
Since 
\begin{eqnarray}
\MH{\frac{d R^{\downarrow}}{d a}(a) = 1+\theta(a),}
\end{eqnarray}
$R(a)$ is a monotonic increasing function of $\underline{a} < a < R(\overline{a})$. Thus, we can define 
the inverse function $a(R ) = a^\downarrow(R )$ of $R(a)$ by
\begin{eqnarray} \Label{eq:definition-a-inverse-multi-markov}
(1+\theta(a(R ))) a(R ) - \theta(a(R )) H_{1+\theta(a(R ))}^{\downarrow,W}(X|Y) = R
\end{eqnarray}
for $R(\underline{a}) < R < H_0^{\downarrow,W}(X|Y)$, where 
$H_0^{\downarrow,W}(X|Y) := \lim_{\theta \to -1} H_{1+\theta}^{\downarrow,W}(X|Y)$.


\MH{Due to \eqref{eq:lower-conditional-renyi-markov-theta-0-derivative},}
when $\theta(a)$ is close to $0$, we have
\begin{align}
&\theta(a) H_{1+\theta(a)}^{\downarrow,W}(X|Y) \nonumber \\
=& \theta(a) H^{W}(X|Y)-\frac{1}{2}V^{W}(X|Y)\theta(a)^2+ o(\theta(a)^2) .
\Label{11-21-4b}
\end{align}
Taking the derivative, \eqref{eq:definition-theta-inverse-multi-markov} implies that
\begin{align}
a= H^{W}(X|Y)- V^{W}(X|Y) \theta(a)+ o(\theta(a)) .
\Label{11-21-5}
\end{align}
Hence, when $R$ is close to $H^{W}(X|Y)$, we have
\begin{align}
R=&(1+\theta (a(R) )) a(R)
-\theta H_{1+\theta(a(R))}^{\downarrow,W}(X|Y) \nonumber \\
=& H^{W}(X|Y)-(1+\frac{\theta(a(R))}{2})\theta(a(R)) V^{W}(X|Y)
\nonumber \\
&+ o(\theta(a(R))) ,\Label{11-21-1}
\end{align}
i.e.,
\begin{align}
\theta(a(R))
= -\frac{R -H^{W}(X|Y)}{V^{W}(X|Y)}+ o(\frac{R -H^{W}(X|Y)}{V^{W}(X|Y)})
\Label{11-21-3}.
\end{align}
Further, Eqs. \eqref{11-21-4b} and \eqref{11-21-5} imply
\begin{align}
& - \theta(a(R )) a(R ) + \theta(a(R )) H_{1+\theta(a(R ))}^{\downarrow,W}(X|Y)
\nonumber \\
=& V^{W}(X|Y)
\frac{\theta(a(R))^2}{2}
+ o(\theta(a(R))^2) 
\nonumber \\
=& 
\frac{V^{W}(X|Y)}{2}
(\frac{R -H^{W}(X|Y)}{V^{W}(X|Y)})^2
+ o((\frac{R -H^{W}(X|Y)}{V^{W}(X|Y)})^2 )
\Label{11-21-4}.
\end{align}

\subsubsection{Functions related to $H_{1+\theta}^{\uparrow,W}(X|Y)$}
For $\theta H_{1+\theta}^{\uparrow,W}(X|Y)$, by the same reason, we can define the inverse 
function $\theta(a) = \theta^\uparrow(a)$ by
\begin{eqnarray} 
&&\frac{d [\theta H_{1+\theta,1+ \theta(a)}^{W}(X|Y)]}{d\theta} \bigg|_{\theta = \theta(a)} 
\nonumber \\
\Label{eq:definition-theta-inverse-markov-optimal-Q}
&=&
\frac{d [\theta H_{1+\theta}^{\uparrow,W}(X|Y)]}{d\theta} \bigg|_{\theta = \theta(a)} 
= a \Label{11-27-3}
\end{eqnarray}
for $\underline{a} < a \le \overline{a}$, where 
$\underline{a}=\underline{a}^{\uparrow} 
:= \lim_{\theta \to \infty} \frac{d [\theta H_{1+\theta}^{\uparrow,W}(X|Y)]}{d\theta}$
and $\overline{a}=\overline{a}^\uparrow 
:= \lim_{\theta \to -1} \frac{d [\theta H_{1+\theta}^{\uparrow,W}(X|Y)]}{d\theta}$.
Here, the first equation in \eqref{11-27-3} follows from \eqref{11-27-1}.
We also define the inverse function $a(R ) = a^\uparrow(R )$ of 
\begin{eqnarray} \Label{eq:definition-R-a-optimal-q-markov}
R^{\uparrow}(a) := (1+\theta(a)) a - \theta(a) H_{1+\theta(a)}^{\uparrow,W}(X|Y) 
\end{eqnarray}
by
\begin{eqnarray} \Label{eq:definition-a-inverse-markov-optimal-Q}
(1+\theta(a(R ))) a(R ) - \theta(a(R )) H_{1+\theta(a(R ))}^{\uparrow,W}(X|Y) = R
\end{eqnarray}
for $R(\underline{a}) < R < H_0^{\uparrow,W}(X|Y)$, where 
$H_0^{\uparrow,W}(X|Y) := \lim_{\theta \to -1} H_{1+\theta}^{\uparrow,W}(X|Y)$.
Then, we can show the following lemma in the same way as Lemma 8 of \cite{HW14}.
\begin{lemma}\Label{L9}
For $R(\underline{a}) < R < H_0^{\uparrow,W}(X|Y)$, 
we have
\begin{eqnarray}
&&\sup_{\theta \ge 0} \frac{- \theta R + \theta H_{1+\theta}^{\uparrow,W}(X|Y)}{1+\theta}
\nonumber \\
&=& - \theta(a(R )) a(R ) + \theta(a(R )) H_{1+\theta(a(R ))}^{\uparrow,W}(X|Y).
\end{eqnarray}
When the rate $R$ is larger than the critical rate $R_{\rom{cr}}$ defined by
\begin{eqnarray}
R_{\rom{cr}} := R\left( \frac{d[\theta H_{1+\theta}^{\uparrow,W}(X|Y)]}{d\theta} \bigg|_{\theta = 1} \right),
\Label{eq62}
\end{eqnarray}
the definition \eqref{eq:definition-R-a-optimal-q-markov} of $R(a)=R^\uparrow(a)$
yields 
\begin{eqnarray}
&&\sup_{0 \le \theta \le 1} \frac{- \theta R + \theta H_{1+\theta}^{\uparrow,W}(X|Y)}{1+\theta}
\nonumber \\
&=& - \theta(a(R )) a(R ) + \theta(a(R )) H_{1+\theta(a(R ))}^{\uparrow,W}(X|Y).
\end{eqnarray}
\end{lemma}

\begin{remark}
As we can find from \eqref{eq:lower-conditional-renyi-markov-theta-0B}, \eqref{eq:lower-conditional-renyi-markov-theta-0-derivative},
and Lemma \ref{lemma:properties-upper-conditional-renyi-transition-matrix}, both the conditional R\'enyi entropies expand as
\begin{eqnarray}
H_{1+\theta}^{\downarrow,W}(X|Y) &=& H^W(X|Y) - \frac{1}{2} \san{V}^W(X|Y) \theta + o(\theta), \\
H_{1+\theta}^{\uparrow,W}(X|Y) &=& H^W(X|Y) - \frac{1}{2} \san{V}^W(X|Y) \theta + o(\theta)
\end{eqnarray}
around $\theta = 0$.
Thus, the difference of these measures significantly appear only when
$|\theta|$ is rather large. 
\end{remark}

\begin{remark}\Label{rem11-14-1}
When ${\cal Y}$ is singleton, 
$H_{1+\theta}^{\downarrow,W}(X|Y)$
coincides with $H_{1+\theta}^{\uparrow,W}(X|Y)$.
So, they are simply called the R\'enyi entropy 
and denoted by $H_{1+\theta}^W(X)$ for $W$.
$\theta^{\downarrow}(a)$, $a^{\downarrow}(R)$, 
$R^{\downarrow}(a)$, $\underline{a}^{\downarrow}$,
and $\overline{a}^{\downarrow} $
coincide with
$\theta^{\uparrow}(a)$, $a^{\uparrow}(R)$, $R^{\uparrow}(a)$,
$\underline{a}^{\uparrow} $,
and $\overline{a}^{\uparrow} $.
They are simplified to 
$\theta(a)$, $a(R)$, and $R(a)$,
$\underline{a}$, and $\overline{a}$.
\end{remark}

\subsection{Information Measures for Markov Chain} 
\Label{subsection-multi-terminal-information-measures-markov}

Let $(\mathbf{X},\mathbf{Y})$ be the Markov chain induced by a transition matrix $W$ and some initial distribution $P_{X_1Y_1}$.
Now, we show how information measures introduced in Section \ref{subsection:multi-terminal-measures-markov} are 
related to the conditional R\'enyi entropy rates.
First, we introduce the following lemma, which
gives finite upper and lower bounds on the lower conditional R\'enyi entropy.

\begin{lemma}[\protect{\cite[Lemma 9]{HW14}}] \Label{lemma:mult-terminal-finite-evaluation-down-conditional-renyi}
Suppose that a transition matrix $W$ satisfies Assumption \ref{assumption-Y-marginal-markov}.
Let $v_\theta$ be the eigenvector of $\tilde{W}_\theta^T$ with respect to the Perron-Frobenius eigenvalue $\lambda_\theta$
such that\footnote{\MH{Since the eigenvector corresponding to the Perron-Frobenius eigenvalue for an irreducible non-negative matrix
has always strictly positive entries\cite[Theorem 8.4.4, p. 508]{horn-johnson}, we can choose  
the eigenvector $v_\theta$ satisfying \eqref{1-14-1}.}}
\begin{align}
\min_{x,y} v_\theta(x,y) = 1. \Label{1-14-1}
\end{align}
Let $w_\theta(x,y) := P_{X_1 Y_1}(x,y)^{1+\theta} P_{Y_1}(y)^{-\theta}$. Then, we have
\begin{eqnarray}
&& (n-1) \theta H_{1+\theta}^{\downarrow,W}(X|Y) + \underline{\delta}(\theta)
 \le  \theta H_{1+\theta}^\downarrow(X^n|Y^n) \nonumber \\
 &\le &  (n-1) \theta H_{1+\theta}^{\downarrow,W}(X|Y) + \overline{\delta}(\theta),
 \end{eqnarray}
 where 
 \begin{eqnarray}
 \overline{\delta}(\theta) &:=& - \log \langle v_\theta | w_\theta \rangle + \log \max_{x,y} v_\theta(x,y), \\
 \underline{\delta}(\theta) &:=& - \log \langle v_\theta | w_\theta \rangle
<0,
 \end{eqnarray}
\MH{and $\langle v_\theta | w_\theta \rangle$ is defined as
$\sum_{x,y} v_\theta(x,y) w_\theta(x,y)$.}
\end{lemma}


From Lemma \ref{lemma:mult-terminal-finite-evaluation-down-conditional-renyi}, we have the following.

\begin{theorem}[\protect{\cite[Theorem 1]{HW14}}] \Label{theorem:asymptotic-down-conditional-renyi2}
Suppose that a transition matrix $W$ satisfies Assumption \ref{assumption-Y-marginal-markov}.
For any initial distribution, we have
\begin{eqnarray}
\lim_{n\to\infty} \frac{1}{n} H_{1+\theta}^\downarrow(X^n|Y^n) &=& H_{1+\theta}^{\downarrow,W}(X|Y), 
\Label{eq:single-terminal-markov-renyi-entropy-asymptotic-1}
\\
\lim_{n\to\infty} \frac{1}{n} H(X^n|Y^n) &=& H^W(X|Y).
\end{eqnarray}
\end{theorem}

We also have the following asymptotic evaluation of the variance:

\begin{theorem}[\protect{\cite[Theorem 2]{HW14}}] \Label{theorem:multi-markov-variance}
Suppose that the transition matrix $W$ satisfies Assumption \ref{assumption-Y-marginal-markov}.
For any initial distribution, we have
\begin{eqnarray}
 \lim_{n\to\infty} \frac{1}{n} \san{V}(X^n|Y^n) = \san{V}^W(X|Y).
\end{eqnarray}
\end{theorem}

Theorem \ref{theorem:multi-markov-variance} is practically important since the limit of the variance can
be described by a single letter characterized quantity. A method to calculate $\san{V}^W(X|Y)$ can be 
found in \cite{hayashi-watanabe:13b}.

Next, we show the lemma that gives finite upper and lower bound on the upper conditional R\'enyi entropy in terms of
the upper conditional R\'enyi entropy for the transition matrix.

\begin{lemma}[\protect{\cite[Lemma 10]{HW14}}] \Label{lemma:multi-terminal-finite-evaluation-upper-conditional-renyi}
Suppose that a transition matrix $W$ satisfies Assumption \ref{assumption-memory-through-Y}. 
Let $v_\theta$ be the eigenvector of
$K_\theta^T$ with respect to the Perron-Frobenius eigenvalue $\kappa_\theta$ such that $\min_y v_\theta(y) = 1$.
Let $w_{Y,\theta}$ be the $|{\cal Y}|$-dimensional vector defined by
\begin{align}
w_{Y,\theta}(y) := \left[ \sum_x P_{X_1 Y_1}(x,y)^{1+\theta} \right]^{\frac{1}{1+\theta}}.
\end{align}
Then, we have
\begin{align}
&(n-1) \frac{\theta}{1+\theta} H_{1+\theta}^{\uparrow,W}(X|Y) + \underline{\xi}(\theta)
 \le \frac{\theta}{1+\theta}H_{1+\theta}^\uparrow(X^n|Y^n)
\nonumber \\
 \le &(n-1) \frac{\theta}{1+\theta} H_{1+\theta}^{\uparrow,W}(X|Y) + \overline{\xi}(\theta),
\end{align}
where 
\begin{eqnarray}
\overline{\xi}(\theta) &:=& - \log \langle v_\theta | w_{Y,\theta} \rangle + \log \max_y v_\theta(y), \\
\underline{\xi}(\theta) &:=& - \log \langle v_\theta | w_{Y,\theta} \rangle.
\end{eqnarray}
\end{lemma}

From Lemma \ref{lemma:multi-terminal-finite-evaluation-upper-conditional-renyi}, we have the following.

\begin{theorem}[\protect{\cite[Theorem 3]{HW14}}] \Label{theorem:asymptotic-up-conditional-renyi}
Suppose that a transition matrix $W$ satisfies Assumption \ref{assumption-memory-through-Y}. 
For any initial distribution, we have
\begin{eqnarray}
\lim_{n\to\infty} \frac{1}{n} H_{1+\theta}^\uparrow(X^n|Y^n) &=& H_{1+\theta}^{\uparrow,W}(X|Y).
 \Label{eq:markov-theta-entropy-up-asymptotic} 
\end{eqnarray}
\end{theorem}

Finally, we show the lemma that gives finite upper and lower bounds on the two-parameter conditional 
R\'enyi entropy in terms of the two-parameter conditional R\'enyi entropy for the transition matrix. 

\begin{lemma}[\protect{\cite[Lemma 11]{HW14}}] \Label{lemma:multi-terminal-finite-evaluation-two-parameter-conditional-renyi}
Suppose that a transition matrix $W$ satisfies Assumption \ref{assumption-memory-through-Y}. 
Let $v_{\theta,\theta^\prime}$ be the eigenvector of
$N_{\theta,\theta^\prime}^T$ with respect to the Perron-Frobenius eigenvalue $\nu_{\theta,\theta^\prime}$ such that 
$\min_y v_{\theta,\theta^\prime}(y) = 1$.
Let $w_{\theta,\theta^\prime}$ be the $|{\cal Y}|$-dimensional vector defined by
\begin{align}
w_{\theta,\theta^\prime}(y) := \left[ \sum_x P_{X_1 Y_1}(x,y)^{1+\theta} \right] 
 \left[ \sum_x P_{X_1 Y_1}(x,y)^{1+\theta^\prime} \right]^{\frac{-\theta}{1+\theta^\prime}}.
\end{align}
Then, we have
\begin{align}
& (n-1) \theta H_{1+\theta,1+\theta^\prime}^W(X|Y) + \underline{\zeta}(\theta,\theta^\prime)
 \le \theta H_{1+\theta,1+\theta^\prime}(X^n|Y^n)
\nonumber \\
 \le & (n-1) \theta H_{1+\theta,1+\theta^\prime}^W(X|Y) + \overline{\zeta}(\theta,\theta^\prime),
\end{align}
where 
\begin{align}
\overline{\zeta}(\theta,\theta^\prime) :=& - \log \langle v_{\theta,\theta^\prime} | 
w_{\theta,\theta^\prime} \rangle
   + \log \max_y v_{\theta,\theta^\prime}(y) + \theta \overline{\xi}(\theta^\prime), \\
\underline{\zeta}(\theta,\theta^\prime) :=& 
 - \log \langle v_{\theta,\theta^\prime} | w_{\theta,\theta^\prime} \rangle + \theta \underline{\xi}(\theta^\prime)
\end{align}
for $\theta > 0$ and
\begin{align}
\overline{\zeta}(\theta,\theta^\prime) :=& - \log \langle v_{\theta,\theta^\prime} | w_{\theta,\theta^\prime} \rangle
   + \log \max_y v_{\theta,\theta^\prime}(y) + \theta \underline{\xi}(\theta^\prime), \\
\underline{\zeta}(\theta,\theta^\prime) :=& 
 - \log \langle v_{\theta,\theta^\prime} | w_{\theta,\theta^\prime} \rangle + \theta \overline{\xi}(\theta^\prime)
\end{align}
for $\theta < 0$.
\end{lemma}

From Lemma \ref{lemma:multi-terminal-finite-evaluation-two-parameter-conditional-renyi}, we have the following.

\begin{theorem}[\protect{\cite[Theorem 4]{HW14}}] \Label{theorem:asymptotic-two-parameter-renyi-markov}
Suppose that a transition matrix $W$ satisfies Assumption \ref{assumption-memory-through-Y}. 
For any initial distribution, we have
\begin{eqnarray}
\lim_{n\to\infty} \frac{1}{n} H_{1+\theta,1+\theta^\prime}(X^n|Y^n) 
 = H_{1+\theta,1+\theta^\prime}^W(X|Y).
\end{eqnarray}
\end{theorem}


\subsection{Analysis with $\theta=\infty$: One-terminal case}\Label{s2d}
To close this section, 
we address the case $\theta=\infty$, which was not discussed in the paper \cite{HW14}.
Since the conditional R\'{e}nyi entropy is monotonically decreasing for $\theta$, 
the conditional R\'{e}nyi entropy with the case $\theta=\infty$
is often called the conditional min entropy.
To avoid difficulty, we first consider the case when ${\cal Y}$ is singleton.

For a single-shot random variable, we have
\begin{eqnarray}
\lim_{\theta \to \infty} H_{1+\theta}(X) &=& H_\infty(X)\\
&:=& - \log \max_x P_X(x),
\Label{eq:definition-single-terminal-min-entropy}
\end{eqnarray}
which is usually called $\min$-entropy.
For each $x \in {\cal X}$, let ${\cal C}_{x}$ be the set of all
Hamilton cycle from $x$ to itself. 
For a path $c=(x_1, x_2, \ldots, x_k)$, 
we define the set $\hat{c}:=\{(x_i,x_{i+1})\}_{i=1}^{k-1}$ 
and the number $|c|$ to be the number of edges in cycle $c$, which is 
the number of elements in the set $\hat{c}$.
Then, we define the $\min$-entropy for $W$ by
\begin{eqnarray}
H_{\infty}^W(X) := - \log \max_{\bar{x} \in {\cal X}} \max_{c \in {\cal C}_{\bar{x}}} 
 \left( \prod_{(x_a,x_b) \in \hat{c}} W(x_b|x_a) \right)^{1/ |c|},
 \Label{eq:definition-of-min-entropy-markov-single-terminal}
\end{eqnarray}
which is characterized as follows.
\begin{lemma} \Label{lemma:limit-of-renyi-markov-single-terminal}
We have
\begin{eqnarray}
\lim_{\theta \to \infty} H_{1+\theta}^W(X) &=& H_{\infty}^W(X). 
\Label{eq:limit-of-renyi-markov-single-terminal-min-entropy}
\end{eqnarray}
\end{lemma}
\begin{proof}
See Appendix \ref{proof-lemma:limit-of-renyi-markov-single-terminal}.
\end{proof}

We also have the following lemma.

\begin{lemma} \Label{lemma:finite-evaluation-min-entropy}
For $(x,x^\prime)$, let ${\cal C}_{x,x^\prime}$ be the set of all Hamilton paths from
$x$ to $x^\prime$. Then, let
\begin{eqnarray}
A := \min_{(\bar{x},\bar{x}^\prime) \atop \bar{x} \neq \bar{x}^\prime} 
\max_{c \in {\cal C}_{\bar{x},\bar{x}^\prime}}
 \prod_{(x_a,x_b) \in \hat{c}} W(x_b|x_a).
\end{eqnarray}
Furthermore, let $x^*$ and $c^* \in {\cal C}_{x^*}$ be such that 
$H_{\infty}^W(X)$ is achieved in \eqref{eq:definition-of-min-entropy-markov-single-terminal}. 
Then, we have
\begin{align}
(n-1) H_\infty^W(X) + \underline{\delta}_\infty 
 \le & H_\infty(X^n) 
\nonumber \\
 \le &(n-1) H_\infty^W(X) + \overline{\delta}_\infty,
 \Label{eq:finite-evaluation-min-entropy}
\end{align}
where 
\begin{align}
\overline{\delta}_\infty :=& |c^*| H_\infty^W(X) - \log \max_x P_{X_1}(x) 
\nonumber \\
&\hspace{15ex} - \log 
\min(A, e^{- H_\infty^W(X) }), \\
\underline{\delta}_\infty :=& - \log \max_x P_{X_1}(x) + \log A.
\end{align}
\end{lemma}
\begin{proof}
See Appendix \ref{proof-lemma:finite-evaluation-min-entropy}.
\end{proof}

From Lemma \ref{lemma:finite-evaluation-min-entropy}, we can derive the following.
\begin{theorem} \Label{theorem:single-terminal-markov-renyi-entropy-asymptotic}
For any initial distribution, we have
\begin{eqnarray}
\lim_{n\to\infty} \frac{1}{n} H_\infty(X^n) &=& H_\infty^W(X). \Label{eq:single-terminal-markov-renyi-entropy-asymptotic-4}
\end{eqnarray}
\end{theorem}

\subsection{Analysis with $\theta=\infty$: Two-terminal case}\Label{s2e}
Next, we proceed to the two-terminal case.
For single-shot random variables $X$ and $Y$, we can derive the following.

\begin{lemma}[\cite{iwamoto:13}]  \Label{lemma:limit-renyi-min-entropy}
We have
\begin{align}
\lim_{\theta \to \infty} H_{1+\theta}^\uparrow(X|Y) =& H_{\infty}^\uparrow(X|Y) \\
:=& - \log \sum_y P_Y(y) \max_x P_{X|Y}(x|y), \\
\lim_{\theta \to \infty} H_{1+\theta}^\downarrow(X|Y) =& H_\infty^\downarrow(X|Y) \\
:=& - \log \max_{x \in {\cal X} \atop y \in \rom{supp}(P_Y)} P_{X|Y}(x|y).
\end{align}
\end{lemma}

We define the lower $\min$-entropy for $W$ by
\begin{align}
&H_\infty^{\downarrow,W}(X|Y) 
\nonumber \\
:=& -\log \max_{(\bar{x},\bar{y}) \in {\cal X}\times {\cal Y}}
 \max_{c \in {\cal C}_{(\bar{x},\bar{y})}} \Biggl( 
\nonumber \\
& \hspace{9ex} 
\prod_{((x^\prime,y^\prime),(x,y)) \in \hat{c}} W_{X|X',Y',Y}(x|x^\prime,y^\prime,y) 
\Biggr)^{1/|c|}.
\Label{eq:definition-of-min2}
\end{align}
Then, similar to Lemma \ref{lemma:limit-of-renyi-markov-single-terminal},
we can show the following lemma.
\begin{lemma} \Label{lemma:extreme-cases-down-conditional-renyi-transition}
We have
\begin{eqnarray}
\lim_{\theta \to \infty} H_{1+\theta}^{\downarrow,W}(X|Y) &=& H_\infty^{\downarrow,W}(X|Y).
\end{eqnarray}
\end{lemma}

Next, we consider the upper $\min$-entropy for $W$.
When $W$ satisfies Assumption \ref{assumption-memory-through-Y}, we note that 
\begin{eqnarray}
T(y|y^\prime) &:=& \max_x W_{X|X',Y',Y}(x|x^\prime,y^\prime,y) \Label{eq:definition-of-T}
\end{eqnarray}
is well defined, i.e., the right hand side of \eqref{eq:definition-of-T}
is independent of $x^\prime$. 
Let $\kappa_\infty$ be the Perron-Frobenius eigenvalue of
$W_Y(y|y^\prime) T(y|y^\prime)$. Then, we define
\begin{eqnarray}
H_\infty^{\uparrow,W}(X|Y) := - \log \kappa_\infty.
\end{eqnarray}
\begin{lemma} \Label{lemma:extreme-cases-up-conditional-renyi-transition}
We have
\begin{eqnarray}
\lim_{\theta \to \infty} H_{1+\theta}^{\uparrow,W}(X|Y) &=& H_\infty^{\uparrow,W}(X|Y).
 \Label{eq:extreme-cases-up-conditional-renyi-transition-min} 
\end{eqnarray}
\end{lemma}
\begin{proof}
See Appendix \ref{appendix:lemma:extreme-cases-up-conditional-renyi-transition}.
\end{proof}

\begin{theorem} \Label{theorem:asymptotic-down-conditional-renyi}
Suppose that a transition matrix $W$ satisfies Assumption \ref{assumption-Y-marginal-markov}.
For any initial distribution, we have
\begin{eqnarray}
\lim_{n\to\infty} \frac{1}{n} H_\infty^\downarrow(X^n|Y^n) &=& H_\infty^{\downarrow,W}(X|Y). 
\Label{eq:markov-min-entropy-down-asymptotic}
\end{eqnarray}
Suppose that a transition matrix $W$ satisfies Assumption \ref{assumption-memory-through-Y}. 
For any initial distribution, we have
\begin{eqnarray}
\lim_{n\to\infty} \frac{1}{n} H_\infty^\uparrow(X^n|Y^n) &=& H_\infty^{\uparrow,W}(X|Y).
 \Label{eq:markov-min-entropy-up-asymptotic}
\end{eqnarray}
\end{theorem}
\begin{proof}
See Appendix \ref{appendix:theorem:asymptotic-down-conditional-renyi}.
\end{proof}

%% file: Single-Random-Number.tex
\section{Uniform Random Number Generation} \Label{section:single-random-number}

In this section, we investigate the uniform random number generation
when there is no information leakage.
Then, we discuss the single terminal Markov chain.
In this case, as is explained in Remark \ref{rem11-14-1},
all quantities with the superscript $\downarrow$ equal those with the superscript $\uparrow$,
and these the superscripts are omitted.
We start this section by showing the problem setting in Section \ref{subsection:single-random-problem-formulation}.
Then, we review and introduce some single-shot bounds in Section \ref{subsection:single-random-single-shot}.
We derive non-asymptotic bounds for the Markov chain in Section \ref{subsection:single-random-finite-markov}.
Then, in Sections \ref{subsection:single-random-large-deviation} and \ref{subsection:single-random-mdp},
we show the asymptotic characterization for the large deviation regime and the moderate deviation regime
by using those non-asymptotic bounds. We also derive the second order rate in Section \ref{subsection:single-random-second-order}.


The results shown in this section are summarized in Table \ref{table:summary:single-random-number}.
The checkmarks $\checkmark$ indicate that the tight asymptotic bounds (large deviation, moderate deviation, and second order)
 can be obtained from those bounds.
The marks $\checkmark^*$ indicate that the large deviation bound can be derived up to the critical rate. 
The computational complexity ``Tail'' indicates that the computational complexities of those bounds 
depend on the computational complexities of tail probabilities.

In Table \ref{table:summary:single-random-number},
we didn't call the bounds of Lemmas \ref{lemma:single-random-number-exponential-bound} and \ref{lemma:single-random-number-leftover-loose} as theorems due to the following reason.
In Subsection \ref{s1-1}, we listed the requirement for the finite-length bounds.
Hence, we give a status of Theorem only for a non-asymptotic bound with a computable form.
However, 
Lemmas \ref{lemma:single-random-number-exponential-bound} and \ref{lemma:single-random-number-leftover-loose}
require the calculation of the tail probability 
whose calculation complexity is not $O(1)$ at least in the Markovian case.
Hence, 
Lemmas \ref{lemma:single-random-number-exponential-bound} and \ref{lemma:single-random-number-leftover-loose}
are not given the status of Theorem
although they derive the asymptotic tight bounds.

\subsection{Problem Formulation} \Label{subsection:single-random-problem-formulation}

We first present the problem formulation by the single shot setting. Let $X$ be a source whose 
distribution is $P$.
A random number generator is a function $f:{\cal X} \to \{1,\ldots,M\}$. The approximation error is defined by
\begin{eqnarray}
\Delta [f] := \frac{1}{2} \| P_{f(X)} - P_{\overline{U}} \|_1,
\end{eqnarray}
where $\overline{U}$ is the uniform random variable on $\{1,\ldots,M\}$. 
For notational convenience, we introduce the infimum of approximation error under the condition that the range size is $M$:
\begin{eqnarray}
\Delta(M) := \inf_{f} \Delta[f].
\end{eqnarray}

When we construct a random number generator, we often use a two-universal hash family ${\cal F}$ and a random function $F$
on ${\cal F}$. Then, we bound the approximation error averaged over the random function by only using the property of two-universality.
\MH{As explained in Subsection \ref{two-cri},
to take into the practical aspects,
we introduce the worst leaked information:}
\begin{eqnarray} \Label{eq:single-random-number-worst-u2}
\overline{\Delta}(M) := \sup_F \mathbb{E}[\Delta[F] ],
\end{eqnarray}
where the supremum is taken over all two-universal hash families from ${\cal X}$ to $\{1,\ldots,M\}$.
From the definition, we obviously have
$\Delta(M) \le \overline{\Delta}(M)$.
When we consider $n$-fold extension, the random number generator and related quantities are denoted with subscript $n$.
Instead of evaluating the approximation error $\Delta(M_n)$ (or $\overline{\Delta}(M_n)$) for given $M_n$, we are also interested in
evaluating 
\begin{eqnarray}
M(n,\varepsilon) &:=& \sup
\{ M_n : \Delta(M_n) \le \varepsilon \}, \Label{eq:definition-single-random-m-epsilon} \\
\overline{M}(n,\varepsilon) &:=& \sup
\{ M_n : \overline{\Delta}(M_n) \le \varepsilon \}
 \Label{eq:definition-single-random-barm-epsilon}
\end{eqnarray}
for given $0 \le \varepsilon < 1$.

When the output size $M$ is too large, 
$\Delta(M)$ and $\overline{\Delta}(M)$ are close to $1$.
So, the criteria $\Delta(M)$ and $\overline{\Delta}(M)$ do not work as 
proper security measures.
In this case, to quantify the performance of the output random number, 
according to Wyner \cite{wyner:75},
to discuss the imperfectness of the generated random number,
we focus on the difference between the entropies of 
the generated random number and the ideal uniform random number,
which is given as
\begin{align}
&\log M - H(P_{f(X)})\nonumber \\
=&
\log M - \sum_{z} \Big(\sum_{x\in f^{-1}(z)} P_X(x)\Big) 
\log \Big(\sum_{x\in f^{-1}(z)} P_X(x)\Big)
\nonumber \\
=&D(P_{f(X)} \| P_{\overline{U}}),
\end{align}
where $D(P\|Q)$ is the divergence between two distributions $P$ and $Q$.
When the block size is $n$, we call the quantity 
$\frac{1}{n}D(P_{f(X)} \| P_{\overline{U}})$
the relative entropy rate.
Then, we focus on the following quantities.
\begin{align}
D(M) &:= \inf_f D(P_{f(X)} \| P_{\overline{U}}) \\
\overline{D}(M) &:= \sup_F \mathbb{E} [ D(P_{F(X)} \| P_{\overline{U}})],
\end{align}
where the supremum is taken over all two-universal hash families from ${\cal X}$ to $\{1,\ldots,M\}$.
\MH{Due to the same reason for $\overline{\Delta}(M)$, 
we consider 
the criterion $\overline{D}(M)$ 
in addition to the criterion $D(M)$.}

\subsection{Single Shot Bounds} \Label{subsection:single-random-single-shot}

In this section, we review existing single shot bounds and also show  novel converse bounds.
For the information measures used below, see Remark \ref{rem11-14-1} in Section \ref{section:preparation-multi}, which explains the information measures when ${\cal Y}$ is singleton.
Furthermore, we need to introduce other information measures.
For $P \in \overline{{\cal P}}({\cal X})$, let 
\begin{eqnarray}
H_{\min}(P ) := \log \frac{1}{\max_x P(x)}
\end{eqnarray}
be the $\min$-entropy.
Then, let 
\begin{eqnarray}
H^\varepsilon_{\min}(P ) := \max_{P^\prime \in {\cal B}^\varepsilon(P )} H_{\min}(P^\prime)
\end{eqnarray}
and 
\begin{eqnarray}
\overline{H}^\varepsilon_{\min}(P ) := \max_{P^\prime \in \overline{{\cal B}}^\varepsilon(P )} H_{\min}(P^\prime)
\end{eqnarray}
be smooth $\min$-entropies, where 
\begin{eqnarray}
{\cal B}^\varepsilon(P ) &:=& \left\{ P^\prime \in {\cal P}({\cal X}) : \frac{1}{2} \|P - P^\prime \|_1 \le \varepsilon \right\}, \\
\overline{{\cal B}}^\varepsilon(P ) &:=& \left\{ P^\prime \in \overline{{\cal P}}({\cal X}) : \frac{1}{2} \|P - P^\prime \|_1 \le \varepsilon \right\},
\end{eqnarray}
and
${\cal P}({\cal X})$
($\overline{{\cal P}}({\cal X})$)
is the set of distributions (sub-distributions) 
over the set ${\cal X}$.

First, we have the following achievability bound.
\begin{lemma}[Lemma 2.1.1 of \cite{han:book}] \Label{lemma:single-random-number-han-book}
We have
\begin{eqnarray}
\Delta(M) \le \inf_{\gamma \ge 0} \left[ P_X\left\{ \log \frac{1}{P_X(
\MH{X}
)} < \gamma \right\} +  \frac{M}{e^{\gamma}} \right]. 
\end{eqnarray}
\end{lemma}

By using the two-universal hash family, we can derive the following bound.
%
\begin{lemma}[\cite{renner:05b}] \Label{lemma:single-random-number-leftover}
We have
\begin{eqnarray}
\overline{\Delta}(M) \le \inf_{0\le \varepsilon \le1}
\left[ 2 \varepsilon + \frac{1}{2} \sqrt{M e^{- \overline{H}_{\min}^\varepsilon(
\MH{P_X})}} \right].
\end{eqnarray}
\end{lemma}

\MH{However, the bound in Lemma \ref{lemma:single-random-number-leftover}
cannot be directly calculated in the Markovian chain.
To resolve this problem, we slightly loosen Lemma \ref{lemma:single-random-number-leftover} as follows.}

\begin{lemma} \Label{lemma:single-random-number-leftover-loose}
We have
\begin{eqnarray}
\overline{\Delta}(M) \le \inf_{\gamma \ge 0} 
\left[ \MH{P_X}\left\{ \log \frac{1}{\MH{P_X(X)}} < \gamma \right\} + \frac{1}{2} \sqrt{\frac{M}{e^\gamma}} \right].
\end{eqnarray}
\end{lemma}

We also have the following achievability bound.
\begin{lemma}[Theorem 1 of \cite{hayashi:10b}] \Label{lemma:single-random-number-exponential-bound}
We have
\begin{eqnarray}
\overline{\Delta}(M) \le \inf_{0 \le \theta \le 1} \frac{3}{2} M^{\frac{\theta}{1+\theta}} e^{- \frac{\theta}{1+\theta} H_{1+\theta}(X)}. 
\end{eqnarray}
\end{lemma}

We also have the following converse bound, which is a special case of 
Lemma \ref{lemma:multi-random-number-monotonicity-converse}
\MH{ahead for the more general non-singleton case.}

\begin{lemma} \Label{lemma:single-random-number-monotonicity-converse}
We have
\begin{eqnarray}
\Delta(M) \ge \min_{H_{\min}^\varepsilon(P ) \ge \log M} \varepsilon.
\end{eqnarray}
\end{lemma}

\MH{Similar to Lemma \ref{lemma:single-random-number-leftover}, 
the bound in Lemma \ref{lemma:single-random-number-monotonicity-converse}
cannot be directly calculated in the Markovian chain.
To resolve this problem, we slightly loosen Lemma \ref{lemma:single-random-number-monotonicity-converse} as follows.}

\begin{lemma} \Label{lemma:single-random-sphere-packing-converse}
We have
\begin{eqnarray}
\Delta(M) \ge \max_{\gamma \ge 0}
\left[ \MH{P_X}\left\{ \log \frac{1}{\MH{P_X(X)}} < \gamma \right\}\left( 1 - \frac{e^\gamma}{M} \right) \right].
\Label{eq:12-12-3}
\end{eqnarray}
\end{lemma}
\begin{proof}
Fix arbitrary $\gamma \ge 0$. Then, from Lemma \ref{lemma:single-random-number-monotonicity-converse}, 
there exists $P^\prime \in {\cal B}^\varepsilon(P )$ such that 
\begin{align} 
\Delta(M) &\ge \frac{1}{2} \| \MH{P_X} - P^\prime \|_1 ,
\Label{eq:12-12} \\
\log \frac{1}{\max_x P^\prime(x)} &\ge \log M.
\Label{eq:proof-single-random-converse-1}
\end{align}
Then, we have
\begin{align}
& \frac{1}{2} \| \MH{P_X} - P^\prime \|_1 
= \max_{S\subset {\cal X}} (\MH{P_X}(S)-P'(S))
\\
\ge& \MH{P_X} 
\left\{ x
: \log \frac{1}{P_X(x)} < \gamma \right\} 
- P^\prime 
\left\{x :
\log \frac{1}{\MH{P_X}(x)} < \gamma  \right\} \\
\ge& \MH{P_X} \left\{ x :
\log \frac{1}{\MH{P_X}(x)} < \gamma \right\} 
-  \frac{1}{M} \left| \left\{x: \log \frac{1}{\MH{P_X}(x)} < \gamma \right\} \right| 
 \Label{eq:proof-single-random-converse-2} \\
\ge& \MH{P_X} \left\{ x :\log \frac{1}{\MH{P_X}(x)} < \gamma \right\} 
- \frac{1}{M} \sum_{x : \log \frac{1}{\MH{P_X}(x)} < \gamma} \MH{P_X}(x) e^{\gamma} \\
=& P_X\left\{ \log \frac{1}{\MH{P_X}(X)} < \gamma \right\}\left( 1 - \frac{e^\gamma}{M} \right),
\Label{eq:12-12-2} 
\end{align}
where \eqref{eq:proof-single-random-converse-2} follows from \eqref{eq:proof-single-random-converse-1}.
\eqref{eq:12-12} and \eqref{eq:12-12-2} yield \eqref{eq:12-12-3}.
\end{proof}

\MH{Although Lemma \ref{lemma:single-random-sphere-packing-converse}
is useful for the large deviation regime and the moderate deviation regime,
it is not useful for the second order regime.
To resolve this problem, 
we loosen Lemma \ref{lemma:single-random-sphere-packing-converse} as follows.}

\begin{lemma}[Lemma 2.1.2 of \cite{han:book}] \Label{lemma:single-random-number-converse}
We have
\begin{eqnarray}
\Delta(M) \ge \max_{\gamma \ge 0} \left[ \MH{P_X}\left\{ \log \frac{1}{\MH{P_X(X)}} < \gamma \right\} - \frac{e^\gamma}{M} \right].
\end{eqnarray}
\end{lemma}

\MH{This fact implies that 
Lemma \ref{lemma:single-random-sphere-packing-converse}
is better than the previous bound given in Lemma \ref{lemma:single-random-number-converse}.}

Furthermore, by using a property of the strong universal hash family
introduced in \cite{hayashi:10b}, 
we can derive the following converse\footnote{%
The paper \cite{hayashi:10b} introduced the strong universal hash family
as a special case of a two-universal hash family.
Theorem 2 of \cite{hayashi:10b} shows that
the strong universal hash family $F$ satisfies 
$ \mathbb{E}[\Delta[F] ]
\ge \left(1-\frac{|\Omega|}{M}\right)^2 P_X(\Omega)$.}.
\begin{lemma}[Theorem 2 of \cite{hayashi:10b}] \Label{lemma:strong-universal-bound}
For any subset $\Omega \subset {\cal X}$ such that $|\Omega| \le M$, we have
\begin{eqnarray}
\overline{\Delta}(M) \ge \left(1-\frac{|\Omega|}{M}\right)^2 \MH{P_X}(\Omega).
\end{eqnarray} 
\end{lemma}

\MH{Similar to Lemmas \ref{lemma:single-random-number-leftover}
and \ref{lemma:single-random-number-monotonicity-converse}, 
the bound in Lemma \ref{lemma:strong-universal-bound}
cannot be directly calculated in the Markovian chain.
To resolve this problem, we modify Lemma \ref{lemma:strong-universal-bound} as follows.}

\begin{lemma} \Label{lemma:strong-universal-bound-tail-probability}
For any $0 < \nu < 1$, we have
\begin{eqnarray} \Label{eq:strong-universal-bound}
\overline{\Delta}(M) \ge (1-\nu)^2 
\MH{P_X} \left\{ \log \frac{1}{\MH{P_X(X)}} \le a(R ) \right\},
\end{eqnarray}
where $R = \log (M\nu)$, and $a(R )$ is the inverse function defined by \eqref{eq:definition-a-inverse-multi-markov}.
\end{lemma}
\begin{proof}
See Appendix \ref{appendix:lemma:strong-universal-bound-tail-probability}.
\end{proof}

To derive a converse bound for ${\Delta}(M)$ based on the R\'{e}nyi entropy,
we substitute the formula in Proposition \ref{theorem:one-shot-tail-converse-2}
in Appendix \ref{Appendix:preparation}
into the bound in  Lemma \ref{lemma:single-random-sphere-packing-converse} 
for $a = \gamma = \log (M/2)$.
So, we have the following.

\begin{theorem} \Label{theorem:random-number-exponential-converse}
We have
\begin{align}
&- \log \Delta(M) \nonumber \\
\le & \inf_{s > 0 \atop \tilde{\theta} > \theta(a)}
\frac{1}{s} 
\bigg[ (1+s) \tilde{\theta} \Big( H_{1+\tilde{\theta}}(X) - H_{1+(1+s)\tilde{\theta}}(X) \Big) \nonumber \\
&  - (1+s) \log \left( 1 - e^{(\theta(a) - \tilde{\theta}) a - \theta(a) H_{1+\theta(a)}(X) + \tilde{\theta} H_{1+\tilde{\theta}}(X)} \right)  \bigg]\nonumber \\
& + \log 2, \Label{11-23-10}
\end{align}
where $a = \log (M/2)$ and $\theta(a)$ is the inverse function defined in \eqref{eq:definition-theta-inverse-multi-markov}.
\end{theorem}
\begin{proof}
We evaluate $- \log \Delta(M) $ by using Lemma \ref{lemma:single-random-sphere-packing-converse}.
To evaluate the probability 
$P_X\left\{ \log \frac{1}{P_X(X)} < a \right\}
=P_X\left\{ \log P_X(X) > - a \right\}$,
we apply 
Proposition \ref{theorem:one-shot-tail-converse-2} in Appendix \ref{Appendix:preparation}
to the random variable $\log P_X(X) $
whose cumulant generating function $\phi(\rho)$ is 
$-\theta H_{1+\theta}(X)$.
Then, $\rho(-a)= \theta(a)$.
Hence,
\begin{align}
& - \log P_X\left\{ \log P_X(X) > - a \right\} 
\nonumber \\
\le & \inf_{s > 0 \atop \tilde{\theta} > \theta(a)}\frac{1}{s}
 \bigg[ (1+s) \tilde{\theta} \Big( H_{1+\tilde{\theta}}(X) - H_{1+(1+s)\tilde{\theta}}(X) \Big) \nonumber \\
 &   -  (1+s) \log \left( 1  -  e^{(\theta(a) - \tilde{\theta}) a 
 -  \theta(a) H_{1+\theta(a)}(X) + \tilde{\theta} H_{1+\tilde{\theta}}(X)} \right)  \bigg]  \Label{11-23-12}.
\end{align}
Since $1-\frac{e^\gamma}{M}=\frac{1}{2}$, we obtain \eqref{11-23-10}.
\end{proof}

To derive a converse bound for $\overline{\Delta}(M)$ based on the R\'{e}nyi entropy,
we substitute the formula in Proposition \ref{theorem:one-shot-tail-converse-2}
in Appendix \ref{Appendix:preparation} into the bound in  Lemma \ref{lemma:strong-universal-bound-tail-probability} 
for $\nu = \frac{1}{2}$.
So, we have the following.

\begin{theorem} \Label{theorem:random-number-strong-universal-exponential-converse}
We have
\begin{align}
&- \log \overline{\Delta}(M) \nonumber \\
\le & \inf_{s > 0 \atop \tilde{\theta} > \theta(a(R ))} \frac{1}{s}
 \Bigg[ (1+s) \tilde{\theta} \biggl( H_{1+\tilde{\theta}}(X) - H_{1+(1+s)\tilde{\theta}}(X) \biggr) 
\nonumber \\
 & \hspace{2ex} - (1+s) \log \bigg( 1 
\nonumber \\
 & \hspace{5ex}
- e^{(\theta(a(R )) - \tilde{\theta}) a(R ) - \theta(a(R )) H_{1+\theta(a(R ))}(X) + \tilde{\theta} H_{1+\tilde{\theta}}(X)} 
\biggr)  \Bigg] 
\nonumber \\
 &
+ 2 \log 2,\Label{11-23-14}
\end{align}
where $R = \log (M /2)$, and $\theta(a)$ and $a(R )$ are the inverse functions defined in
\eqref{eq:definition-theta-inverse-multi-markov} and \eqref{eq:definition-a-inverse-multi-markov}.
\end{theorem}

\begin{proof}
We evaluate $- \log \overline{\Delta}(M) $ by using Lemma \ref{lemma:strong-universal-bound-tail-probability} with $\nu=\frac{1}{2}$.
The probability 
$P_X\left\{ \log \frac{1}{P_X(X)} < a(R) \right\}
=
P_X\left\{ \log P_X(X) > - a(R) \right\}
$ can be evaluated by \eqref{11-23-12}.
Since $(1-\nu)^2 =\frac{1}{2^2}$, we obtain \eqref{11-23-14}.
\end{proof}

Finally, we address the \MH{relative entropy rate}.
As the direct part, we have the following theorem.

\begin{theorem}\Label{Th11-25-1}
\MH{The relative entropy} $\overline{D}(M) $ is evaluated as
\begin{align}
\overline{D}(M) 
\le \frac{1}{\theta}\log ( 1+ M^\theta e^{-\theta H_{1+\theta}(X)}).
\end{align}
\end{theorem}

\begin{proof}
Lemma 10 of \cite{matsumoto-hayashi:11} shows that
any two-universal hash function $F$ satisfies 
the relation
\begin{align}
\mathbb{E}
[M^\theta e^{-\theta H_{1+\theta}(F(X))}]
\le 1+ M^\theta e^{-\theta H_{1+\theta}(X)},
\end{align}
which implies that
$\mathbb{E} [\log M - H(F(X))]
\le 
\mathbb{E} [\log M - H_{1+\theta}(F(X))]
=\mathbb{E}
\frac{1}{\theta} \log 
 (M^\theta e^{-\theta H_{1+\theta}(F(X))})
\le
\frac{1}{\theta} \log 
\mathbb{E}
 (M^\theta e^{-\theta H_{1+\theta}(F(X))})
\le
\frac{1}{\theta}\log ( 1+ M^\theta e^{-\theta H_{1+\theta}(X)})$.
\end{proof}

As the converse part, we have the following theorem.
\begin{proposition}\Label{Th11-25-2}
\begin{align}
D(M) \ge \log M - H(P_X)
\Label{11-25-1}
\end{align}
\end{proposition}

\begin{proof}
Inequality \eqref{11-25-1} follows from the inequality
$H(P_X) \ge H(P_{f(X)})$.
\end{proof}

\subsection{\MH{Finite-Length} Bounds for Markov Source} \Label{subsection:single-random-finite-markov}
\MH{In this subsection, we derive several finite-length bounds for Markovian source with 
a computable form.
Unfortunately, it is not easy to evaluate how tight these bounds are only with their formula.
Their tightness will be discussed 
by considering the asymptotic limit 
in the remaining subsections of this section.
Since we assume the irreducibility for the transition matrix describing the Markovian chain,
the following bounds hold with any initial distribution.}

To lower bound $- \log \overline{\Delta}(M_n)$ by the R\'{e}nyi entropy of transition matrix,
we substitute the formula 
for the R\'{e}nyi entropy given in Lemma \ref{lemma:mult-terminal-finite-evaluation-down-conditional-renyi}
into the bound in Lemma \ref{lemma:single-random-number-exponential-bound},
we have the following bound.

\begin{theorem} \Label{theorem:single-random-finite-markov-direct-1}
Let $R := \frac{1}{n}\log M_n$. Then
we have
\begin{align}
&- \log \overline{\Delta}(M_n) 
\nonumber \\
\ge & \sup_{0 \le \theta \le 1} \frac{-\theta n R + (n-1) \theta H_{1+\theta}^W(X)  + \underline{\delta}(\theta)}{1+\theta} -\log(3/2).
 \Label{eq:single-random-finite-markov-direct-1-1} 
 \end{align}
\end{theorem}

To upper bound $- \log {\Delta}(M_n)$ by the R\'{e}nyi entropy of transition matrix,
we substitute the formula 
for the tail probability given in 
and Proposition \ref{proposition:general-markov-tail-converse}
with $a=R$
into the bound in 
Lemma \ref{lemma:single-random-sphere-packing-converse} 
with $\gamma=nR$,
we have the following bound.

\begin{theorem} \Label{theorem:single-random-sphere-packing-converse-finite-markov}
Let $R = \frac{1}{n}\log (M_n/2)$. If $\underline{a} < R < H^W(X)$, then we have
\begin{align}
& -\log \Delta(M_n) \nonumber \\
\le & \inf_{s > 0 \atop \tilde{\theta} > \theta(R )}
\frac{1}{s}\Bigg[
  (n-1)(1+s) \tilde{\theta} \Bigl[ H_{1+\tilde{\theta}}^W(X) \!-\! H_{1+(1+s)\tilde{\theta}}^W(X) \Bigr]
\!+\! \delta_1 \nonumber \\
&  - (1+s) \log \biggl(1 
\nonumber \\
&\hspace{6ex}- e^{(n-1)[(\theta(R ) - \tilde{\theta}) R - \theta(R ) H_{1+\theta(R )}^W(X) + \tilde{\theta} H_{1+\tilde{\theta}}^W(X)] + \delta_2} 
\biggr)  
   \Bigg] 
\nonumber \\
&+\log 2, \Label{11-21-7}
\end{align}
where $\theta(a)$ is the inverse function defined in \eqref{eq:definition-theta-inverse-multi-markov}, and 
\begin{eqnarray}
\delta_1 &=& (1+s) \overline{\delta}(\tilde{\theta}) - \underline{\delta}((1+s)\tilde{\theta}), \\
\delta_2 &=& (\theta(R ) - \tilde{\theta}) R + \overline{\delta}(\tilde{\theta}) - \underline{\delta}(\theta(R )).
\end{eqnarray}
\end{theorem}

\begin{proof}
Theorem \ref{theorem:single-random-sphere-packing-converse-finite-markov}
can be shown by the same way as Theorem 
\ref{theorem:random-number-exponential-converse} 
with replacing the role of Proposition \ref{theorem:one-shot-tail-converse-2} in Appendix \ref{Appendix:preparation}
 by Proposition \ref{proposition:general-markov-tail-converse}.
\end{proof}
\MH{To upper bound $- \log \overline{\Delta}(M_n)$ by the R\'{e}nyi entropy of transition matrix,
we substitute the formula 
for the tail probability given in 
and Proposition \ref{proposition:general-markov-tail-converse}
with $a=R$ into the bound in Lemma \ref{lemma:strong-universal-bound},
we have the following bound.}

\begin{theorem} \Label{theorem:single-random-number-strong-universal-finite-markov-converse}
Let $R$ be such that
\begin{align}
&(n-1) R + \left\{ (1+\theta(a(R ))) a(R ) - \underline{\delta}(\theta(a(R ))) \right\} 
\nonumber \\
=& \log (M_n/2).
\end{align}
If $R(\underline{a}) < R < H^W(X)$, then we have
\begin{align}
\lefteqn{ -\log \overline{\Delta}(M_n) } \nonumber \\
  \le & \inf_{s > 0 \atop \tilde{\theta} > \theta(a(R ))}
\frac{1}{s}
\bigg[
  (n-1)(1+s)\tilde{\theta} \bigg( H_{1+\tilde{\theta}}^W(X) - H_{1+(1+s)\tilde{\theta}}^W(X) \bigg) 
\nonumber \\
 & \hspace{15ex} + \delta_1 
 - (1+s) \log \left(1 - e^{C_{1,n}} \right)  
   \bigg]  +2\log 2, 
   \Label{11-21-9}
\end{align}
where $\theta(a)$ and $a(R )$ are the inverse functions defined in \eqref{eq:definition-theta-inverse-multi-markov}
and \eqref{eq:definition-a-inverse-multi-markov}, and 
\begin{align}
C_{1,n}:=&
(n-1)\bigg[(\theta(a(R ))- \tilde{\theta})a(R )
\nonumber \\
& \hspace{6ex}
- \theta(a(R ) ) H_{1+\theta(a(R ))}^W(X) 
+ \tilde{\theta} H_{1+\tilde{\theta}}^W(X)\bigg] + \delta_2,
\nonumber \\
\delta_1 :=& (1+s) \overline{\delta}(\tilde{\theta}) - \underline{\delta}((1+s)\tilde{\theta}), \nonumber\\
\delta_2 :=& (\theta(a(R )) - \tilde{\theta})a(R ) + \overline{\delta}(\tilde{\theta}) - \underline{\delta}(\theta(a(R ))).
\nonumber
\end{align}
\end{theorem}
\begin{proof}
See Appendix \ref{appendix:theorem:single-random-number-strong-universal-finite-markov-converse}.
\end{proof}

\MH{To upper bound $\overline{D}(e^{nR}) 
$ by the R\'{e}nyi entropy of transition matrix,
we substitute the formula 
for the R\'{e}nyi entropy given in Lemma 
\ref{lemma:mult-terminal-finite-evaluation-down-conditional-renyi}
into the bound in Theorem \ref{Th11-25-1},
we have the following bound for the relative entropy rate $\frac{1}{n}\overline{D}(e^{nR})$.}

\begin{theorem}\Label{Th11-25-3}
When $R-H^W_{1+\theta}(X)\ge 0$, 
for $\theta \in [0,1]$, we have 
\begin{align}
\frac{1}{n}\overline{D}(e^{nR}) 
\le 
 R-\frac{n-1}{n}H^W_{1+\theta}(X) + 
\frac{1}{\theta n}(\log 2 - \underline{\delta}(\theta) ).\Label{2-7-1}
\end{align}
\end{theorem}

\begin{proof}
Theorem \ref{Th11-25-1} and Lemma 
\ref{lemma:mult-terminal-finite-evaluation-down-conditional-renyi}
yield $(a)$ and $(b)$, respectively, in the following way.
\begin{align}
& D(e^{nR})  \nonumber \\
\stackrel{(a)}{\le} &
\frac{1}{\theta}\log ( 1+ e^{\theta(n R- H_{1+s}(X^n))}) \nonumber \\
\stackrel{(b)}{\le} &
\frac{1}{\theta}\log 
( 1+ e^{\theta(n R- (n-1) H^W_{1+\theta}(X))- \underline{\delta}(\theta ) }) \nonumber\\
= & n (R-H^W_{1+\theta}(X)) 
\nonumber \\
&
+ \frac{1}{\theta}\log ( e^{n \theta (H^W_{1+\theta}(X)-R)} 
+ e^{\theta H^W_{1+\theta}(X)- \underline{\delta}(\theta) }) \nonumber\\
\le & n (R-H^W_{1+\theta}(X)) + 
\frac{1}{\theta}\log ( 1
+ e^{\theta H^W_{1+\theta}(X)- \underline{\delta}(\theta) }) \nonumber\\
\le & n (R-H^W_{1+\theta}(X)) + 
\frac{1}{\theta}\log ( 2 e^{\theta H^W_{1+\theta}(X)
- \underline{\delta}(\theta) }) \nonumber\\
= & n (R-H^W_{1+\theta}(X)) + 
\frac{1}{\theta}(\log 2 + \theta H^W_{1+\theta}(X)- \underline{\delta}(\theta) ) \nonumber\\
=& n R-(n-1)H^W_{1+\theta}(X) + 
\frac{1}{\theta}(\log 2 - \underline{\delta}(\theta) ).
\end{align}
\end{proof}

\MH{To lower bound $\overline{D}(e^{nR}) $ 
by the R\'{e}nyi entropy of transition matrix,
we substitute the other formula 
for the R\'{e}nyi entropy given in 
Lemma \ref{lemma:mult-terminal-finite-evaluation-down-conditional-renyi} 
into the bound in Proposition \ref{Th11-25-2},
we have the following bound for the relative entropy rate $\frac{1}{n}\overline{D}(e^{nR})$.}

\begin{theorem}\Label{Th11-25-4}
For $\theta \in [0,1]$, we have 
\begin{align}
\frac{1}{n}D(e^{nR}) \ge  R - \frac{n-1}{n} H^W_{1-\theta}(X)
+ \frac{\underline{\delta}(-\theta)}{\theta n}
\Label{11-25-4}
\end{align}
\end{theorem}

\begin{proof}
Lemma \ref{lemma:mult-terminal-finite-evaluation-down-conditional-renyi} implies that
\begin{align}
H(X^n) \le H_{1-\theta}(X^n)  
\le (n-1) H^W_{1-\theta}(X)
- \frac{\underline{\delta}(-\theta)}{\theta}.
\end{align}
Hence, using Proposition \ref{Th11-25-2}, we obtain \eqref{11-25-4}.
\end{proof}

\subsection{Large Deviation} \Label{subsection:single-random-large-deviation}

\MH{Taking the limit in the formulas in
Theorems \ref{theorem:single-random-finite-markov-direct-1} and \ref{theorem:single-random-number-strong-universal-finite-markov-converse}}, we have the following.
\begin{theorem}\Label{T11-22-1}
For $R < H^W(X)$, we have
\begin{eqnarray} 
\liminf_{n\to\infty} - \frac{1}{n} \log \overline{\Delta}(e^{nR}) \ge \sup_{0 \le \theta \le 1} \frac{-\theta R + \theta H_{1+\theta}^W(X)}{1 +\theta}.
\Label{eq:single-random-ldp-direct}
\end{eqnarray}
On the other hand, for $R(\underline{a}) < R < H^W(X)$, we have
\begin{align}
& \lefteqn{\limsup_{n\to\infty} - \frac{1}{n} \log \overline{\Delta}(e^{nR}) }
\nonumber \\
\le & - \theta(a(R )) a(R ) + \theta(a(R )) H_{1+\theta(a( R))}^W(X)
\Label{11-21-13} \\
=& 
\sup_{0 \le \theta} 
\frac{-\theta R + \theta H_{1+\theta}^W(X)}{1 +\theta}.
\Label{11-27-6}
\end{align}
\end{theorem}

Due to Lemma \ref{L9},
the lower bound \eqref{eq:single-random-ldp-direct}
and the upper bound \eqref{11-27-6}
coincide when $R$ is not less than the critical rate $R_{\mathrm{cr}}$
given in \eqref{eq62}.

\begin{proof}
\eqref{eq:single-random-finite-markov-direct-1-1} yields \eqref{eq:single-random-ldp-direct}.
Lemma \ref{L9} guarantees \eqref{11-27-6}.
So, we will prove \eqref{11-21-13} as follows.

We fix $s > 0$ and $\tilde{\theta} > \theta(a(R ))$.
Then, \eqref{11-21-9} implies that
\begin{align}
 \lim_{n \to \infty} -\frac{1}{n}\log \overline{\Delta}(M_n) 
  \le & 
\frac{1+s}{s}
\tilde{\theta} \left( H_{1+\tilde{\theta}}^W(X) - H_{1+(1+s)\tilde{\theta}}^W(X) \right).
\Label{11-21-11}
\end{align}
Taking the limit $s \to 0$ and $\tilde{\theta} \to \theta(a(R ))$, we have
\begin{align}
& \frac{1+s}{s}\tilde{\theta} \left\{ H_{1+\tilde{\theta}}^W(X) - H_{1+(1+s)\tilde{\theta}}^W(X) \right\} 
\nonumber \\
=&
\frac{1}{s}
\left(
\tilde{\theta} H_{1+\tilde{\theta}}^W(X) - (1+s)\tilde{\theta} H_{1+(1+s)\tilde{\theta}}^W(X) 
\right)
+ \tilde{\theta}H_{1+\tilde{\theta}}^W(X) \nonumber \\
\to &
- 
\tilde{\theta}
\frac{d {\theta} H_{1+{\theta}}^W(X) }{d\theta} \biggl|_{\theta=\tilde{\theta}}
+ \tilde{\theta}H_{1+\tilde{\theta}}^W(X) 
\hspace{13ex} \hbox{(as $s \to 0$)} \nonumber \\
\to &
- 
\theta(a(R ))
\frac{d {\theta} H_{1+{\theta}}^W(X) }{d\theta} \biggl|_{\theta=\theta(a(R ))}
+ \theta(a(R )) H_{1+\theta(a(R ))}^{W}(X) 
\nonumber \\
& \hspace{35ex} \quad \hbox{(as $\tilde{\theta} \to \theta(a(R ))$)} \nonumber \\
\stackrel{(a)}{=} & 
\theta(a(R )) a + \theta(a(R )) H_{1+\theta(a(R ))}^{W}(X)
\Label{11-21-10},
\end{align}
where $(a)$ follows from 
\eqref{eq:definition-theta-inverse-markov-optimal-Q}.
Hence, \eqref{11-21-10} and \eqref{11-21-11} imply that
\begin{align}
 \lim_{n \to \infty} -\frac{1}{n}\log \overline{\Delta}(M_n) 
\le 
\theta(a(R )) a + \theta(a(R )) H_{1+\theta(a(R ))}^{W}(X),
\Label{11-21-12}
\end{align}
which implies 
\eqref{11-21-13}.
\end{proof}

For the general class of functions, we can derive the following converse bound from Theorem \ref{theorem:single-random-sphere-packing-converse-finite-markov}.
\begin{theorem}
For $\underline{a} < R < H^W(X)$, we have
\begin{align}
 \limsup_{n\to\infty} -\frac{1}{n} \log \Delta(e^{nR}) 
\le  - \theta(R ) R + \theta(R ) H_{1+\theta(R )}^W(X).
\end{align}
\end{theorem}

\subsection{Moderate Deviation} \Label{subsection:single-random-mdp}

Taking the limit with $R=H^W(X) - n^{-t}\delta$ in
Theorem \ref{theorem:single-random-finite-markov-direct-1} and 
Theorem \ref{theorem:single-random-sphere-packing-converse-finite-markov} (or 
Theorem \ref{theorem:single-random-number-strong-universal-finite-markov-converse}), we have the following.
\begin{theorem}\Label{T11-22-10}
For arbitrary $t \in (0,1/2)$ and $\delta > 0$, we have
\begin{align}
& \lim_{n\to\infty} - \frac{1}{n^{1-2t}} \log \Delta\left(e^{nH^W(X) - n^{1-t}\delta} \right)
\nonumber \\
=& \lim_{n\to\infty} - \frac{1}{n^{1-2t}} \log \overline{\Delta}\left(e^{nH^W(X) - n^{1-t}\delta} \right) \nonumber \\
=& \frac{\delta^2}{2 \san{V}^W(X)}.
\end{align}
\end{theorem}
\begin{proof}
We apply Theorem \ref{theorem:single-random-finite-markov-direct-1} and 
Theorem \ref{theorem:single-random-sphere-packing-converse-finite-markov}
to the case with $R=H^W(X) - n^{-t}\delta$, i.e.,
$\theta(a(R) )=  - n^{-t}\frac{\delta}{\san{V}^W(X)}+ o(n^{-t})$.
Eqs. \eqref{11-21-3} and \eqref{eq:single-random-finite-markov-direct-1-1} 
in Theorem \ref{theorem:single-random-finite-markov-direct-1}
imply that
\begin{align}
& - \log \overline{\Delta}(M_n) \nonumber \\
\ge & \sup_{0 \le \theta \le 1} \frac{-\theta n R + (n-1) \theta H_{1+\theta}^W(X)}{1+\theta}
\nonumber \\
&+\inf_{0 \le \theta \le 1} \frac{\underline{\delta}(\theta)}{1+\theta} -\log(3/2)\nonumber \\
\ge & n^{1-2t} \frac{\delta^2}{2 \san{V}^W(X)} + o(n^{1-2t}).
\end{align}
We fix an arbitrary $s>0$.
Since $\theta(R)=  - n^{-t}\frac{\delta}{\san{V}^W(X)}+ o(n^{-t})$,
we can choose $\tilde{\theta} >\theta(R) $ such that
$\tilde{\theta}=- n^{-t}\frac{\delta}{\san{V}^W(X)}+ o(n^{-t})$.
Then, \eqref{11-21-7} implies that
\begin{align}
& \lim_{n \to \infty}
-\frac{1}{n^{1-2t}}\log \Delta(M_n) \nonumber \\
  \le & 
\lim_{n \to \infty}
n^{2t} \frac{1+s}{s} 
\tilde{\theta} \left\{ H_{1+\tilde{\theta}}^W(X) - H_{1+(1+s)\tilde{\theta}}^W(X) \right\} 
\nonumber\\
=& \lim_{n \to \infty}
n^{2t} \frac{1+s}{s} 
s \tilde{\theta}^2 
\frac{d H_{1+{\theta}}^W(X) }{d\theta}\biggl|_{\theta= \tilde{\theta}}
=
(1+s)\frac{\delta^2}{2 \san{V}^W(X)}.
\end{align}
Taking the limit $s \to 0$, we obtain the desired argument.
\end{proof}

\subsection{Second Order} \Label{subsection:single-random-second-order}

By applying the central limit theorem to
Lemmas 
\ref{lemma:single-random-number-leftover-loose} and \ref{lemma:single-random-number-converse}, 
and by using Theorem \ref{theorem:multi-markov-variance}, we have the following.
\begin{theorem}\Label{T11-22-14}
For arbitrary $\varepsilon \in (0,1)$, we have
\begin{align}
&\lim_{n\to\infty} \frac{\log M(n,\varepsilon) - nH^W(X)}{\sqrt{n}} \nonumber \\
=& \lim_{n\to\infty} \frac{\log \overline{M}(n,\varepsilon) - nH^W(X)}{\sqrt{n}} 
= \sqrt{\san{V}^W(X)} \Phi^{-1}(\varepsilon).\Label{1-22-3}
\end{align}
\end{theorem}

\begin{proof}
The central limit theorem for Markovian process \cite{kontoyiannis:03,Jones2004,Meyn1994} \cite[Corollary 6.2.]{hayashi-watanabe:13b} guarantees that
the random variable $(-\log P_{X^n}(X^n) -n H^W (X))/\sqrt{n}$ 
asymptotically obeys the normal distribution with the average $0$ and the variance $\san{V}^W(X)$.
Let $R=\sqrt{\san{V}^W(X)} \Phi^{-1}(\varepsilon)$.
Substituting
$M=e^{nH^W(X)+ \sqrt{n}R }$ and $\gamma= nH^W(X)+ \sqrt{n}R + n^{\frac{1}{4}}$
in Lemma \ref{lemma:single-random-number-leftover-loose}, we have
\begin{align}
\lim_{n \to \infty} \overline{\Delta} (e^{nH^W(X)+ \sqrt{n}R }) 
\le  \epsilon. \Label{1-22-1}
\end{align}
Also, substituting
$M=e^{nH^W(X)+ \sqrt{n}R }$ and $\gamma= nH^W(X)+ \sqrt{n}R - n^{\frac{1}{4}}$
in Lemma \ref{lemma:single-random-number-converse}, we have
\begin{align}
\lim_{n \to \infty} \Delta (e^{nH^W(X)+ \sqrt{n}R }) 
\ge  \epsilon.\Label{1-22-2}
\end{align}
Combining \eqref{1-22-1} and \eqref{1-22-2}, we obtain \eqref{1-22-3}.
\end{proof}


\begin{figure*}[!t]
\includegraphics[width=0.5\linewidth]{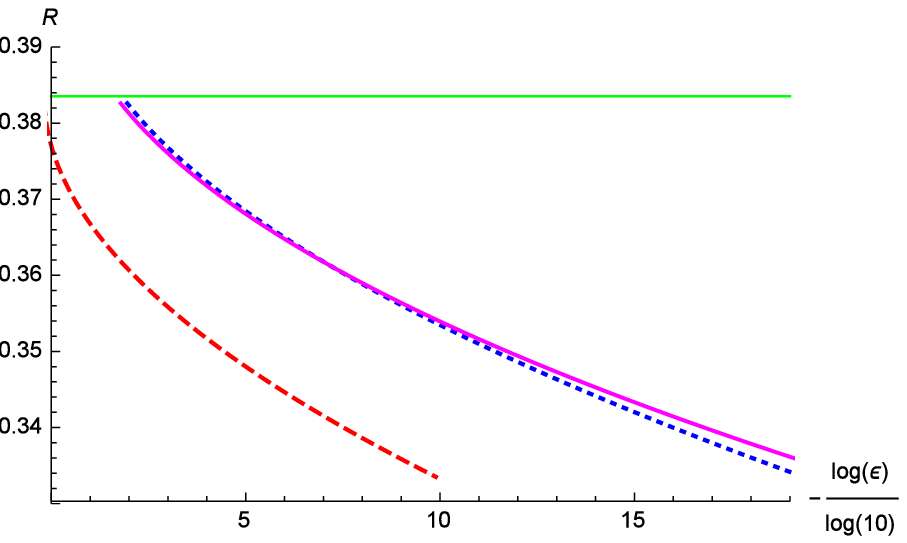}
\includegraphics[width=0.5\linewidth]{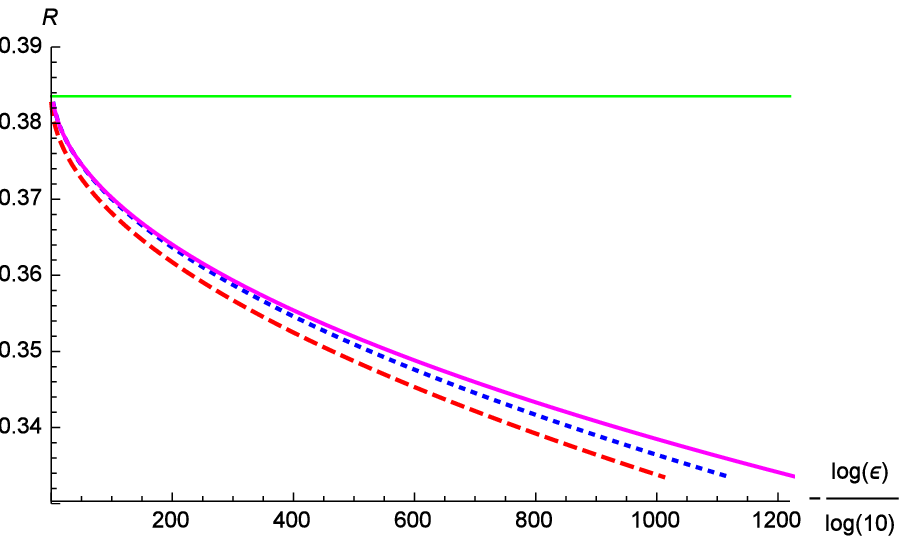}
\caption{Comparisons of the bounds for $p=0.1$ and $q=0.2$.
The left and right graphs express the cases with 
$n=10000$ and $1000000$, respectively.
The horizontal axis is $-\log_{10}(\varepsilon)$, and 
the vertical axis is the rate $R$ (nats).
The red dashed curve is the achievability bound in Theorem \ref{theorem:single-random-finite-markov-direct-1}.
The blue dotted curve is the converse bound in Theorem \ref{theorem:single-random-number-strong-universal-finite-markov-converse}.
The purple thick curve is the converse bound in Theorem \ref{theorem:single-random-sphere-packing-converse-finite-markov}.
The green normal horizontal line is the entropy $H^W(X)$.}
\label{ff2}
\end{figure*}

\subsection{\MH{Relative Entropy Rate (RER)}}\Label{Equivocation Rate} 
\MH{Taking the limit in} Theorems \ref{Th11-25-3} and \ref{Th11-25-4}, 
we have the following.
\begin{theorem}\Label{Th11-25-5}
\MH{The relative entropy rate (RER) is asymptotically calculated as}
\begin{align}
\lim_{n \to \infty} \frac{1}{n}D(e^{nR}) = 
\lim_{n \to \infty} \frac{1}{n}\overline{D}(e^{nR}) = 
[R-H^W(X)]_+  ,\Label{11-25-10}
\end{align}
where $[x]_+:= \max(x,0)$.
\end{theorem}

\begin{proof}
When $R \ge H^W_{1+\theta}(X)$,
\eqref{2-7-1} of Theorem \ref{Th11-25-3} implies that
\begin{align}
\lim_{n \to \infty} \frac{1}{n}\overline{D}(e^{nR}) 
\le R-H^W_{1+\theta}(X) 
\Label{11-25-8}
\end{align}
for $\theta \in (0,1)$.
Since $\overline{D}(e^{nR}) \ge \overline{D}(e^{nR'})$
for $R \ge R'$,
\eqref{11-25-8} implies that
\begin{align}
\lim_{n \to \infty} \frac{1}{n}\overline{D}(e^{nR}) 
\le [R-H^W_{1+\theta}(X)]_+ 
\end{align}
for $\theta \in (0,1)$ and any $R$.

Also, \eqref{11-25-4} of Theorem \ref{Th11-25-4} implies that
\begin{align}
\lim_{n \to \infty} \frac{1}{n}D(e^{nR}) 
\ge R-H^W_{1-\theta}(X) 
\Label{11-25-9}
\end{align}
for $\theta \in (0,1)$ and any $R$.
Since $D(e^{nR}) \ge 0$, we have
\begin{align}
\lim_{n \to \infty} \frac{1}{n}D(e^{nR}) 
\ge [R-H^W_{1-\theta}(X)]_+ 
\Label{11-25-9b}
\end{align}
for $\theta \in (0,1)$ and any $R$.
Taking the limit $\theta \to 0$, we have
\eqref{11-25-10}.
\end{proof}

%% file: example.tex
\subsection{Numerical Example} \label{subsection:single-source-numerical}


In this section, we numerically evaluate the achievability bound in Theorem \ref{theorem:single-random-finite-markov-direct-1}
and the converse bounds in 
Theorems \ref{theorem:single-random-sphere-packing-converse-finite-markov}
and \ref{theorem:single-random-number-strong-universal-finite-markov-converse}.
As shown in Theorem \ref{T11-22-1}, 
the finite-length bounds in 
Theorems \ref{theorem:single-random-finite-markov-direct-1}
and \ref{theorem:single-random-number-strong-universal-finite-markov-converse}
achieve the optimal rate in the sense of Large deviation 
when $R$ is larger than the critical rate.
Hence, we can expect that 
the converse bounds in 
Theorem \ref{theorem:single-random-number-strong-universal-finite-markov-converse}
is better than that in  
Theorem \ref{theorem:single-random-sphere-packing-converse-finite-markov}.
Now, we numerically demonstrate how 
the converse bounds in 
Theorem \ref{theorem:single-random-number-strong-universal-finite-markov-converse}
is better than that 
in Theorem \ref{theorem:single-random-sphere-packing-converse-finite-markov}.
Note that the single-shot bounds for second order 
in Lemmas \ref{lemma:single-random-number-leftover-loose} and \ref{lemma:single-random-number-converse}
are not given in  a computable form with Markovian case.
So, we compare the bounds given in Theorems 
\ref{theorem:single-random-finite-markov-direct-1}, \ref{theorem:single-random-sphere-packing-converse-finite-markov}
and \ref{theorem:single-random-number-strong-universal-finite-markov-converse}.

\begin{figure}[htbp]
\begin{center}
\includegraphics[scale=0.4]{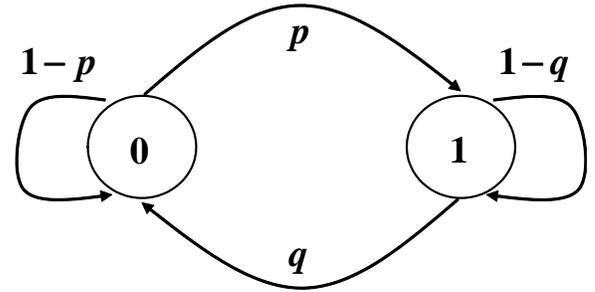}
\end{center}
\caption{The description of the transition matrix.}
\label{Fig:transition}
\end{figure}%

We consider a binary transition matrix $W$ given by Fig.~\ref{Fig:transition}, i.e.,
\begin{eqnarray}
W = \left[
\begin{array}{cc}
1-p & q \\
p & 1-q
\end{array}
\right].
\end{eqnarray}
In this case, the stationary distribution is
\begin{eqnarray}
\tilde{P}(0) &=& \frac{q}{p+q}, \\
\tilde{P}(1) &=& \frac{p}{p+q}.
\end{eqnarray}
The entropy is
\begin{eqnarray}
H^W(X) = \frac{q}{p+q} h(p ) + \frac{p}{p+q} h(q ),
\end{eqnarray}
where $h(\cdot)$ is the binary entropy function.
The tilted transition matrix is 
\begin{eqnarray}
\tilde{W}_\theta = \left[
\begin{array}{cc}
(1-p)^{1+\theta} & q^{1+\theta} \\
p^{1+\theta} & (1-q)^{1+\theta}
\end{array}
\right].
\end{eqnarray}
The Perron-Frobenius eigenvalue is
\begin{align}
\lambda_\theta 
=& \frac{(1-p)^{1+\theta} + (1-q)^{1+\theta} }{2}\nonumber \\
& + \frac{\sqrt{\{(1-p)^{1+\theta} - (1-q)^{1+\theta} \}^2 + 4 p^{1+\theta} q^{1+\theta}} }{2}
\end{align}
and its normalized eigenvector is
\begin{eqnarray}
\tilde{P}_\theta(0) &=& 
\frac{q^{1+\theta}}{\lambda_\theta - (1-p)^{1+\theta} + q^{1+\theta}}, \\
\tilde{P}_\theta(1) &=& 
\frac{\lambda_\theta - (1-p)^{1+\theta}}{\lambda_\theta - (1-p)^{1+\theta} + q^{1+\theta}}.
\end{eqnarray}
The eigenvector of $\tilde{W}_\theta^T$ satisfying \eqref{1-14-1} is also given by
\begin{eqnarray}
v_\theta(0) &=& 
\frac{q^{1+\theta}}
{\min(\lambda_\theta - (1-p)^{1+\theta}, q^{1+\theta})}, \\
v_\theta(1) &=& 
\frac{\lambda_\theta - (1-p)^{1+\theta}}
{\min(\lambda_\theta - (1-p)^{1+\theta}, q^{1+\theta})}.
\end{eqnarray}
From these calculations, we can evaluate the bounds in Theorems 
\ref{theorem:single-random-finite-markov-direct-1},
\ref{theorem:single-random-sphere-packing-converse-finite-markov},
and 
\ref{theorem:single-random-number-strong-universal-finite-markov-converse}.
When the initial distribution is given as $P_X(0)=1$ and  $P_X(1)=0$,
for $p = 0.1$, $q=0.2$, 
we plotted the bounds in Fig.~\ref{ff2} for fixed block length $n=10000$ and $n=1000000$
and varying $\varepsilon= \Delta(M)$ or $\overline{\Delta}(M)$. 
The two bounds in Theorems \ref{theorem:single-random-sphere-packing-converse-finite-markov}
and 
\ref{theorem:single-random-number-strong-universal-finite-markov-converse}
have similar values in the left of Fig.~\ref{ff2}.
However,
the bound in Theorem \ref{theorem:single-random-number-strong-universal-finite-markov-converse}
has a clear advantage in the right of Fig.~\ref{ff2}.
That is, to clarify the advantage of Theorem \ref{theorem:single-random-number-strong-universal-finite-markov-converse},
we need a very huge size $n$ and a very small $\epsilon$.
Although one may consider that $n=1000000$ is too large to realize,
this size is realizable as follows.
A typical two-universal hash family can be realized by using 
Toeplitz matrix.
This kind two-universal hash family with $n=10^8$ 
was realized efficiently by using a typical personal computer \cite[Appendix B]{H-T}\cite{asai:11}.

%% file: Multi-Random-Number.tex
\section{Secure Uniform Random Number Generation} \Label{section:multi-random-number}

In this section, we investigate the secure random number generation 
with partial information leakage, which is also known as the privacy amplification.
We start this section by showing the problem setting in Section \ref{subsection:multi-random-problem-formulation}.
Then, we review and introduce some single-shot bounds in Section \ref{subsection:multi-random-one-shot}.
We derive non-asymptotic bounds for the Markov chain in Section \ref{subsection:multi-random-finite-markov}.
Then, in Sections \ref{subsection:multi-random-large-deviation} and \ref{subsection:multi-random-mdp},
we show the asymptotic characterization for the large deviation regime and the moderate deviation regime
by using those non-asymptotic bounds. We also derive the second order rate in Section \ref{subsection:multi-random-second-order}.


The results shown in this section are summarized in Table \ref{table:summary:multi-random-number}.
The checkmarks $\checkmark$ indicate that the tight asymptotic bounds (large deviation, moderate deviation, and second order)
 can be obtained from those bounds.
The marks $\checkmark^*$ indicate that the large deviation bound can be derived up to the critical rate. 
The computational complexity "Tail" indicates that the computational complexities of those bounds 
depend on the computational complexities of tail probabilites.
It should be noted that Theorem \ref{theorem:multi-random-finite-markov-assumption-1-direct}
is derived from a special case ($Q_{Y} = P_Y$) of Lemma \ref{lemma:exponential-bound}.
The asymptotically optimal choice is $Q_Y = P_Y^{(1+\theta)}$, which corresponds to \eqref{eq:pa-exponential-bound} of Lemma \ref{lemma:exponential-bound}.
Under Assumption \ref{assumption-Y-marginal-markov}, we can derive the bound of the Markov case only for that
special choice of $Q_Y$, while under Assumption \ref{assumption-memory-through-Y}, we can derive the bound of
the Markov case for the optimal choice of $Q_Y$.
Here, we didn't call several lemmas as theorems although they derive the asymptotic tight bound.
This is because they are not computable form as explained in 
the beginning of Section \ref{section:single-random-number}.

\subsection{Problem Formulation} \Label{subsection:multi-random-problem-formulation}

The privacy amplification is conducted by a function $f:{\cal X} \to \{1,\ldots,M\}$. The security of the generated key is evaluated by
\begin{eqnarray} \Label{eq:definition-security}
\Delta [f] := \frac{1}{2} \| P_{f(X)Y} - P_{\overline{U}} \times P_Y \|_1,
\end{eqnarray}
where $\overline{U}$ is the uniform random variable on $\{1,\ldots,M\}$ and $\| \cdot \|_1$ is the variational distance. 
For notational convenience, we introduce the infimum of the security criterion under the condition that the range size is $M$:
\begin{eqnarray}
\Delta(M) := \inf_{f} \Delta [f].
\end{eqnarray}

When we construct a function for the privacy amplification, 
we often use a two-universal hash family ${\cal F}$ and a random function $F$
on ${\cal F}$. Then, we bound the security criterion averaged over the random function by only using the property of two-universality.
\MH{As explained in Subsection \ref{two-cri},
to take into the practical aspects,
we introduce the worst leaked information:}
\begin{eqnarray} \Label{eq:multi-random-number-worst-u2}
\overline{\Delta}(M) := \sup_F \mathbb{E}[\Delta [F]],
\end{eqnarray}
where the supremum is taken over all two-universal hash families from ${\cal X}$ to $\{1,\ldots,M\}$.
From the definition, we obviously have
\begin{eqnarray}
\Delta(M) \le \overline{\Delta}(M).
\end{eqnarray}

When we consider $n$-fold extension, the security criteria are denoted by
$\Delta(M_n)$ or $\overline{\Delta}(M_n)$. 
As in the single-terminal case, we also introduce the quantities $M(n,\varepsilon)$ and $\overline{M}(n,\varepsilon)$
(cf.~\eqref{eq:definition-single-random-m-epsilon} and \eqref{eq:definition-single-random-barm-epsilon}).

\begin{remark} \Label{remark:alternative-security}
Note that the security definition in \eqref{eq:definition-security} implies the universal composable security criterion \cite{canetti:01, pfitzmann:00}.
A slightly weaker security criterion defined by
\begin{eqnarray}
\inf_{Q_Y} \frac{1}{2} \| P_{f(X)Y} - P_{\overline{U}} \times Q_Y \|_1
\end{eqnarray}
also implies the universal composable security criterion. In fact some literatures employs this kinds
of security criteria \cite{shikata:13, tomamichel:12, koashi:09}. 
Since the triangle inequality 
and the information processing inequality 
$\| Q_Y -  P_Y \|_1 \le
\| P_{\overline{U}} \times Q_Y - P_{f(X)Y} \|_1 $ imply
\begin{align*}
&
\frac{1}{2} \| P_{f(X)Y} - P_{\overline{U}} \times P_Y \|_1 \\
\le &
\frac{1}{2} \| P_{f(X)Y} - P_{\overline{U}} \times Q_Y \|_1 
+
\frac{1}{2} \| P_{\overline{U}} \times Q_Y - P_{\overline{U}} \times P_Y \|_1  \\
=&
\frac{1}{2} \| P_{f(X)Y} - P_{\overline{U}} \times Q_Y \|_1 
+
\frac{1}{2} \| Q_Y -  P_Y \|_1 \\
\le &
\frac{1}{2} \| P_{f(X)Y} - P_{\overline{U}} \times Q_Y \|_1 
+
\frac{1}{2} 
\| P_{\overline{U}} \times Q_Y - P_{f(X)Y} \|_1 ,
\end{align*}
we have
\begin{eqnarray}
\frac{1}{2} \| P_{f(X)Y} - P_{\overline{U}} \times P_Y \|_1 
\le \| P_{f(X)Y} - P_{\overline{U}} \times Q_Y \|_1
\end{eqnarray}
holds for any $Q_Y$. Thus, the two criteria differ only in constant factor, which means that 
the asymptotic behaviors of the large deviation regime and the moderate deviation regime 
are not affected by the choice of the security criteria.

\MH{For the second order regime, the same fact can be shown as follows.}
The achievability part (Lemma \ref{lemma:multi-information-spectrum} \MH{given in Subsection \ref{subsection:multi-random-one-shot}})
can be used without modification since the optimization over $Q_Y$ is already incorporated into the bound.
For the converse part, we need to replace $H_{\min}^\varepsilon(P_{XY}|P_Y)$
with $H_{\min}^\varepsilon(P_{XY}|Q_Y)$ in Lemma \ref{lemma:multi-random-number-monotonicity-converse}
\MH{given in Subsection \ref{subsection:multi-random-one-shot}}.
Then, the converse bound in Lemma \ref{lemma:multi-random-sphere-packing-converse} \MH{given in Subsection \ref{subsection:multi-random-one-shot}}
is modified accordingly, i.e.,
\begin{eqnarray*}
\Delta(M) \ge 
\inf_{Q_Y} \max_{\gamma \ge 0}\left[ P_{XY}\left\{ \log \frac{Q_Y(y)}{P_{XY}(x,y)} < \gamma \right\}\left(1 - \frac{e^\gamma}{M} \right) \right].
\end{eqnarray*}
However, by noting the inequality
\begin{align}
& P_{XY}\left\{ \log \frac{Q_Y(y)}{P_{XY}(x,y)} < \gamma \right\} \nonumber \\
\ge & P_{XY}\left\{ \log \frac{1}{P_{X|Y}(x|y)} < \gamma - \nu \right\} - P_{XY}\left\{ \log \frac{Q_Y(y)}{P_Y(y)} > \nu  \right\} \nonumber \\
\ge & P_{XY}\left\{ \log \frac{1}{P_{X|Y}(x|y)} < \gamma - \nu \right\} - e^{-\nu}
\end{align}
for any $\nu > 0$, 
the choice $Q_Y = P_Y$ turns out to be the optimal choice asymptotically up to $o(\sqrt{n})$.
Thus, the asymptotic behavior of the second order regime is also not affected by the choice of the security criteria.
\end{remark}

When the output size $M$ is too large, 
$\Delta(M)$ is close to $1$ anymore.
In this case, to quantify the performance of the output random number, 
according to Csisz\'{a}r-Narayan \cite{csiszar:04},
we focus on 
the relative entropy between the generated random number and the ideal random number as follows.
\begin{align}
D(P_{f(X)Y} \| P_{\overline{U}}\times P_Y)
=&
\log M - H(f(X) |Y) \nonumber \\
=& I(f(X);Y)+
D(P_{f(X)} \| P_{\overline{U}}).
\end{align}
Since this quantity can be regarded as 
a modification of the mutual information $I(f(X);Y)$,
we call it the modified mutual information.
This quantity is naturally given under axiomatic conditions \cite{hayashi:13}.
Then, we address the following quantities.
\begin{align}
D(M) :=& \inf_f D(P_{f(X)Y} \| P_{\overline{U}}\times P_Y) \\
\overline{D}(M) :=& \sup_F \mathbb{E} [ D(P_{F(X)Y F} \| P_{\overline{U}}\times P_Y)]
\nonumber \\
=& D(P_{F(X)Y F} \| P_{\overline{U}}\times P_Y\times P_F)
\end{align}
where the supremum is taken over all two-universal hash families from ${\cal X}$ to $\{1,\ldots,M\}$.
The reason why we consider such a supremum 
is the same as the case of $\overline{\Delta}(M)$.

\subsection{Single Shot Bounds} \Label{subsection:multi-random-one-shot}

In this section, we review existing single shot bounds, and show a novel converse bound.
For the information measures used below, see Section \ref{section:preparation-multi}.
We also introduce the following information measures. 
For $P_{XY} \in \overline{{\cal P}}({\cal X} \times {\cal Y})$ and $Q_Y \in {\cal P}({\cal Y})$\footnote{Technically, we restrict 
$Q_Y$ to be such that $\rom{supp}(P_Y) \subset \rom{supp}(Q_Y)$.},  let
\begin{eqnarray}
H_{\min}(P_{XY}|Q_Y) := - \log \max_{x,y} \frac{P_{XY}(x,y)}{Q_Y(y)}
\end{eqnarray}
be the conditional $\min$-entropy. Then, for $P_{XY} \in {\cal P}({\cal X}\times{\cal Y})$, let 
\begin{eqnarray}
H_{\min}^\varepsilon(P_{XY}|Q_Y) := \max_{P^\prime_{XY} \in {\cal B}^\varepsilon(P_{XY})} H_{\min}(P^\prime_{XY}|Q_Y)
\end{eqnarray}
and
\begin{eqnarray}
\overline{H}_{\min}^\varepsilon(P_{XY}|Q_Y) := \max_{P^\prime_{XY} \in \overline{{\cal B}}^\varepsilon(P_{XY})} H_{\min}(P^\prime_{XY}|Q_Y)
\end{eqnarray}
be the smooth $\min$-entropy, where 
\begin{align*}
{\cal B}(P_{XY}) :=& \left\{ P^\prime_{XY} \in {\cal P}({\cal X} \times {\cal Y}) : \frac{1}{2} \| P_{XY} - P^\prime_{XY} \|_1 \le \varepsilon \right\}, \\
\overline{{\cal B}}(P_{XY}) :=& \left\{ P^\prime_{XY} \in \overline{{\cal P}}({\cal X} \times {\cal Y}) : \frac{1}{2} \| P_{XY} - P^\prime_{XY} \|_1 \le \varepsilon \right\}.
\end{align*}

By using the two-universal hash family, we can derive the following bound.
\begin{lemma}[\cite{renner:05b}] \Label{lemma:multi-min-entropy}
For any $Q_Y \in {\cal P}({\cal Y})$, we have
\begin{eqnarray*}
\overline{\Delta}(M) \le 2 \varepsilon  + \frac{1}{2} \sqrt{M e^{- \overline{H}_{\min}^\varepsilon(P_{XY}|Q_Y)}}.
\end{eqnarray*}
\end{lemma}

\MH{However, the bound in Lemma \ref{lemma:multi-min-entropy}
cannot be directly calculated in the Markovian chain.
To resolve this problem, we slightly loosen Lemma \ref{lemma:multi-min-entropy} as follows.}
(cf.~\cite[Theorem 23]{hayashi:13} or \cite[Lemma 3]{watanabe:13c}).

\begin{lemma} \Label{lemma:multi-information-spectrum}
For any $Q_Y \in {\cal P}({\cal Y})$, we have
\begin{eqnarray*}
\overline{\Delta}(M) \le \inf_{\gamma \ge 0}\left[ P_{XY}\left\{ \log \frac{Q_Y(y)}{P_{XY}(x,y)} < \gamma \right\} +  \frac{1}{2} \sqrt{\frac{M}{e^\gamma}} \right].
\end{eqnarray*} 
\end{lemma}

We also have the following exponential bound.
\begin{lemma}[\cite{hayashi:10b}] \Label{lemma:exponential-bound}
We have
\begin{align}
& \overline{\Delta}(M) \nonumber \\
\le & 
\min_{Q_Y \in {\cal P}({\cal Y})} \inf_{0 \le \theta \le 1} \frac{3}{2} M^{\frac{\theta}{1+\theta}} e^{- \frac{\theta}{1+\theta} H_{1+\theta}(P_{XY}|Q_Y)} 
 \Label{eq:pa-exponential-bound-alternative-2} \\
=& \inf_{0 \le \theta \le 1} \frac{3}{2} M^{\frac{\theta}{1+\theta}} e^{- \frac{\theta}{1+\theta} H_{1+\theta}^\uparrow(X|Y)}.
\Label{eq:pa-exponential-bound}
\end{align}
\end{lemma}


For the converse bound, the following is known\footnote{See also \cite{watanabe:13c} for a proof that is specialized for the classical case.}.
\begin{lemma}[\cite{renner:05b}] 
\Label{lemma:multi-random-number-monotonicity-converse}
We have
\begin{eqnarray}
\Delta(M) \ge \min_{H_{\min}^\varepsilon(P_{XY}|P_Y) \ge \log M} \varepsilon.
\end{eqnarray}
\end{lemma}

\MH{Similar to Lemma \ref{lemma:multi-min-entropy}, 
the bound in Lemma \ref{lemma:multi-random-number-monotonicity-converse}
cannot be directly calculated in the Markovian chain.
To resolve this problem, we slightly loosen Lemma \ref{lemma:multi-random-number-monotonicity-converse} as follows.}

\begin{lemma} \Label{lemma:multi-random-sphere-packing-converse}
We have
\begin{eqnarray}
\Delta(M) \ge \max_{\gamma \ge 0}\left[ P_{XY}\left\{ \log \frac{1}{P_{X|Y}(x|y)} < \gamma \right\}\left(1 - \frac{e^\gamma}{M} \right) \right].
\end{eqnarray}
\end{lemma}
\begin{proof}
The proof is exactly the same as Lemma \ref{lemma:single-random-sphere-packing-converse}.
\end{proof}

\MH{Although Lemma \ref{lemma:multi-random-sphere-packing-converse}
is useful for the large deviation regime and the moderate deviation regime,
it is not useful for the second order regime.
To resolve this problem, 
we loosen Lemma \ref{lemma:multi-random-sphere-packing-converse} as follows.}

\begin{lemma} \Label{lemma:multi-random-number-converse}
We have
\begin{eqnarray}
\Delta(M) \ge \sup_{\gamma \ge 0}\left[ P_{XY}\left\{ \log \frac{1}{P_{X|Y}(x|y)} < \gamma \right\} - \frac{e^{\gamma}}{M} \right].
\end{eqnarray}
\end{lemma}


Furthermore, by using a property of the strong universal hash family, 
we can derive the following converse \MH{as a generalization of Lemma \ref{lemma:strong-universal-bound}.}

\begin{lemma} \Label{lemma:strong-universal-bound-multi}
For $\{ \Omega_y \}_{y\in{\cal Y}}$ such that $|\Omega_y| \le N \le M$ for every $y \in {\cal Y}$, let $\Omega = \cup_{y\in{\cal Y}} \Omega_y \times \{y\}$.
Then, we have
\begin{eqnarray}
\overline{\Delta}(M) \ge \left(1 - \frac{N}{M} \right)^2 P_{XY}(\Omega).
\end{eqnarray}
\end{lemma}
\begin{proof}
We apply Lemma \ref{lemma:strong-universal-bound} to each $P_{X|Y}(\cdot|y)$ and take average over $y$. Then, we can derive the lemma
since $|\Omega_y| \le N$ by the assumption. 
\end{proof}

\MH{Similar to Lemmas \ref{lemma:multi-min-entropy} and \ref{lemma:multi-random-number-monotonicity-converse}, 
the bound in Lemma \ref{lemma:strong-universal-bound-multi}
cannot be directly calculated in the Markovian chain.
To resolve this problem, we slightly loosen Lemma \ref{lemma:strong-universal-bound-multi} as follows.}

\begin{lemma} \Label{lemma:strong-universal-bound-multi-tail}
For any $0 < \nu < 1$, we have
\begin{eqnarray} \Label{eq:pa-strong-hash-one-shot-multi}
\overline{\Delta}(M) \ge (1-\nu)^2 P_{XY}\left\{ \log \frac{P_Y^{(1+\theta(a(R )))}(y)}{P_{XY}(x,y)} \le a(R ) \right\},
\end{eqnarray}
where $R = \log (M\nu)$, and $\theta(a)$ and $a(R )$ 
are the inverse functions $\theta^{\uparrow}(a)$ and $a^{\uparrow}(R )$ defined by \eqref{eq:definition-rho-inverse-Gallager-one-shot} 
and \eqref{eq:definition-a-inverse-Gallager-one-shot} respectively.
\end{lemma}
\begin{proof}
See Appendix \ref{appendix:lemma:strong-universal-bound-multi-tail}.
\end{proof}


To derive a converse bound for ${\Delta}(M)$ based on the conditional R\'{e}nyi entropy,
we substitute the formula in Proposition \ref{theorem:one-shot-tail-converse-2}
in Appendix \ref{Appendix:preparation} into the bound in  Lemma \ref{lemma:multi-random-sphere-packing-converse} 
for $a = \gamma = \log (M/2)$.
So, we have the following.
\begin{theorem} \Label{theorem:multi-random-exponential-converse}
We have
\begin{align}
& \lefteqn{ - \log \Delta(M) } \nonumber \\
\le & \inf_{s > 0 \atop \tilde{\theta} > \theta(a)} \frac{1}{s}
\Bigg[
 (1+s) \tilde{\theta} \bigg( H_{1+\tilde{\theta}}^\downarrow(X|Y) - H_{1+(1+s)\tilde{\theta}}^\downarrow(X|Y) \bigg) \nonumber \\
& \hspace{3ex} - (1+s) \log \bigg(1 
\nonumber \\
& \hspace{9ex} - e^{ (\theta(a)-\tilde{\theta}) a - \theta(a) H_{1+\theta(a)}^\downarrow(X|Y) 
  + \tilde{\theta} H_{1+\tilde{\theta}}^\downarrow(X|Y) } \bigg)
\Bigg] 
\nonumber \\
& + \log 2,
\end{align}
where $a = \log (M/2)$, 
and $\theta(a)$ is the inverse function 
$\theta^\downarrow(a )$ defined by \eqref{eq:definition-inverse-theta-multi-one-shot-2}.
\end{theorem}

\begin{proof}
Theorem \ref{theorem:multi-random-exponential-converse} can be shown by the same way as Theorem \ref{theorem:single-random-sphere-packing-converse-finite-markov}
with replacing the role of Lemma \ref{lemma:single-random-sphere-packing-converse}
 by Lemma \ref{theorem:multi-random-exponential-converse}.
\end{proof}

To derive a converse bound for $\overline{\Delta}(M)$ based on the conditional R\'{e}nyi entropy,
we substitute the formula in Proposition \ref{theorem:one-shot-tail-converse-2}
in Appendix \ref{Appendix:preparation}
into the bound in  Lemma \ref{theorem:multi-random-exponential-converse-stong-universal} 
for $\nu = \frac{1}{2}$.
So, we have the following.

\begin{theorem} \Label{theorem:multi-random-exponential-converse-stong-universal}
We have
\begin{align}
\lefteqn{ - \log \overline{\Delta}(M) } \nonumber \\
\le& \inf_{s > 0 \atop \tilde{\theta} >\theta(a(R ))} \frac{1}{s}
\bigg[
 (1+s) \tilde{\theta} \bigg(
H_{1+\tilde{\theta},1+\theta(a(R ))}(X|Y) \nonumber \\
 & - H_{1+(1+s)\tilde{\theta},1+\theta(a(R ))}(X|Y) \bigg) 
 - (1+s) \log \left(1 - e^{ C_{2,n}} \right)
\bigg]  \nonumber \\
&+ 2\log 2,
\end{align}
where $R = \log (M/2)$,
\begin{align*}
C_{2,n}:=&
[\theta(a(R )) - \tilde{\theta}] a(R ) - \theta(a(R )) H_{1+\theta(a(R ))}^\uparrow(X|Y) 
\\
&   + \tilde{\theta} H_{1+\tilde{\theta},1+\theta(a(R ))}(X|Y)  ,
\end{align*}
and
$\theta(a)$ and $a(R )$ are the inverse functions 
$\theta^\uparrow(a )$ and $a^\uparrow(R )$
defined by \eqref{eq:definition-rho-inverse-Gallager-one-shot} 
and \eqref{eq:definition-a-inverse-Gallager-one-shot} respectively. 
\end{theorem}

\begin{proof}
Theorem \ref{theorem:multi-random-exponential-converse-stong-universal}
 can be shown by the same way as Theorem \ref{theorem:random-number-strong-universal-exponential-converse}
with replacing the role of Lemma \ref{lemma:strong-universal-bound-tail-probability}
 by Lemma \ref{lemma:strong-universal-bound-multi-tail}.
\end{proof}

Finally, we address \MH{the modified mutual information rate (MMIR).}
As the direct part, we have the following theorem.

\begin{theorem}\Label{Th11-25b-1}
The maximum modified mutual information $\overline{D}(M) $ 
among two-universal hash family
is bounded as
\begin{align}
\overline{D}(M) 
\le \frac{1}{\theta}\log ( 1+ M^\theta e^{-\theta H_{1+\theta}(X|Y)}).
\end{align}
\end{theorem}

\begin{proof}
Lemma 10 of \cite{matsumoto-hayashi:11} shows that
any two-universal hash function $F$ satisfies 
the relation
\begin{align}
\mathbb{E}
 (M^\theta e^{-\theta H_{1+\theta}(F(X|Y))})
\le 1+ M^\theta e^{-\theta H_{1+\theta}(X|Y)},
\end{align}
which implies that
$
\mathbb{E} [\log M - H(F(X|Y))]
\le 
\mathbb{E} [\log M - H_{1+s}(F(X)|Y)]
\le
\frac{1}{s}\log 
\mathbb{E}
 (M^s e^{-s H_{1+s}(F(X)|Y)})
\le
\frac{1}{s}\log ( 1+ M^s e^{-s H_{1+s}(X|Y)})$.
\end{proof}

As the converse part, we have the following theorem.
\begin{proposition}\Label{Th11-25b-2}
\begin{align}
D(M) \ge \log M - H(P_X)
\Label{11-25b-1}
\end{align}
\end{proposition}

\begin{proof}
Inequality \eqref{11-25b-1} follows from the inequality
$H(X|Y) \ge H(f(X)|Y)$.
\end{proof}

\subsection{\MH{Finite-Length} Bounds for Markov Source} \Label{subsection:multi-random-finite-markov}

\MH{Since we assume the irreducibility for the transition matrix describing the Markovian chain,
the following bounds hold with any initial distribution.}
\MH{To lower bound $- \log \overline{\Delta}(M_n)$ by the lower conditional R\'{e}nyi entropy of transition matrix,
we substitute the formula 
for the lower conditional R\'{e}nyi entropy given in Lemma \ref{lemma:mult-terminal-finite-evaluation-down-conditional-renyi}
into the bound in 
Lemma \ref{lemma:exponential-bound} for $Q_{Y^n} = P_{Y^n}$,
we have the following achievability bound.}

\begin{theorem} \Label{theorem:multi-random-finite-markov-assumption-1-direct}
Suppose that a transition matrix $W$ satisfies Assumption \ref{assumption-Y-marginal-markov}.
Let $R := \frac{1}{n}\log M_n$. Then 
we have
\begin{align}
& - \log \overline{\Delta}(M_n) \nonumber \\
\ge & \sup_{0 \le \theta \le 1} \frac{-\theta n R + (n-1) \theta H_{1+\theta}^{\downarrow,W}(X|Y) + \underline{\delta}(\theta)}{1+\theta} - \log (3/2).
 \Label{eq:multi-random-number-finite-markov-direct-assumption-1} 
\end{align}
\end{theorem}

\MH{To upper bound $- \log {\Delta}(M_n)$ by the lower conditional R\'{e}nyi entropy of transition matrix,
we substitute the formula 
for the tail probability given in 
and Proposition \ref{proposition:general-markov-tail-converse}
with $a=R$
into the bound in 
Lemma \ref{lemma:multi-random-sphere-packing-converse} 
with $\gamma=nR$,
we have the following converse bound.}

\begin{theorem} \Label{theorem:multi-random-finite-markov-assumption-1-converse}
Suppose that a transition matrix $W$ satisfies Assumption \ref{assumption-Y-marginal-markov}.
Let $R := \frac{1}{n}\log (M_n/2) $. For any $\underline{a} < R < H^W(X|Y)$, we have
\begin{align}
\lefteqn{ - \log \Delta(M_n) } \nonumber \\
\le &
\inf_{s > 0 \atop \tilde{\theta} > \theta(a)} \frac{1}{s}\Bigg[
 (n-1) (1+s) \tilde{\theta} 
\bigg( H_{1+\tilde{\theta}}^{\downarrow,W}(X|Y) 
\nonumber \\
 &
- H_{1+(1+s)\tilde{\theta}}^{\downarrow,W}(X|Y) \bigg) + \delta_1 
- (1+s) \log \left( 1 - e^{C_{3,n}} \right)
\Bigg]  
\nonumber \\
 &
+ \log 2,\Label{11-22-17}
\end{align}
where  $\theta(a)$ is the inverse function 
$\theta^\downarrow(a )$
defined by \eqref{eq:definition-theta-inverse-multi-markov}, 
and
\begin{align}
C_{3,n}:=&
(n-1) \bigg( (\theta(R ) - \tilde{\theta}) R - \theta(R ) H_{1+\theta(R )}^{\downarrow,W}(X|Y)
\nonumber \\   
&\hspace{10ex}+ \tilde{\theta} H_{1+\tilde{\theta}}^{\downarrow,W}(X|Y)  \bigg) + \delta_2, \\
\delta_1 :=& (1+s) \overline{\delta}(\tilde{\theta}) - \underline{\delta}((1+s)\tilde{\theta}), \\
\delta_2 :=& (\theta(R ) - \tilde{\theta}) R - \underline{\delta}(\theta(R )) + \overline{\delta}(\tilde{\theta}).
\end{align}
\end{theorem}

\begin{proof}
Theorem \ref{theorem:multi-random-finite-markov-assumption-1-converse} can be shown by the same way as Theorem \ref{theorem:single-random-sphere-packing-converse-finite-markov}
with replacing the roles of Lemma \ref{lemma:single-random-sphere-packing-converse}
and Proposition \ref{theorem:one-shot-tail-converse-2} in Appendix \ref{Appendix:preparation}
 by Lemma \ref{theorem:multi-random-exponential-converse}
and Proposition \ref{proposition:general-markov-tail-converse}.
\end{proof}

Next, we derive tighter bounds under Assumption \ref{assumption-memory-through-Y}. 
To lower bound $- \log \overline{\Delta}(M_n)$ by the upper conditional R\'{e}nyi entropy of transition matrix,
we substitute the formula 
for the upper conditional R\'{e}nyi entropy given in Lemma \ref{lemma:multi-terminal-finite-evaluation-upper-conditional-renyi}
into the bound in 
Lemma \ref{lemma:exponential-bound},
we have the following achievability bound.

\begin{theorem} \Label{theorem:multi-random-finite-markov-assumption-2-direct}
Suppose that a transition matrix $W$ satisfies Assumption \ref{assumption-memory-through-Y}. 
Let $R := \frac{1}{n} \log M_n$. Then we have
\begin{align}
& - \log \overline{\Delta}(M_n) \nonumber \\
\ge & 
\sup_{0 \le \theta \le 1} \frac{-\theta n R + (n-1) \theta H_{1+\theta}^{\uparrow,W}(X|Y) }{1+\theta} + \underline{\xi}(\theta) - \log (3/2). 
\Label{11-22-5}
\end{align}
\end{theorem}


To upper bound $- \log {\Delta}(M_n)$ by the upper conditional R\'{e}nyi entropy of transition matrix,
we substitute the formula 
for the tail probability given in 
and Proposition \ref{theorem:one-shot-tail-converse-2} in Appendix \ref{Appendix:preparation}
into the bound in 
Lemma \ref{lemma:strong-universal-bound-multi}\footnote{We cannot apply
Proposition \ref{proposition:general-markov-tail-converse} here since we cannot apply Lemma \ref{lemma:finite-evaluation-of-cgf}
for $\phi(\tilde{\rho};P_{X^nY^n}|Q_{Y^n}^{(1-\rho)})$. Instead, we need to apply 
Lemma \ref{lemma:multi-terminal-finite-evaluation-two-parameter-conditional-renyi}.}, 
we have the following converse bound.

\begin{theorem} \Label{theorem:multi-random-finite-markov-assumption-2-converse}
Suppose that a transition matrix $W$ satisfies Assumption \ref{assumption-memory-through-Y}. 
Let $R$ be such that
\begin{align}
&(n-1) R + \Bigg( (1+\theta(a(R ))) (a(R ) -  \underline{\xi}(\theta(a(R )))) \Bigg)\nonumber \\
=& \log(M_n/2).
\end{align}
If $R(\underline{a}) < R < H^W(X|Y)$, then we have
\begin{align}
\lefteqn{ - \log \overline{\Delta}(M_n) } \nonumber \\
\le & 
 \inf_{s > 0 \atop \tilde{\theta} > \theta(a(R ))}\frac{1}{s}\Bigg[
 (n-1) 
(1+s) \tilde{\theta} \bigg( H_{1+\tilde{\theta},1+\theta(a(R ))}^W(X|Y) 
\nonumber \\
&\hspace{11ex}- H_{1+(1+s)\tilde{\theta},1+\theta(a(R ))}^W(X|Y) \bigg) 
+ \delta_1 \nonumber \\
 &\hspace{11ex}   - (1+s) \log \left(1 - e^{C_{4,n}} \right)
 \Bigg]  + 2\log 2, \Label{11-22-9}
\end{align}
where 
$\theta(a)$ and $a(R )$ are the inverse functions 
$\theta^\uparrow(a )$ and $a^\uparrow(R )$
defined by \eqref{eq:definition-theta-inverse-markov-optimal-Q} 
and \eqref{eq:definition-a-inverse-markov-optimal-Q} respectively, 
\begin{align}
C_{4,n}:=&
(n-1) \Big[ (\theta (a (R ) ) - \tilde{\theta}) (a(R )) 
\nonumber \\
&\hspace{8ex} - \theta (a (R  ) ) H_{1+\theta(a(R ))}^{\uparrow,W}(X|Y) 
\nonumber \\
&\hspace{11ex}+ \tilde{\theta} H_{1+\tilde{\theta},1+\theta(a(R ))}^W(X|Y)  
\Big] + \delta_2
\\
\delta_1 :=& (1+s) \overline{\zeta}(\tilde{\theta},\theta(a(R ))) - \underline{\zeta}((1+s)\tilde{\theta}, \theta(a(R ))) , \\
\delta_2 :=& (\theta(a(R )) - \tilde{\theta}) (a(R ) ) - \underline{\zeta}(\theta(a(R )), \theta(a(R ))) 
\nonumber \\ &
+ \overline{\zeta}(\tilde{\theta},\theta(a(R ))).
\end{align}
\end{theorem}
\begin{proof}
See Appendix \ref{appendix:theorem:multi-random-finite-markov-assumption-2-converse}.
\end{proof}

\MH{We derive finite-length bounds for 
modified mutual information rate under Assumption 
\ref{assumption-Y-marginal-markov} 
by substituting 
the formula for the lower conditional R\'{e}nyi entropy given in Lemma 
\ref{lemma:mult-terminal-finite-evaluation-down-conditional-renyi}
into the bound in Theorem \ref{Th11-25b-1}.}

\begin{theorem}\Label{Th11-25b-3}
When $R-H_{1+\theta}^{\downarrow,W}(X|Y)\ge 0$, 
for $\theta \in [0,1]$, we have 
\begin{align}
\overline{D}(e^{nR}) 
\le 
n R-(n-1)H^{\downarrow,W}_{1+\theta}(X|Y)) + 
\frac{1}{\theta}(\log 2 - \underline{\delta}(\theta) )).
\end{align}
\end{theorem}

\begin{proof}
Theorem \ref{Th11-25b-3} can be shown as the same way as
Theorem \ref{Th11-25-3} by replacing 
$H^{W}_{1+\theta}(X)$ and Theorem \ref{Th11-25-1}
by $H^{\downarrow,W}_{1+\theta}(X|Y)$ and Theorem \ref{Th11-25b-1},
respectively.
\end{proof}

\MH{To lower bound $\overline{D}(e^{nR}) $ 
by the lower conditional R\'{e}nyi entropy of transition matrix,
we substitute the other formula 
for the lower conditional R\'{e}nyi entropy given in 
Lemma \ref{lemma:mult-terminal-finite-evaluation-down-conditional-renyi} 
into the bound in Proposition \ref{Th11-25b-2},
we have the following bound.}

\begin{theorem}\Label{Th11-25b-4}
For $\theta \in [0,1]$, we have 
\begin{align}
D(e^{nR}) \ge n R - (n-1) H^{\downarrow,W}_{1-\theta}(X)
+ \frac{\underline{\delta}(-\theta)}{\theta}
\Label{11-25b-4}
\end{align}
\end{theorem}

\begin{proof}
Theorem \ref{Th11-25b-4} can be shown as the same way as
Theorem \ref{Th11-25-4} by replacing 
$H^{W}_{1-\theta}(X)$ and Proposition \ref{Th11-25-2}
by $H^{\downarrow,W}_{1-\theta}(X|Y)$ and Proposition \ref{Th11-25b-2},
respectively.
\end{proof}

\subsection{Large Deviation} \Label{subsection:multi-random-large-deviation}

\MH{We can show the following theorem in the same way as Theorem \ref{T11-22-1}
by taking the limit in
Theorems \ref{theorem:multi-random-finite-markov-assumption-1-direct} and 
\ref{theorem:multi-random-finite-markov-assumption-1-converse}
with use of Lemma \ref{L10}.}

\begin{theorem}
Suppose that a transition matrix $W$ satisfies Assumption \ref{assumption-Y-marginal-markov}.
For $R < H^W(X|Y)$, we have
\begin{eqnarray}
\liminf_{n\to\infty} - \frac{1}{n} \log \overline{\Delta}\left(e^{nR} \right) 
 \ge \sup_{0 \le \theta \le 1} \frac{- \theta R + \theta H_{1+\theta}^{\downarrow,W}(X|Y)}{1 + \theta}.
\end{eqnarray}
On the other hand, for $\underline{a} < R < H^W(X|Y)$, we have
\begin{align}
\lefteqn{\limsup_{n\to\infty} - \frac{1}{n} \log \Delta\left(e^{nR} \right) }
\nonumber \\
\le 
& - \theta(R ) R + \theta(R ) H_{1+\theta(R )}^{\downarrow,W}(X|Y)
\Label{11-15-1} \\
=&
\sup_{0 \le \theta } 
- \theta R + \theta H_{1+\theta}^{\downarrow,W}(X|Y),
\Label{11-27-9}
\end{align}
where  $\theta(a)$ is the inverse function 
$\theta^\downarrow(a )$
defined by \eqref{eq:definition-theta-inverse-multi-markov}.
\end{theorem}

Under Assumption \ref{assumption-memory-through-Y}, 
\MH{taking the limit in} Theorems \ref{theorem:multi-random-finite-markov-assumption-2-direct}
and \ref{theorem:multi-random-finite-markov-assumption-2-converse}, we have the following tighter bound.
\begin{theorem}\Label{th30}
Suppose that a transition matrix $W$ satisfies Assumption \ref{assumption-memory-through-Y}. 
For $R < H^W(X|Y)$, we have
\begin{eqnarray} \Label{eq:multi-random-assumption-2-ldp-direct}
\liminf_{n\to\infty} - \frac{1}{n} \log \overline{\Delta}\left( e^{nR}\right) 
 \ge \sup_{0 \le \theta \le 1} \frac{- \theta R + \theta H_{1+\theta}^{\uparrow,W}(X|Y)}{1+\theta}.\Label{11-27-11}
\end{eqnarray}
On the other hand, for $R(\underline{a}) < R < H^W(X|Y)$, we have
\begin{align}
\lefteqn{\limsup_{n\to\infty} - \frac{1}{n} \log \overline{\Delta}\left( e^{nR}\right) }
\nonumber \\
\le & - \theta(a(R )) a(R ) + \theta(a(R )) H_{1+\theta(a(R ))}^{\uparrow,W}(X|Y) 
\Label{11-15-2} \\
=&
\sup_{0 \le \theta } \frac{- \theta R + \theta H_{1+\theta}^{\uparrow,W}(X|Y)}{1+\theta},\Label{11-27-10}
\end{align}
where 
$\theta(a)$ and $a(R )$ are the inverse functions 
$\theta^\uparrow(a )$ and $a^\uparrow(R )$
defined by \eqref{eq:definition-theta-inverse-markov-optimal-Q} 
and \eqref{eq:definition-a-inverse-markov-optimal-Q} respectively.
\end{theorem}   

Due to Lemma \ref{L9},
the lower bound \eqref{11-27-11}
and the upper bound \eqref{11-15-2}
coincide when 
$R$ is not less than the critical rate $R_{\mathrm{cr}}$.

\begin{proof}
\eqref{11-22-5} in Theorem \ref{theorem:multi-random-finite-markov-assumption-2-direct} yields \eqref{eq:multi-random-assumption-2-ldp-direct}.
Lemma \ref{L9} guarantees \eqref{11-27-10}.
So, we will prove \eqref{11-15-2}.

We fix $s > 0$ and $\tilde{\theta} > \theta(a(R ))$.
Then, \eqref{11-22-9} implies that
\begin{align}
\lefteqn{ \lim_{n \to \infty} -\frac{1}{n}\log \overline{\Delta}(M_n) }
\nonumber \\
  \le & 
\!\frac{1\!+\!s}{s}\!
\tilde{\theta} \Biggl( H_{1+\tilde{\theta},1+\theta(a(R ))}^W(X|Y) 
\!- \! H_{1+(1+s)\tilde{\theta},1+\theta(a(R ))}^W(X|Y) \Biggr)
\Label{11-22-11}
\end{align}
Similar to \eqref{11-21-10},
taking the limits $s \to 0$ and $\tilde{\theta} \to \theta(a(R ))$, we have
\begin{align}
& \frac{1+s}{s} 
\tilde{\theta} \Bigg( H_{1+\tilde{\theta},1+\theta(a(R ))}^W(X|Y) 
\nonumber \\
&\hspace{19ex} - H_{1+(1+s)\tilde{\theta},1+\theta(a(R ))}^W(X|Y) \Bigg) \nonumber \\
\to &
- 
\tilde{\theta}
\frac{d {\theta} H_{1+{\theta},1+\theta(a(R ))}^W(X|Y) }{d\theta} \biggl|_{\theta=\tilde{\theta}}
\nonumber \\
&\hspace{15ex}+ \tilde{\theta}H_{1+\tilde{\theta},1+\theta(a(R ))}^W(X|Y) 
\quad \hbox{(as $s \to 0$)} \nonumber\\
\to &
- 
\theta(a(R ))
\frac{d {\theta} H_{1+{\theta},1+\theta(a(R ))}^W(X|Y) }{d\theta} \biggl|_{\theta=\theta(a(R ))}
\nonumber \\
&\hspace{8ex}+ \theta(a(R )) H_{1+\theta(a(R ))}^{\uparrow,W}(X|Y) 
\quad \hbox{(as $\tilde{\theta} \to \theta(a(R ))$)} \nonumber\\
\stackrel{(a)}{=} & 
\theta(a(R )) a + \theta(a(R )) H_{1+\theta(a(R ))}^{\uparrow,W}(X|Y)
\Label{11-22-10}.
\end{align}
where $(a)$ follows from \eqref{eq:definition-theta-inverse-markov-optimal-Q}.
Hence, \eqref{11-22-10} and \eqref{11-22-11} imply that
\begin{align}
 \lim_{n \to \infty} -\frac{1}{n}\log \overline{\Delta}(M_n) 
\le 
\theta(a(R )) a + \theta(a(R )) H_{1+\theta(a(R ))}^{\uparrow,W}(X|Y),
\Label{11-22-12}
\end{align}
which implies \eqref{11-15-2}.
\end{proof}

\subsection{Moderate Deviation} \Label{subsection:multi-random-mdp}

\MH{Taking the limit with $R=H^W(X|Y) - n^{-t}\delta$ in}
Theorem \ref{theorem:multi-random-finite-markov-assumption-1-direct} and 
Theorem \ref{theorem:multi-random-finite-markov-assumption-1-converse}, we have the following.
\begin{theorem}
Suppose that a transition matrix $W$ satisfies Assumption \ref{assumption-Y-marginal-markov}.
For arbitrary $t \in (0,1/2)$ and $\delta > 0$, we have
\begin{align}
\lefteqn{\lim_{n\to\infty} - \frac{1}{n^{1-2t}} \log \Delta\left(e^{nH^W(X|Y) - n^{1-t}\delta} \right)}\nonumber \\
=& \lim_{n\to\infty} - \frac{1}{n^{1-2t}} \log \overline{\Delta}\left(e^{nH^W(X|Y) - n^{1-t}\delta} \right) 
= \frac{\delta^2}{2 \san{V}^W(X|Y)}.
\end{align}
\end{theorem}

\begin{proof}
This theorem can be shown by the same way as Theorem \ref{T11-22-10} by
replacing \eqref{eq:single-random-finite-markov-direct-1-1} 
and \eqref{11-21-7}
by \eqref{eq:multi-random-number-finite-markov-direct-assumption-1}
and \eqref{11-22-17}, respectively.
\end{proof}

\subsection{Second Order} \Label{subsection:multi-random-second-order}

By applying the central limit theorem 
to Lemmas \ref{lemma:multi-information-spectrum} and \ref{lemma:multi-random-number-converse}, 
and by using Theorem \ref{theorem:multi-markov-variance}, we have the following.
\begin{theorem}
Suppose that a transition matrix $W$ satisfies Assumption \ref{assumption-Y-marginal-markov}.
For arbitrary $\varepsilon \in (0,1)$, we have
\begin{align}
&\lim_{n\to\infty} \frac{\log M(n,\varepsilon) - nH^W(X|Y)}{\sqrt{n}} \nonumber \\
=& \lim_{n\to\infty} \frac{\log \overline{M}(n,\varepsilon) - nH^W(X|Y)}{\sqrt{n}} \nonumber \\
=& \sqrt{\san{V}^W(X|Y)} \Phi^{-1}(\varepsilon).
\end{align}
\end{theorem}

\begin{proof}
The central limit theorem for Markovian process \cite{kontoyiannis:03,Jones2004,Meyn1994} \cite[Corollary 6.2.]{hayashi-watanabe:13b} guarantees that
the random variable $(\log P_{X^n| Y^n}(X^n|Y^n) -n H^W (X|Y))/\sqrt{n}$ 
asymptotically obeys the normal distribution with the average $0$ and the variance $\san{V}^W(X|Y)$.
This theorem can be shown by the same way as Theorem \ref{T11-22-14} by
replacing the roles of Lemmas \ref{lemma:single-random-number-leftover-loose} and \ref{lemma:single-random-number-converse}
by those of 
Lemmas \ref{lemma:multi-information-spectrum} and \ref{lemma:multi-random-number-converse} with $Q_Y=P_Y$, 
respectively.
\end{proof}

\subsection{\MH{Modified Mutual Information Rate (MMIR)}}\Label{Equivocation Rate-2} 
Taking the limit in Theorems \ref{Th11-25b-3} and \ref{Th11-25b-4}, 
we have the following.
\begin{theorem}\Label{Th11-25b-5}
Suppose that a transition matrix $W$ satisfies Assumption \ref{assumption-Y-marginal-markov}.
\MH{The modified mutual information rate (MMIR) is asymptotically calculated as}
\begin{align}
\lim_{n \to \infty} \frac{1}{n}D(e^{nR}) = 
\lim_{n \to \infty} \frac{1}{n}\overline{D}(e^{nR}) = 
[R-H^{W}(X|Y)]_+  .
\end{align}
\end{theorem}
\begin{proof}
Theorem \ref{Th11-25b-5}
can be shown as the same way as Theorem \ref{Th11-25-5}.
\end{proof}

%% file: conclusion.tex
\section{Discussion and Conclusion}

In this paper, we have derived the non-asymptotic bounds on
the uniform random number generation with/without information leakage
for the Markovian case.
\MH{In these bounds,
the difference between 
${\Delta}(M)$ and $\overline{\Delta}(M)$
is asymptotically negligible 
at least in the moderate deviation regime and the second order regime.
The same relation holds between $D(M)$ and $\overline{D}(M)$. 
Hence, we can conclude that it is enough to employ 
any two-universal hash function even for the Markovian case.}

Here, 
\MH{to discuss the practical importance of non-asymptotic results, 
we shall remark a difference 
of the uniform random number generation from channel and source coding.}
When we construct a practical system, we need to consider two issues:
\begin{itemize}
\item How to {\em quantitatively} guarantee the performance,
\item How to implement the system efficiently.
\end{itemize}
The uniform random number generation do not have to care about decoding complexity
although the coding problems requires decoding, which requires huge amount of calculation complexity.
Furthermore, it is also known that universal$_2$ hash functions can be 
constructed by combination of Toeplitz matrix and the identity matrix.
This construction has small amount of complexity and was implemented in a real demonstration \cite{asai:11}.
Hence, our non-asymptotic results can be directly used as a performance guarantee of a practical system \MH{even when the source distribution has a memory.}

Recently, Tsurumaru et al \cite{tsurumaru:11} proposed a new class of hash functions, so called 
$\varepsilon$-almost dual universal hash functions.
Then, the recent paper \cite{H-T} invented more efficient hash functions with less random seeds, which belong to $\varepsilon$-almost dual universal hash functions.
Hence, it is needed to extend our result to $\varepsilon$-almost dual universal hash functions.
Fortunately, another recent paper \cite{hayashi:13} has already shown similar results with 
$\varepsilon$-almost dual universal hash functions in the i.i.d. case.
So, it is not so difficult to extend the results in \cite{hayashi:13} to the Markovian case.

\MH{In this paper, we have assumed that the transition matrix describing the Markovian chain is irreducible.
When the transition matrix has several irreducible components,
we need to consider the mixture distribution among the possible irreducible components,
which is defined by the initial distribution.
As discussed in \cite[Theorem 1]{RAL:01}, 
in the finite state space,
the asymptotic behavior of the (conditional) R\'{e}nyi entropy
is characterized by the maximum (conditional) R\'{e}nyi entropy
among the possible irreducible components, which depend on the initial distribution.
Hence, 
for large deviation and moderate deviation, 
the exponential decreasing rate of the leaked information
can be evaluated by the minimum rate among the possible irreducible components.
On the other hand, in the case of the mixture of the i.i.d. case, 
when we fix the first and second orders of the coding rate,
the limit of the decoding error probability 
is given by the stochastic mixture of the Gaussian distributions corresponding to the i.i.d. sources \cite{nomura:13}.
So, for the second order analysis for the Markovian case,
we can expect the similar characterization 
by using the stochastic mixture of the Gaussian distributions corresponding to the irreducible components.
Such an analysis is remained for a future study.}

%% file: Appendix-Preparation.tex
\section{Tail probability} \Label{Appendix:preparation}

In converse proofs, we use some techniques to bound tail probabilities in \cite{hayashi-watanabe:13,hayashi-watanabe:13b}.
For this purpose, we need to translate some terminologies in statistics into terminologies in
information theory.
In this appendix, we introduce some terminologies and bounds from \cite{hayashi-watanabe:13,hayashi-watanabe:13b}.
For proofs, see \cite{hayashi-watanabe:13,hayashi-watanabe:13b}.

\subsection{Single-Shot Setting}

Let $Z$ be a real valued random variable with distribution $P$. Let 
\begin{eqnarray} \Label{eq:definition-single-shot-cgf}
\phi(\rho) := \log \mathsf{E}\left[ e^{\rho Z}\right] 
= \log \sum_z P(z) e^{\rho \MH{z}}
\end{eqnarray}
be the cumulant generating function (CGF). 
Let us introduce an exponential family
\begin{eqnarray} \Label{eq:one-shot-tail-exponential-family}
P_\rho(z) := P(z)e^{\rho z - \phi(\rho)}.
\end{eqnarray}
By differentiating the CGF, we find that 
\begin{eqnarray}
\phi^\prime(\rho) = \mathsf{E}_\rho[Z] 
:= \sum_z P_\rho(z) z. 
\end{eqnarray}
We also find that 
\begin{eqnarray} \Label{eq:one-shot-tail-second-derivative}
\phi^{\prime\prime}(\rho) = \sum_z P_\rho(z) \left( z - \mathsf{E}_\rho[Z] \right)^2.
\end{eqnarray}
We assume that $Z$ is not constant. Then, \eqref{eq:one-shot-tail-second-derivative} implies that $\phi(\rho)$ is a strict convex function
and $\phi^\prime(\rho)$ is monotonically increasing. Thus, we can define the inverse function $\rho(a)$ of $\phi^\prime(\rho)$
by 
\begin{eqnarray} \Label{eq:definition-cgf-single-inverse-function}
\phi^\prime(\rho(a)) = a.
\end{eqnarray}

Let 
\begin{eqnarray}
D_{1+s}(P\|Q) := \frac{1}{s} \log \sum_z P(z)^{1+s} Q(z)^{-s}
\end{eqnarray}
be the R\'enyi divergence. Then, we have the following relation:
\begin{align} \Label{eq:relation-renyi-divergence-cgf}
s D_{1+s}(P_{\tilde{\rho}}\|P_{\rho}) =  \phi((1+s) \tilde{\rho} - s \rho) - (1+s) \phi(\tilde{\rho}) + s \phi(\rho).
\end{align}
The following bounds on tail probabilities will be used later.
\begin{proposition}[\protect{\cite[Theorem A.2]{hayashi-watanabe:13b}}] \Label{theorem:one-shot-tail-converse-2}
For any $a > \mathsf{E}[Z]$, we have
\begin{align}
\lefteqn{ - \log P\{ Z \ge a \} } \nonumber \\
\le & \inf_{s > 0 \atop \tilde{\rho} \in \mathbb{R}, \sigma \ge 0 }
\frac{1}{s}
 \Bigg[ \phi((1+s)\tilde{\rho}) - (1+s)\phi(\tilde{\rho}) 
 \nonumber \\
& \hspace{5ex} - (1+s) \log \bigg(1- e^{ - [\sigma a - \phi(\tilde{\rho}+\sigma) + \phi(\tilde{\rho}) ] } \bigg) \Bigg] 
 \Label{eq:tail-converse-2-0c} \\
\le & \inf_{s > 0 \atop \tilde{\rho} > \rho(a) }
 \frac{1}{s} \Bigg[ \phi((1+s)\tilde{\rho}) - (1+s)\phi(\tilde{\rho}) 
 \nonumber \\
&  - (1+s) \log \bigg(1- e^{ - [
(\tilde{\rho} -\rho(a))
 a - \phi(\tilde{\rho}+\sigma) + \phi(\tilde{\rho}) ] } \bigg) \Bigg] 
 \Label{eq:tail-converse-2-0} .
 \end{align}
Similarly, for any $a < \mathsf{E}[Z]$, we  have
\begin{align}
\lefteqn{ - \log P\{ Z \le a \} } \nonumber \\
\le& \inf_{s > 0 \atop \tilde{\rho} \in \mathbb{R}, \sigma \ge 0 }
\frac{1}{s} \Bigg[ \phi((1+s)\tilde{\rho}) - (1+s)\phi(\tilde{\rho}) 
 \nonumber \\
& \hspace{5ex} - (1+s) \log \bigg(1- e^{ - [\sigma a - \phi(\tilde{\rho}+\sigma) + \phi(\tilde{\rho}) ] } \bigg) \Bigg] 
 \Label{eq:tail-converse-2-0-opposite-c} \\
\le & \inf_{s > 0 \atop \tilde{\rho} < \rho(a) }
\frac{1}{s} \Bigg[ \phi((1+s)\tilde{\rho}) - (1+s)\phi(\tilde{\rho}) 
 \nonumber \\
&  - (1+s) \log \bigg(1- e^{ - [(\rho(a)-\tilde{\rho})
a - \phi(\tilde{\rho}+\sigma) + \phi(\tilde{\rho}) ] } \bigg) \Bigg] 
 \Label{eq:tail-converse-2-0-opposite} .
 \end{align}
\end{proposition}

\subsection{Transition Matrix}
The discussion in this and the next subsections is a generalization of that for the lower conditional R\'{e}nyi entropy $H_{1+\theta}^{\downarrow,W}(X|Y)$
in the following sense.
In these subsections, 
the set ${\cal Z}$, and the functions $g$, $\tilde{g}$, and $\phi(\rho)$
are addressed.
The set ${\cal Z}$ is the generalization of ${\cal X} \times {\cal Y}$,
and the functions $g$, $\tilde{g}$, and $\phi(\rho)$
are the generalizations of $\log W -\log W_Y$,
$\log P_{X_1Y_1}-\log P_{Y_1}$,
and $-\theta H_{1+\theta}^{\downarrow,W}(X|Y)$, respectively.
Under this generalization, the same notation has the same meaning as for the lower conditional R\'{e}nyi entropy $H_{1+\theta}^{\downarrow,W}(X|Y)$.

Let $\{ W(z|z^\prime) \}_{(z,z^\prime) \in {\cal Z}^2}$ be 
an ergodic and irreducible transition matrix, and let $\tilde{P}$ be its
stationary distribution.
For a function $g:{\cal Z} \times {\cal Z} \to \mathbb{R}$,  let
\begin{eqnarray}
\mathsf{E}[g] := \sum_{z,z^\prime} \tilde{P}(z^\prime) W(z|z^\prime) g(z,z^\prime).
\end{eqnarray}
We also introduce the following tilted matrix:
\begin{eqnarray}
\tilde{W}_\rho(z|z^\prime) := W(z|z^\prime) e^{\rho g(z,z^\prime)}.
\end{eqnarray}
Let $\lambda_\rho$ be the Perron-Frobenius eigenvalue of $W_\rho$. 
Then, the CGF for $W$ with generator $g$ is defined by
\begin{eqnarray} \Label{eq:definition-transition-cgf}
\phi(\rho) := \log \lambda_\rho.
\end{eqnarray}

\begin{lemma} \Label{lemma:general-markov-cgf-strict-convexity}
The function $\phi(\rho)$ is a convex function of $\rho$, and it is strict convex iff. $\phi^{\prime\prime}(0) > 0$.
\end{lemma}
From Lemma \ref{lemma:general-markov-cgf-strict-convexity}, $\phi^\prime(\rho)$ is monotone increasing function.
Thus, we can define the inverse function $\rho(a)$ of $\phi^\prime(\rho)$ by
\begin{eqnarray} \Label{eq:definition-inverse-general-markov}
\phi^\prime(\rho(a)) = a.
\end{eqnarray}

\subsection{Markov Chain}

Let $\mathbf{Z} = \{ Z^n \}_{n=1}^\infty$ be the Markov chain induced by $W(z|z^\prime)$ and 
an initial distribution $P_{Z_1}$. For functions $g:{\cal Z} \times {\cal Z} \to \mathbb{R}$ 
and $\tilde{g}:{\cal Z} \to \mathbb{R}$, let $S_n := \sum_{i=2}^n g(Z_i,Z_{i-1}) + \tilde{g}(Z_1)$. Then,
the CGF for $S_n$ is given by
\begin{eqnarray}
\phi_n(\rho) := \log \mathsf{E}\left[ e^{\rho S_n} \right].
\end{eqnarray}
We will use the following finite evaluation for $\phi_n(\rho)$. 
\begin{lemma}[\protect{\cite[Lemma 5.1]{hayashi-watanabe:13b}}] \Label{lemma:finite-evaluation-of-cgf}
Let $v_\rho$ be the eigenvector of $\tilde{W}_\rho^T$ 
with respect to the Perron-Frobenius eigenvalue $\lambda_\rho$
such that $\min_{z} v_\rho(z) =1$. Let $w_\rho(z) := P_{Z_1}(z) e^{\rho \tilde{g}(z)}$. Then, we have
\begin{eqnarray}
(n-1) \phi(\rho) + \underline{\delta}_\phi(\rho) 
\le \phi_n(\rho) \le (n-1) \phi(\rho) + \overline{\delta}_\phi(\rho),
\end{eqnarray}
where 
\begin{eqnarray}
\overline{\delta}_\phi(\rho) &:=& \log \langle v_\rho | w_\rho \rangle, \\
\underline{\delta}_\phi(\rho) &:=& \log \langle v_\rho | w_\rho \rangle - \log \max_z v_\rho(z).
\end{eqnarray}
\end{lemma}

From this lemma, we have the following.
\begin{corollary}
For any initial distribution and $\rho \in \mathbb{R}$, we have
\begin{eqnarray}
\lim_{n\to \infty } \phi_n(\rho) = \phi(\rho).
\end{eqnarray}
\end{corollary}

The relation
\begin{eqnarray}
\lim_{n\to\infty} \frac{1}{n} \san{E}[S_n] = \phi^\prime(0) 
= \san{E}[g]
\end{eqnarray}
is well known.
Furthermore, we also have the following.
\begin{lemma} \Label{lemma:appendix-variance-limit}
For any initial distribution, we have
\begin{eqnarray}
\lim_{n \to \infty} \frac{1}{n} \mathrm{Var}\left[ S_n \right] = \phi^{\prime\prime}(0).
\end{eqnarray}
\end{lemma}

Finally, we also use the following bound on tail probabilities.
\begin{proposition}[\protect{\cite[Theorem 7.2]{hayashi-watanabe:13b}}] \Label{proposition:general-markov-tail-converse}
For any $a > \mathsf{E}[g]$, we have
\begin{align}
\lefteqn{ - \log P\{ S_n \ge a n \} } \nonumber \\
\le & \inf_{s > 0 \atop \tilde{\rho} > \rho(a)} 
\frac{1}{s}
 \Bigg[ (n-1) \big( \phi((1+s)\tilde{\rho}) - (1+s)\phi(\tilde{\rho}) \big) + \delta_1 
 \nonumber \\
&  - (1+s) \log\left( 1- e^{ (n-1) [(\tilde{\rho} - \rho(a)) a + \phi(\rho(a)) - \phi(\tilde{\rho})] + \delta_2} \right) \Bigg] , 
  \Label{eq:markov-tail-bound-1} 
\end{align}
where 
\begin{eqnarray}
\delta_1 &:=& \overline{\delta}_\phi((1+s)\tilde{\rho}) - (1+s) \underline{\delta}_\phi(\tilde{\rho}), \\
\delta_2 &:=&  (\tilde{\rho} - \rho(a)) a + \overline{\delta}_\phi(\rho(a)) - \underline{\delta}_\phi(\tilde{\rho}).
\end{eqnarray}
Similarly, for any $a < \mathsf{E}[g]$, we have
\begin{align}
\lefteqn{ - \log P\{ S_n \le a n \} } \nonumber \\
\le & \inf_{s > 0 \atop \tilde{\rho} < \rho(a)} 
\frac{1}{s} \Bigg[ (n-1) \Big( \phi((1+s)\tilde{\rho}) - (1+s)\phi(\tilde{\rho}) \Big) + \delta_1 
\nonumber \\ 
&  - (1+s) \log\left( 1- e^{ (n-1) [(\tilde{\rho} - \rho(a)) a + \phi(\rho(a)) - \phi(\tilde{\rho})] + \delta_2} \right) \Bigg] . 
  \Label{eq:markov-tail-bound-1-opposite} 
\end{align}
\end{proposition}


\section{Proof of Lemma \ref{lemma:finite-evaluation-min-entropy}}
\Label{proof-lemma:finite-evaluation-min-entropy}
We first prove the following lemma.
\begin{lemma} \Label{lemma:preparation-evaluation-min-entropy}
Suppose that $x_1 = x_n$. Then, we have
\begin{eqnarray}
\prod_{i=2}^n W(x_i|x_{i-1}) \le e^{- (n-1) H_\infty^W(X)}.
\end{eqnarray}
\end{lemma}
\begin{proof}
When cycle $c = \{(x_1,x_2),\ldots,(x_{n-1},x_n)\}$ is a Hamilton cycle, the statement is obvious 
from the definition of $H_\infty^W(X)$. Otherwise, there exists a Hamilton cycle
$c^\prime = \{(x_j,x_{j+1}),\ldots,(x_{k-1},x_k)\}$ in $c$. Then, we have
\begin{align}
&\prod_{i=2}^n W(x_i|x_{i-1}) \nonumber \\
 =& \prod_{(x^\prime,x) \in c \backslash c^\prime} W(x|x^\prime) \prod_{(x^\prime,x) \in c^\prime} W(x|x^\prime)  \nonumber\\
 \le& \prod_{(x^\prime,x) \in c \backslash c^\prime} W(x|x^\prime) e^{- (k-j) H_\infty^W(X)}.
\end{align}
Since $c \backslash c^\prime$ is also a cycle, by repeating this procedure, we have the statement of the lemma.
\end{proof}

We now go back to the proof of Lemma \ref{lemma:finite-evaluation-min-entropy}.
To prove the left hand side inequality of \eqref{eq:finite-evaluation-min-entropy}, we need to upper bound
$\max_{x^n} P_{X^n}(x^n)$. 

For a given $x^n$ satisfying the relation $x_1 \neq x_n$,
we chose 
an extension $x^m=(x_1, \ldots, x_m)$
of $x^n$ as follows.
(1) 
$x_m$ is chosen to be $x_1 $.
(2) 
The path $c=\{ (x_n,x_{n+1}),\ldots,(x_{m-1},x_m)\}$ from $x_n$ to $x_m$ is chosen as the Hamilton path   
$\argmax_{c \in {\cal C}_{x_n,x_1}}
\prod_{(x_a,x_b) \in \hat{c}} W(x_b|x_a)$.
Then, we have
\begin{align}
A P_{X^n}(x^n) \le & P_{X^m}(x^m)  
\stackrel{(a)}{\le}   \max_x P_{X_1}(x) e^{-(m-1) H_\infty^W(X)}  \nonumber\\
\le &  \max_x P_{X_1}(x) e^{-(n-1) H_\infty^W(X)},
\end{align}
where $(a)$ follows from Lemma \ref{lemma:preparation-evaluation-min-entropy}. 
For a given $x^n$ satisfying the relation $x_1 = x_n$,
Lemma \ref{lemma:preparation-evaluation-min-entropy}
implies that
\begin{align}
P_{X^n}(x^n) \le  \max_x P_{X_1}(x) e^{-(n-1) H_\infty^W(X)}.
\end{align}
Since $A \le 1$,
we have the left hand side inequality of \eqref{eq:finite-evaluation-min-entropy} in the both case.

To show the opposite inequality,
let $\tilde{x} = \argmax_x P_{X_1}(x)$.
Assume that $\tilde{x} \neq x^*$.
Then, 
let $x^m$ be the sequence such that it start with $\tilde{x}$, the first part
constitutes a Hamilton path 
$c_o=\argmax_{c \in {\cal C}_{\tilde{x},x^*}}
\prod_{(x_a,x_b) \in \hat{c}} W(x_b|x_a)$
and then the sequence corresponding to the cycle $c^*$ is repeated $\lceil (n-|c_o|) / |c^*| \rceil$ times. 
Then, we have
\begin{align}
\max_{x^n} P_{X^n}(x^n)
&\ge \max_{{x^m}'} P_{X^m}({x^m}') 
\ge P_{X^m}(x^m)  \nonumber \\
&\ge P_{X_1}(\tilde{x}) A e^{- \lceil (n-|c_o|) / |c^*| \rceil |c^*| H_\infty^W(X)}  \nonumber\\
&\ge P_{X_1}(\tilde{x}) A e^{- \{ (n-|c_o|) + |c^*| \} H_\infty^W(X)} 
 \nonumber \\
&\ge P_{X_1}(\tilde{x}) A e^{- \{ (n-1) + |c^*| \} H_\infty^W(X)}.
\Label{12-23-1}
\end{align}
Assume that $\tilde{x} = x^*$.
Then, we construct $x^m$ in the same way with omitting the first part.
So, we have 
\begin{align}
\max_{x^n} P_{X^n}(x^n)
&\ge \max_{{x^m}'} P_{X^m}({x^m}') 
\ge P_{X^m}(x^m)  \nonumber\\
&\ge P_{X_1}(\tilde{x}) 
e^{- \lceil n / |c^*| \rceil |c^*| H_\infty^W(X)}  \nonumber\\
&= P_{X_1}(\tilde{x}) 
e^{- \{ n + |c^*| \} H_\infty^W(X)} 
\Label{12-23-2}
\end{align}
Combining \eqref{12-23-1} and \eqref{12-23-2},
we have the right hand side inequality of \eqref{eq:finite-evaluation-min-entropy}.
\qed

\section{Proof of Lemma \ref{lemma:limit-of-renyi-markov-single-terminal}}
\Label{proof-lemma:limit-of-renyi-markov-single-terminal}

To prove \eqref{eq:limit-of-renyi-markov-single-terminal-min-entropy}, we use the limiting results
\eqref{eq:single-terminal-markov-renyi-entropy-asymptotic-1} and \eqref{eq:single-terminal-markov-renyi-entropy-asymptotic-4}.
More precisely, we have
\begin{align}
&\lim_{\theta \to \infty} H_{1+\theta}^W(X)
= \lim_{\theta \to \infty} \lim_{n\to\infty} \frac{1}{n} H_{1+\theta}(X^n)  \nonumber \\
=& \lim_{n \to \infty} \lim_{\theta \to \infty} \frac{1}{n} H_{1+\theta}(X^n) 
= \lim_{n \to \infty} \frac{1}{n} H_\infty(X^n) 
= H_\infty^W(X).
\end{align}
To complete the proof, we need to show that the order of the limits can be changed, 
which is justified if $\overline{\delta}(\theta) / \theta$ and $\underline{\delta}(\theta) / \theta$
are bounded. For this purpose, it suffices to show $w_\theta(x) \le M^{1+\theta}$
and $v_\theta(x) \le \tilde{M}^{1+\theta}$ for some constants $M, \tilde{M}$
because these relations imply that
\begin{align*}
& - \frac{1}{\theta} \log |{\cal X}| (M\tilde{M})^{1+\theta}
\le \frac{ \underline{\delta}(\theta) }{\theta}
\le \frac{ \overline{\delta}(\theta) }{\theta} \\
\le & \frac{ \underline{\delta}(\theta) }{\theta}
+ \frac{1}{\theta} \log \tilde{M}^{1+\theta}
\le \frac{1}{\theta} \log \tilde{M}^{1+\theta}.
\end{align*}
The former is obvious. To prove the latter, without loss of generality, 
we can assume that ${\cal X}=\{1,2, \ldots, |{\cal X}|\}$
and that $v_\theta(1) \ge \cdots \ge v_\theta(|{\cal X}|) = 1$.
Since $\tilde{W}_\theta^T$ is irreducible, we can fix an integer $m$ such that 
$(\tilde{W}_\theta^T)^m(|{\cal X}| | 1) > 0$.
Since $v_\theta$ is an eigenvector, we have
\begin{eqnarray}
\sum_{x^\prime} (\tilde{W}_\theta^T)^m(x|x^\prime) v_\theta(x^\prime) 
= (\lambda_\theta)^m v_\theta(x).
\Label{eq:eigen-equation}
\end{eqnarray}
On the other hand, we have
\begin{align}
& (\tilde{W}_\theta^T)^m(1|x^\prime)  \nonumber \\
= &
\sum_{x_1, x_2, \ldots, x_{m-1}}
\tilde{W}_\theta^T( 1 |x_{m-1})
\cdots
\tilde{W}_\theta^T( x_2 |x_1)
\tilde{W}_\theta^T( x_1 |x' ) \nonumber \\
\le & 
|{\cal X}|^{m-1} 
\left( 
\max_{x,\bar{x} } 
\tilde{W}_\theta^T( x| \bar{x} ) 
\right)^m 
\!\!=\!|{\cal X}|^{m-1} 
\left( \max_{x,\bar{x}} W(\bar{x} |x)^{1+\theta} \right)^{m} 
\nonumber \\
= &
|{\cal X}|^{m-1} 
\left( \max_{x,\bar{x}} W(\bar{x} |x) \right)^{m(1+\theta)} .
\Label{1-23-1}
\end{align}
Since there exists, at least, one sequence $x_1, x_2, \ldots, x_{m-1}$
such that 
$\tilde{W}_\theta^T( |{\cal X}| |x_{m-1})
\cdots
\tilde{W}_\theta^T( x_2 |x_1)
\tilde{W}_\theta^T( x_1 | 1 ) >0 $, we have
\begin{align}
&(\tilde{W}_\theta^T)^m(|{\cal X}| | 1) 
\nonumber\\
= &
\sum_{x_1, x_2, \ldots, x_{m-1}}
\tilde{W}_\theta^T( |{\cal X}| |x_{m-1})
\cdots
\tilde{W}_\theta^T( x_2 |x_1)
\tilde{W}_\theta^T( x_1 | 1 ) \nonumber\\
\ge & 
\left( \min_{x,\bar{x} \atop W(\bar{x}|x) > 0} 
\tilde{W}_\theta^T( x| \bar{x} ) 
\right)^{m}
=
\left( \min_{x,\bar{x} \atop W(\bar{x}|x) > 0} W(\bar{x}|x) \right)^{m(1+\theta)}.
\Label{1-23-2}
\end{align}
Thus, we have
\begin{align}
& v_\theta(1)
 =
\frac{(\lambda_\theta)^m v_\theta(1)}
{(\lambda_\theta)^m v_\theta(|{\cal X}|)} 
 \stackrel{(a)}{=} 
\frac{\sum_{x^\prime} (\tilde{W}_\theta^T)^m(1|x^\prime) v_\theta(x^\prime)}{\sum_{x^\prime} (\tilde{W}_\theta^T)^m(|{\cal X}| |x^\prime) v_\theta(x^\prime)} \nonumber\\
 \le & \frac{\sum_{x^\prime} (\tilde{W}_\theta^T)^m(1|x^\prime) v_\theta(x^\prime)}{ (\tilde{W}_\theta^T)^m(|{\cal X}| | 1) v_\theta(1)} 
 \le \sum_{x^\prime} \frac{(\tilde{W}_\theta^T)^m(1|x^\prime)}
{(\tilde{W}_\theta^T)^m(|{\cal X}||1)} \nonumber\\
 \stackrel{(b)}{\le} &
\sum_{x^\prime} 
\frac{ |{\cal X}|^{m-1} 
\left( \max_{x,\bar{x}} W(\bar{x} |x) \right)^{m}}
{\left( \min_{x,\bar{x} \atop W(\bar{x}|x) > 0} W(\bar{x}|x) 
\right)^{m(1+\theta)}} 
\nonumber\\
 = & 
|{\cal X}|^{m} 
\left(
\frac{\left( \max_{x,\bar{x}} W(\bar{x} |x) \right)^{m}}
{\left( \min_{x,\bar{x} \atop W(\bar{x}|x) > 0} W(\bar{x}|x) \right)^{m}} 
\right)^{1+\theta} \nonumber\\
\le &
\left(
\frac{|{\cal X}|^{m} 
\left( \max_{x,\bar{x}} W(\bar{x}|x) \right)^{m(1+\theta)}}
{\left( \min_{x,\bar{x} \atop W(\bar{x}|x) > 0} W(\bar{x}|x) \right)^{m}} 
\right)^{1+\theta},
 \Label{eq:proof-bounded-delta-1}
\end{align}
where 
$(a)$ and $(b)$ follow from 
\eqref{eq:eigen-equation}
and the pair of \eqref{1-23-1} and \eqref{1-23-2}, respectively.
Hence, we have the desired bound.
\qed

\section{Proof of Lemma \ref{lemma:extreme-cases-up-conditional-renyi-transition}}
\Label{appendix:lemma:extreme-cases-up-conditional-renyi-transition}
Since
$1 \le 
 \sum_x \left( \frac{W_{X|X',Y',Y}(x|x^\prime,y^\prime,y)}{T(y|y^\prime)} \right)^{1+\theta}
\le |{\cal X}|$,
we have
\begin{align}
& K_\theta(y|y^\prime) 
\nonumber\\
=& W_Y(y|y^\prime) T(y|y^\prime) 
 \Bigg( \sum_x \!\bigg( \!\frac{W_{X|X',Y',Y}(x|x^\prime,y^\prime,y)}{T(y|y^\prime)} \!\bigg)^{1+\theta} \!\Bigg)^{\frac{1}{1+\theta}} 
 \nonumber\\
 \to & W_Y(y|y^\prime) T(y|y^\prime)
\end{align}
as $\theta \to \infty$. 
Thus, by the continuity of eigenvalues with respect to the matrix, we have
$\kappa_\theta \to \kappa_\infty$, which implies \eqref{eq:extreme-cases-up-conditional-renyi-transition-min}. \qed

\section{Proof of Theorem \ref{theorem:asymptotic-down-conditional-renyi}}\Label{appendix:theorem:asymptotic-down-conditional-renyi}
To prove \eqref{eq:markov-min-entropy-down-asymptotic},
we note that $P_{X^n|Y^n}$ can be written as 
\begin{align}
& P_{X^n|Y^n}(x^n|y^n) \nonumber \\
= &P_{X_1|Y_1}(x_1|y_1) \prod_{i=2}^n 
W_{X|X',Y',Y}(x_i | x_{i-1},y_{i-1},y_i).
\end{align}
Thus, in a similar manner as the proof of Lemma \ref{lemma:finite-evaluation-min-entropy}, 
we can derive an upper bound and a lower bound on $H_\infty^{\downarrow}(X^n|Y^n)$,
from which we can derive \eqref{eq:markov-min-entropy-down-asymptotic}. 

On the other hand, to show \eqref{eq:markov-min-entropy-up-asymptotic},
we have
\begin{align}
& e^{- H_{\infty}^\uparrow(X^n|Y^n)} \nonumber \\
=& \sum_{y^n} P_{Y^n}(y^n) \max_{x^n } P_{X^n|Y^n}(x^n|y^n) \nonumber \\
=& P_{Y_1}(y_1) \max_{x_1} P_{X_1|Y_1}(x_1|y_1) \prod_{i=2}^n W_Y(y_i|y_{i-1}) T(y_i | y_{i-1}).
\end{align}
Thus, in a similar manner as the proof of Lemma \ref{lemma:multi-terminal-finite-evaluation-upper-conditional-renyi}
shown in \cite[Lemma 10]{HW14},
we can derive an upper bound and a lower bound on $H_{\infty}^\uparrow(X^n|Y^n)$, from which we can 
derive \eqref{eq:markov-min-entropy-up-asymptotic}.


%% file: Appendix-Single-Random-Number.tex
\section{Proof of Lemma \ref{lemma:strong-universal-bound-tail-probability}} 
\label{appendix:lemma:strong-universal-bound-tail-probability}

Let
\begin{eqnarray} 
\Omega = \left\{ x : \log \frac{1}{P_X(x)} \le a \right\}.
\end{eqnarray}
Then, for $\rho \le 1$, we have
\begin{align}
|\Omega| \le& \sum_{x\in \Omega} e^{(1-\rho)\left(a - \log\frac{1}{P_X(x)}\right)} \nonumber \\
\le & \sum_x P_X(x)^{1-\rho} e^{(1-\rho)a} 
= e^{(1-\rho)a + \phi(\rho;P)},
\end{align}
\MH{where $\phi(\rho;P)$ is defined in \eqref{eq:definition-single-shot-cgf}.}
Here, we set $\rho = \rho(a)$ and $a = a(R )$. 
Then, by noting \eqref{eq:definition-a-inverse-multi-markov}, we have
\begin{eqnarray}
|\Omega| \le e^R 
= M \nu.
\end{eqnarray}
Thus, by using Lemma \ref{lemma:strong-universal-bound}, we have \eqref{eq:strong-universal-bound}. \qed

\section{Proof of Theorem \ref{theorem:single-random-number-strong-universal-finite-markov-converse}}
\label{appendix:theorem:single-random-number-strong-universal-finite-markov-converse}

The proof proceed almost in a similar manner as the proof of Lemma \ref{lemma:strong-universal-bound-tail-probability}.
Let 
\begin{eqnarray}
\Omega = \left\{ x^n : \log \frac{1}{P_{X^n}(x^n)} \le a n \right\}.
\end{eqnarray}
Then, for any $\rho \le 1$, we have
\begin{eqnarray}
|\Omega| &\le& e^{(1-\rho)an  + \phi(\rho;P_{X^n})} \nonumber\\
&=& e^{(1+\theta) a n - \theta H_{1+\theta}(X^n)} \nonumber\\
&\le& e^{(1+\theta)an - (n-1) \theta H_{1+\theta}^W(X) - \underline{\delta}(\theta)},
\end{eqnarray}
where we changed variable as $\rho = - \theta$ and used Lemma \ref{lemma:mult-terminal-finite-evaluation-down-conditional-renyi}.
Here, we set $\theta = \theta(a)$ and $a = a(R )$. Then, by noting \eqref{eq:definition-a-inverse-multi-markov}, we have
\begin{align}
|\Omega| \le e^{(n-1) R + \left\{ (1+\theta(a(R ))) a(R ) - \underline{\delta}(\theta(a(R ))) \right\}} 
= \frac{M_n}{2}.
\end{align}
Thus, by using Lemma \ref{lemma:strong-universal-bound}, we have
\begin{eqnarray}
\overline{\Delta}(M_n) \ge \frac{1}{4} P_{X^n}\left\{ \log \frac{1}{P_{X^n}(x^n)} \le a(R ) n \right\}.
\end{eqnarray}
Finally, by using Proposition \ref{proposition:general-markov-tail-converse}, 
and changing the variable as $\tilde{\rho} = - \tilde{\theta}$, we have the assertion of the theorem. \qed

%% file: Appendix-Multi-Random-Number.tex
\section{Proof of Lemma \ref{lemma:strong-universal-bound-multi-tail}}
\label{appendix:lemma:strong-universal-bound-multi-tail}

Let
\begin{eqnarray}
\Omega_y = \left\{ x : \log \frac{P_Y^{(1+\theta)}(y)}{P_{XY}(x,y) }\le a \right\}.
\end{eqnarray}
Then, for any $\theta \ge -1$, we have
\begin{align}
|\Omega_y| 
\le& \sum_{x \in \Omega_y} e^{(1+\theta)\left(a - \log \frac{P_Y^{(1+\theta)}(y)}{P_{XY}(x,y)} \right)} 
\nonumber \\
\le& e^{(1+\theta)a} \sum_x \frac{P_{XY}(x,y)^{1+\theta}}{P_Y^{(1+\theta)}(y)^{1+\theta}}  
\nonumber \\
\stackrel{(a)}{=} & e^{(1+\theta)a} \sum_x \Bigg[ P_{XY}(x,y)^{1+\theta} 
\nonumber \\
& \hspace{10ex} \cdot
\frac{\left[ \sum_y \left( \sum_{x^\prime} P_{XY}(x^\prime,y)^{1+\theta} \right)^{\frac{1}{1+\theta}} \right]^{1+\theta}}{\sum_{x^{\prime\prime}} P_{XY}(x^{\prime\prime},y)^{1+\theta}} 
\Bigg]\nonumber \\
\stackrel{(b)}{=} & e^{(1+\theta)a - \theta H_{1+\theta}^\uparrow(X|Y)},
\end{align}
where $(a)$and $(b)$ follow from \eqref{eq:single-shot-optimal-conditioning-distribution} and \eqref{11-14-6}, respectively.
Thus, by setting $\theta = \theta(a)$ and $a = a(R )$, and by noting \eqref{eq:definition-a-inverse-Gallager-one-shot}, we have
\begin{eqnarray}
|\Omega_y| \le e^R 
 = M \nu.
\end{eqnarray}
Thus, from Lemma \ref{lemma:strong-universal-bound-multi},
we have \eqref{eq:pa-strong-hash-one-shot-multi}.
\qed

\section{Proof of Theorem \ref{theorem:multi-random-finite-markov-assumption-2-converse}}
\label{appendix:theorem:multi-random-finite-markov-assumption-2-converse}

The proof proceed in a similar manner as the proof of Lemma \ref{lemma:strong-universal-bound-multi-tail}.
Let 
\begin{eqnarray}
\Omega_{y^n} = \left\{ x^n : \log \frac{P_{Y^n}^{(1+\theta)}(y^n)}{P_{X^nY^n}(x^n,y^n)} \le an  \right\}.
\end{eqnarray}
Then, for any $\theta \ge -1$, we have (cf.~the proof of Lemma \ref{lemma:strong-universal-bound-multi-tail})
\begin{eqnarray}
|\Omega_{y^n}| &\le& e^{(1+\theta) an - \theta H^\uparrow_{1+\theta}(X^n|Y^n)} \nonumber \\
&\le& e^{(1+\theta)an  - (n-1) \theta H_{1+\theta}^{\uparrow,W}(X|Y) - (1+\theta)\underline{\xi}(\theta)},
\end{eqnarray}
where we 
used Lemma \ref{lemma:multi-terminal-finite-evaluation-upper-conditional-renyi} in the inequality.
Here, we set $\theta = \theta(a)$ and $a = a(R )$. Then, by noting \eqref{eq:definition-a-inverse-markov-optimal-Q}, we have
\begin{eqnarray}
|\Omega_{y^n}| &\le& e^{(n-1) R +\left\{ (1+\theta(a(R ))) (a(R ) -  \underline{\xi}(\theta(a(R )))) \right\}} \nonumber  \\
&=& \frac{M_n}{2}.
\end{eqnarray}
Thus, by using Lemma \ref{lemma:strong-universal-bound-multi}, we have
\begin{eqnarray}
\overline{\Delta}(M_n) \ge \frac{1}{4} 
P_{X^nY^n}\left\{ \log \frac{P_{Y^n}^{(1+\theta(a(R )))}(y^n)}{P_{X^nY^n}(x^n,y^n)} \le a(R ) n  \right\}.
\end{eqnarray}
Here, we denote the CGF with $Z=\log \frac{Q_Y(Y)}{P_{XY}(X,Y)}$
by $\phi(\theta;P_{XY}|Q_Y) $.
Then, we have
\begin{align}
\theta H_{1+\theta}^{\uparrow}(P_{XY}|Q_Y) 
= - \phi (-\theta; P_{XY}| P_{Y}^{(1+\theta (a(R )) )}).
\Label{11-14-8}
\end{align}

Applying 
\eqref{eq:tail-converse-2-0-opposite-c} of Proposition \ref{theorem:one-shot-tail-converse-2}
to the random variable
$Z=\log \frac{P_{Y}^{(1+\theta(a(R )))}(Y)}{P_{XY}(X,Y)}$, 
we have
\begin{align*}
& - \log P_{X^nY^n}\left\{ \log \frac{P_{Y^n}^{(1+\theta(a(R )))}(y^n)}{P_{X^nY^n}(x^n,y^n)} \le a(R ) n  \right\} \\
\le &
\inf_{s > 0 \atop \tilde{\rho} \in \mathbb{R}, \sigma \ge 0} \frac{1}{s}
 \Biggl[ 
\phi((1+s)\tilde{\rho};
P_{X^n Y^n}|P_{Y^n}^{(1+\theta(a(R )))})
\\
&
-\! (\!1\!+\!s\!)
\phi(\tilde{\rho};P_{X^n Y^n}|P_{Y^n}^{(1+\theta(a(R )))}) 
-\! (\!1\!+\!s\!) \log \left(1- e^{ 
C_5
} \right) \Biggr],
\end{align*}
where
\begin{align*}
C_5:=&
- \Big[\sigma a - 
\phi(\tilde{\rho}+\sigma;P_{X^n Y^n}|P_{Y^n}^{(1+\theta(a(R )))})
\\
&\hspace{21ex}+ 
\phi(\tilde{\rho};P_{X^n Y^n}|P_{Y^n}^{(1+\theta(a(R )))})
\Big] .
\end{align*}
We choose the variable $\tilde{\rho} $ to be $- \tilde{\theta}$ and 
restrict the variable 
$\sigma $ to be $\tilde{\theta}- \theta(a(R ) )$
with the condition 
$\tilde{\theta} > \theta(a(R ) )$.
Then, we use \eqref{11-14-8} and Lemma \ref{lemma:multi-terminal-finite-evaluation-two-parameter-conditional-renyi}. 
Hence, we have the assertion of theorem.
\qed